\documentclass[11pt]{article}
\usepackage[hmargin=1in,vmargin=1in]{geometry}
\usepackage{hchang}
\usepackage{thm-restate}
\usepackage{cleveref}
\usepackage{comment}
\usepackage{tcolorbox}
\usepackage{textpos}

\newtheorem{theorem}{Theorem}[section]
\newtheorem{lemma}[theorem]{Lemma}
\newtheorem{claim}[theorem]{Claim}
\newtheorem{corollary}[theorem]{Corollary}
\newtheorem{definition}[theorem]{Definition}
\newtheorem{observation}[theorem]{Observation}
\newtheorem{remark}[theorem]{Remark}
\newtheorem{conj}[theorem]{Conjecture}

\newcommand{\dist}{\ensuremath{\delta}}
\newcommand{\len}[1]{{\norm{#1}}}

\newcommand{\dom}{\ensuremath{\mathrm{dom}}}
\newcommand{\initdom}{\ensuremath{\mathrm{dom_0}}}

\newcommand{\bdry}{{\partial\!}}
\newcommand{\eps}{\e}
\newcommand{\bag}{\mathsf{bag}}
\newcommand{\out}[1]{\ensuremath{\mathrm{out}(#1)}}
\newcommand{\basev}{\ensuremath{\hat v}}

\def\cA{\ensuremath{\mathcal{A}}}
\def\cN{\ensuremath{\mathcal{N}}}
\def\cX{\ensuremath{\mathcal{X}}}
\def\cS{\ensuremath{\mathcal{S}}}
\def\cT{\ensuremath{\mathcal{T}}}
\def\cC{\ensuremath{\mathcal{C}}}
\def\cF{\ensuremath{\mathcal{F}}}

\def\Oh{\ensuremath{O}}

\def\bbC{\ensuremath{\mathbb{C}}}
\def\bbF{\ensuremath{\mathbb{F}}}
\def\bbR{\ensuremath{\mathbb{R}}}

\def\frakC{\ensuremath{\mathfrak{C}}}
\def\frakF{\ensuremath{\mathfrak{F}}}

\newcommand{\tw}{\mathrm{tw}}

\begin{document}
\begin{titlepage}
\title{Embedding Planar Graphs into Graphs of Treewidth $O(\log^{3} n)$%
\thanks{
  Ma.P.\ is supported by Polish National Science Centre SONATA BIS-12 grant number 2022/46/E/ST6/00143.
  This work is a part of project BOBR (Mi.\ P\/.) that has received funding from the European Research Council (ERC) under the European Union's Horizon 2020 research and innovation programme (grant agreement No.~948057).
  Hung Le was supported by the NSF CAREER Award No. CCF-223728, an NSF Grant No. CCF-2121952, and a Google Research Scholar Award.}}

\author{
    Hsien-Chih Chang\thanks{Dartmouth College.  Email: {\tt hsien-chih.chang@dartmouth.edu}.} \and
    Vincent Cohen-Addad\thanks{Google Research, France. Email: {\tt cohenaddad@google.com}.}\and
    Jonathan Conroy\thanks{Dartmouth College.  Email: {\tt jonathan.conroy.gr@dartmouth.edu}.} \and
    Hung Le\thanks{University of Massachusetts Amherst. Email: {\tt hungle@cs.umass.edu}.}\and
      Marcin Pilipczuk%
  \thanks{Institute of Informatics, University of Warsaw, Poland.  Email: {\tt m.pilipczuk@uw.edu.pl}.}
    \and
    Micha\l{} Pilipczuk%
  \thanks{Institute of Informatics, University of Warsaw, Poland.  Email: {\tt michal.pilipczuk@mimuw.edu.pl}.}
}

\date{}

\maketitle
\thispagestyle{empty}

\begin{abstract}
Cohen-Addad, Le, Pilipczuk, and Pilipczuk \cite{CLPP23} recently constructed a stochastic embedding with expected $1+\eps$ distortion of $n$-vertex planar graphs (with polynomial aspect ratio) into graphs of treewidth $O(\eps^{-1}\log^{13} n)$.  Their embedding is the first to achieve polylogarithmic treewidth.  However, there remains a large gap between the treewidth of their embedding and the treewidth lower bound of $\Omega(\log n)$ shown by Carroll and Goel~\cite{CG04}.  In this work, we substantially narrow the gap by constructing a stochastic embedding with treewidth  $O(\eps^{-1}\log^{3} n)$.

We obtain our embedding by improving various steps in the CLPP construction.  First, we streamline their embedding construction by showing that one can construct a low-treewidth embedding for any graph from (i) a stochastic hierarchy of clusters and (ii) a stochastic balanced cut. We shave off some logarithmic factors in this step by using a single hierarchy of clusters.  Next, we construct a stochastic hierarchy of clusters with optimal separating probability and hop bound based on shortcut partition~\cite{CCLMST23a,CCLMST24}.  Finally, we construct a stochastic balanced cut with an improved trade-off between the cut size and the number of cuts. This is done by a new analysis of the contraction sequence introduced by~\cite{CLPP23}; our analysis gives an optimal treewidth bound for graphs admitting a contraction sequence. 
\end{abstract}

\begin{textblock}{20}(-1.75, 3.8)
\includegraphics[width=40px]{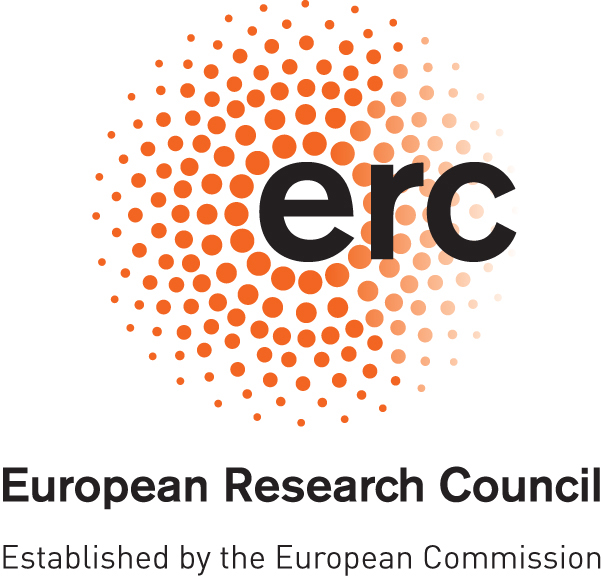}%
\end{textblock}
\begin{textblock}{20}(-1.75, 4.8)
\includegraphics[width=40px]{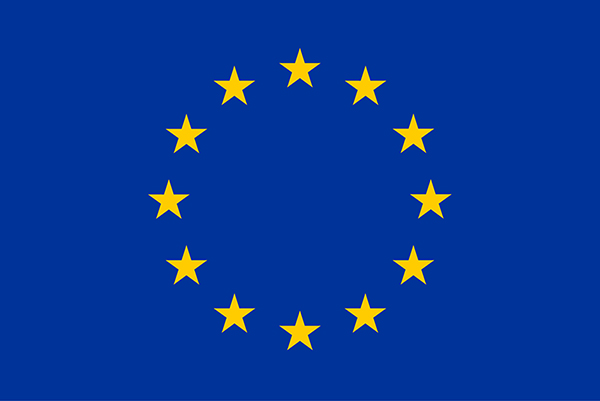}%
\end{textblock}

\end{titlepage}

\section{Introduction}

A general algorithmic technique to design approximation algorithms on graphs is \emph{metric embedding}: given an edge-weighted graph $G$ and a problem $\Pi$, embed $G$ into a structurally simpler edge-weighted graph $H$ with a small \emph{distortion} and solve $\Pi$ on $H$ to get an approximate solution for $\Pi$ in $G$. 
The distortion often decides the approximation ratio, while the running time to solve $\Pi$ depends on the simplicity of the structure of $H$. 
An important line of work initiated by Bartal~\cite{Bartal96,Bartal98} culminated in the tree embedding theorem by Fakcharoenphol, Rao, and Talwar~\cite{FRT04}, who showed that any $n$-vertex graph can be stochastically embedded into a tree with distortion $O(\log n)$. 
This result gives polynomial-time $O(\log n)$-approximation algorithms for a host of optimization problems in graphs, for example, buy-at-bulk network design, metric labeling, minimum cost communication network; see Section 4 in~\cite{FRT04} for a more comprehensive list of applications. 

However, the $O(\log n)$ distortion for embedding into trees is tight~\cite{Bartal96} even for planar graphs~\cite{CG04,CJVL08}; therefore, techniques using embedding into trees cannot give better than $O(\log n)$ approximation ratio. To achieve a better distortion, one has to not only restrict the structure of the input graph $G$ but also extend the class of host graphs $H$ to be broader than trees. 
Here, we consider a well-studied and natural setting: $G$ is planar or minor-free, and $H$ has a small \emph{treewidth}. 
Until recently, most results in this direction were negative: 
(1) the expected distortion remains $\Omega(\log n)$ even if $H$ has \emph{bounded treewidth}~\cite{CJVL08}; 
(2) if the distortion is a constant $c\geq 1$, then the treewidth must be $\Omega(c^{-1}\cdot \log n)$~\cite{CG04}. 
If one looks for a deterministic embedding, the treewidth of $H$ must be $\Omega(\sqrt{n})$ for a constant distortion~\cite{CG04}. 
However, as observed by Cohen-Addad, Le, Pilipczuk, and Pilipczuk~\cite{CLPP23}, these lower bounds do not rule out the following conjecture:

\begin{conj}
\label{conj:planar-embed} 
For every $\eps \in (0,1)$, any edge-weighted $n$-vertex planar graph can be embedded with expected distortion $1+\eps$ into a distribution of graphs with treewidth $O(\eps^{-1} \cdot \log n)$.
\end{conj}

\Cref{conj:planar-embed}, if true, will imply, among other things, the first PTAS (polynomial-time approximation scheme) for the \textsc{$k$-MST} ($k$-minimum spanning tree) problem in planar graphs: find a spanning tree of minimum weight that contains at least $k$ vertices a graph. 
After a long line of works~\cite{RSMRR96,AABV95,ARV99,Garg96,AR98,AK06}, Garg~\cite{Garg05} showed that this problem admits a $2$-approximation in polynomial time in general graphs. 
In planar graphs, it has been a long-standing open problem to design a PTAS for $k$-MST; only a QPTAS (quasi-polynomial time approximation scheme) was known by a very recent work of Cohen-Addad~\cite{CA22}. 

Cohen-Addad, Le, Pilipczuk, and Pilipczuk~\cite{CLPP23} gave strong evidence supporting \Cref{conj:planar-embed}: they constructed the first embedding with polylogarithmic treewidth. The precise treewidth bound is 
\[
O(\eps^{-1} \cdot \log^{6}(n\Phi)\cdot \log^5(n\Phi/\eps) \cdot \log^2 n) 
\] 
where $\EMPH{$\Phi$} \coloneqq \frac{\max_{u,v}d_G(u,v)}{\min_{u\not= v}d_G(u,v)}$ is the \EMPH{aspect ratio}. 
Specifically, for graphs with polynomial aspect ratios $\Phi = n^{O(1)}$, the treewidth they obtained is roughly $O(\eps^{-1}\log^{13} n)$. (In many interesting applications, including those considered in~\cite{CLPP23}, one can assume that the aspect ratio is polynomial.) Clearly there is a big gap between the treewidth upper bound achieved by \cite{CLPP23} and the treewidth bound in \Cref{conj:planar-embed}. In this work, we substantially narrow the gap, and specifically, for the case of polynomial aspect ratio, we achieve treewidth bound $O(\eps^{-1}\log^3 n)$. 
Our result extends to any apex-minor-free graphs, which include planar graphs as a subclass.

\begin{theorem}\label{thm:main} 
Let $G$ be an $n$-vertex apex-minor-free graph with aspect ratio $\Phi$. For any parameter $\eps \in (0,1)$, $G$ can be stochastically embedded with expected distortion $1+\eps$ into a distribution of graphs of treewidth $O(\eps^{-1}\cdot \log^2 \Phi \cdot \log(n\Phi))$.
\end{theorem}
Our \Cref{thm:main} immediately leads to faster QPTAS  for several problems on planar graphs, including capacitated vehicle routing, capacitated $k$-median, and capacitated/uncapacitated facility location.

\paragraph{A brief overview of the CLPP embedding.} 
Our construction builds on ideas of~\cite{CLPP23}, which we briefly review here. 
For simplicity of exposition, we assume that the aspect ratio $\Phi = n^{O(1)}$ and hence $\log\Phi = O(\log n)$. 
The key idea in~\cite{CLPP23} is the construction of a $(\tau,\psi)$-stochastic balanced cut $\frakF$. Roughly speaking, a cut $\cF$ in $\frakF$ is a (small) set of at most $\tau$ clusters taken from a (stochastic) hierarchy $\bbC$ of clusters with exponentially decreasing diameter, called a \EMPH{clustering chain}. 
(This notion of cut is non-standard.)  
Intuitively a cut $\cF$ is \emph{balanced} if $G\setminus \cF$, the graph obtained by removing from $G$ all the vertices in the clusters in $\cF$, has connected components of size at most $2n/3$. 
Loosely speaking, a $(\tau,\psi)$-stochastic balanced cut $\frakF$ is a distribution of balanced cuts of size at most $\tau$ such that every cluster in $\bbC$ is sampled with at most $1/\psi$ probability.  \cite{CLPP23} showed that for apex-minor-free graphs (and hence for planar graphs), one can construct a $(\tau,\psi)$-stochastic balanced cut with $\tau = \psi\cdot\log^6 n$. 
(Parameter $\tau$ will dictate the treewidth of the embedding, while $\psi$ will dictate the distortion.)

The embedding algorithm is then rather simple: sample a cut $\cF$ from $\frakF$, recurse on each cluster $C\in \cF$ and each connected component of $G\setminus \cF$, and then form the embedding of $G$ from the embeddings in the recursions by (i) taking a single representative vertex per cluster in $\cF$ and (ii) connecting each representative to every other vertex of the graph.  

The depth of the recursion is shown to be $L = O(\log^2 n)$ by the fact that clusters in $\cF$ have a diameter at most half of the diameter of $G$, and the cut is balanced. 
(A subtle point is that connected components of $G$ might have a diameter bigger than $G$, and hence, the recursion depth is much larger than $\log n$.) 
Every time a pair $(u,v)$ is cut by a cluster $C$ of at level $i$ of $\bbC$ (of diameter $2^i$) in the balanced cut $\cF$, the additive distortion is $O(2^i)$ since their shortest path $\pi_G(u,v)$ is now rerouted (by the embedding) through the representative of $C$. 
(A pair $(u,v)$ is cut by a cluster $C$ if $\pi_G(u,v)$ intersects $C$.) 
By constructing $\bbC$ carefully, in expectation, one could show that the number of clusters at level $i$  intersecting $\pi(u,v)$ is about $O(\log^2 n)\cdot {\delta_G(u,v)}/{2^i}$, and hence the \emph{expected additive distortion at one level of the recursion} is:

\begin{equation}\label{eq:dist-CLPP-1level}
   O(2^i)\cdot \frac{1}{\psi}\cdot O(\log^2 n) \cdot \frac{\delta_G(u,v)}{2^i} 
   = \frac{O(\log^2 n)\cdot \delta_G(u,v)}{\psi}~,
\end{equation}
where $\frac{1}{\psi}\cdot O(\log^2 n)\cdot \frac{\delta_G(u,v)}{2^i}$ is the probability that a cluster intersecting $\pi(u,v)$ is chosen. 

Over $L$ levels of the recursion, the expected additive distortion is roughly $L \cdot O(\log^2 n)\cdot \delta_G(u,v)/\psi$. 
A technical subtlety is that a single edge $(x,y)$ could be cut by up to $O(\log n)$ clusters in different levels of $\bbC$ since the depth of $\bbC$ is $O(\log n)$; this introduces another $O(\log n)$ factor in the distortion. 
(Indeed, in this paper we suffer from the same extra log factor which is inherent in the framework.) 
Therefore, the final expected additive distortion is:
\begin{equation}\label{eq:dist-CLPP-all}
    \frac{L \cdot O(\log^3 n) \cdot \delta_G(u,v)}{\psi} = \frac{O(\log^5 n) \cdot \delta_G(u,v)}{\psi}
\end{equation}
To get an expected additive distortion of $\eps \cdot \delta_G(u,v)$ so to imply an expected \emph{multiplicative} distortion $1+\eps$, one sets $\psi \coloneqq \frac{\log^5(n)}{\eps}$, giving $\tau = \log^6(n) \cdot \psi = \eps^{-1}\log^{11} n$. 
Since every sampled cut has size at most $\tau$, and the recursion depth is $L$, the final treewidth of the embedding is $O(\tau \cdot L) = O(\eps^{-1} \log^{13} n)$. 

One could probably improve one log factor in the analysis of the algorithm by~\cite{CLPP23} by showing that the number of clusters at level $i$ of $\bbC$ intersecting $\pi(u,v)$ is instead $O(\log n)\cdot \frac{\delta_G(u,v)}{2^i}$, which is tight for their algorithm. 
By redoing the analysis, one could see that this log factor improvement shaves two log factors in the relationship between $\psi$ and $\tau$, giving $\tau = \psi \cdot \log^4 n$. 
All things together imply an embedding with treewidth $O(\eps^{-1} \log^{10} n)$. 
Significantly improving over the $O(\log^{10} n)$ treewidth bound seems to require new ideas for every step in the framework.

\paragraph{Our ideas.} 
We propose three new ideas on every component of the CLPP embedding.
First, we streamline the framework, showing that one can construct a low-treewidth embedding from two objects: 
(i) a (stochastic) clustering chain $\bbC$ with \emph{$\beta$-separating property}: the probability of cutting an edge $e$ at level $i$ of $\bbC$ is $\beta \cdot \len{e}/2^i$ where $2^i$ is the diameter of clusters at level~$i$, and 
(ii) a $(\tau,\psi)$-stochastic balanced cut. 
In this streamlined framework the treewidth of the embedding is $O(\tau + \log n)$ while the expected additive distortion is $O(\beta\log^2 n/\psi) \cdot \delta_G(u,v)$, 
compared to the original treewidth bound $O(\tau \log^2 n)$ and additive distortion $O(\log^5 n/\psi) \cdot \delta_G(u,v)$.
(See \Cref{thm:embedding}; notice this theorem holds even for general graphs.)
Our key idea here is to use a single clustering chain $\bbC$ throughout the recursion instead of recomputing the clustering chain for every connected component of $G\setminus \cF$ as done by \cite{CLPP23}. At every step of the recursion, the subgraph on which we recurse is a cluster of $\bbC$, and hence, we are guaranteed to make progress in \emph{every level of the recursion}: either the diameter reduced by a constant factor or the number of important vertices (those that the recursion has to handled) reduced by a constant factor while the diameter remains unchanged. Therefore, the recursion depth is $O(\log n)$. 

Next, we show how to construct a (stochastic) $\beta$-separating clustering chain $\bbC$ with $\beta = O(1)$ for any minor-free graphs, see \Cref{thm:beta-separating}. The construction of  \cite{CLPP23} can be seen as constructing a $\beta$-separating clustering chain $\bbC$ with $\beta = O(\log^2 n)$.  (See Lemma 3.6 in~\cite{CLPP23}, arXiv version.) 
Their construction does not exploit minor-closed property and holds for general graphs. 
As we remarked above, one could probably improve their analysis to obtain $\beta = O(\log n)$, and this is the best for general graphs. We achieve constant $\beta$ by exploiting minor-closed property. 
Although the same $\beta$ can be achieved by the KPR decomposition (Klein-Plotkin-Rao~\cite{KPR93}), for the purpose of constructing balanced cuts in the next step, we need an additional guarantee that the clustering chain $\bbC$ has \emph{bounded-hop property}. 
Roughly speaking, $\bbC$ has a bounded-hop property if, for every cluster $C$ in $\bbC$, the \emph{cluster contraction graph}, obtained from $G[C]$ by contracting every child cluster of $C$ into a single vertex, has bounded hop-diameter.  (See \Cref{def:beta-sep} for a precise definition.)
KPR decomposition (and its variants~\cite{AGMW10,AGGNT14,Fil19padded}) does not guarantee the bounded-hop property --- this is also why CLPP embedding cannot use KPR in their construction. 
Here, we build on the shortcut partition in a white-box fashion for minor-free graphs recently introduced by~\cite{CCLMST24} to achieve the constant hop bound (and constant $\beta$). 
The basic idea is to sample a random radius in the ball carving process of~\cite{CCLMST24}. The analysis of~\cite{CCLMST24} already provides a constant hop bound; our job is to show that for every edge, not many random balls can potentially ``reach'' it, giving a constant $\beta$ in the separating probability. Clearly, constant $\beta$ and constant hop bound are the best one can hope for.  
Note that the original shortcut partition construction is deterministic, whereas our $\beta$-separating clustering chain is necessarily stochastic; 
analyzing such variant requires us to prove certain new property about the shortcut partition that has not been established before.  
For comparison, recently Filtser~\cite{Fil24} constructed a padded-decomposition based off the buffered cop decomposition~\cite{CCLMST24}, but they emphasized that only a weak diameter guarantee can be achieved.  
Here we construct a strong-diameter padded decomposition with constant hop guarantee (\Cref{lem:random-shortcut}) which is needed for the clustering chain.
It turns out the proof is rather technical.%
\footnote{For the readers who are familiar with the buffered cop decomposition: Roughly speaking, we need to bound the number of supernodes $X$ that are at most distance $O(\Delta)$ away from a fixed vertex $\hat{v}$ within the domain of $X$.  However, it is not with respect to the final domain $\dom(X)$, but the \emph{initial} domain $\dom_0(X)$ at the time when $X$ was initiated.  As a result, many more supernodes could be distance $O(\Delta)$ away.  See Section~\ref{SS:cop-threateners} and \Cref{lem:bdd-threatener-dom0} for a precise statement.}

Finally, we construct a $(\tau,\psi)$-stochastic balanced cut for apex-minor-free graphs with $\tau = O(\log n)\cdot \psi$. 
Recall that in~\cite{CLPP23}, $\tau = O(\log^6 n)\cdot \psi$ (see Lemma 4.2 in \cite{CLPP23}, arXiv version). 
By plugging the clustering chain with constant hop bound above into the construction of balanced cuts by~\cite{CLPP23}, we already remove a factor of $\log^4 n$, giving $\tau = O(\log^2 n)\cdot \psi$. 
One key idea to remove another $\log n$ factor is a new analysis of \emph{contraction sequence}: roughly speaking, an \emph{$(a,b,c)$-contraction sequence} represents a contraction of a graph $G$ into a single vertex through $b$ rounds, where in each round, we contract a certain number of vertex-disjoint subgraphs of radius at most $c$; the total number of subgraphs we contract in $b$ rounds is $a$. 
\cite{CLPP23} showed that if an apex-minor-free graph $G$ admits an $(a,b,c)$-contraction sequence, then $\tw(G) = O(bc\sqrt{a})$. 
Their basic idea is to argue that there exists a set $S$ of $a$ vertices such that every other vertex is within hop distance $O(bc)$ from the set and then apply known techniques to get treewidth bound $O(bc\sqrt{a})$. 
Here, we improve the treewidth bound by a factor of $\sqrt{b}$, which we show is optimal.  
First, by a simple reduction, we can assume that $c = 1$. Second, which is also the bulk of the technical details, we show the following structural property: Let $Z$ be any maximal set of vertices in $G$ such that their pairwise (hop) distance in $G$ is $\Omega(b)$, then $|Z| \leq \frac{a}{b}$. 
Thus, by standard tools, we can conclude that $\tw(G) = O(\sqrt{Z}b) = O(\sqrt{ab})$ (for the case $c=1$). 
This, in turn, leads to our improved bound $\tau = O(\log n)\cdot \psi$. 
Proving this structural property requires an in-depth analysis of the contraction sequence, which might be of independent interest. 

By incorporating all the ideas above, we obtain an embedding with treewidth $O(\eps^{-1}\log^3 n)$ as stated in \Cref{thm:main}. 
It seems to us that, to further improve any logarithmic factor, one has to significantly deviate from the framework of~\cite{CLPP23}. 
Recall that our expected distortion is $1 + O(\beta\log^2 n/\psi)$. 
One log factor comes from the fact that we must handle distances in $O(\log n)$ scales; this factor is unavoidable and also reflects in the treewidth lower bound $\Omega(\log n)$ by~\cite{CG04}. 
Another log factor comes from the technical subtlety mentioned in the overview of~\cite{CLPP23}: a single edge $(x, y)$ could be cut by up to $\Omega(\log n)$ clusters in different levels of $\bbC$. 
This is inherent in any clustering chain: there is always an edge that is cut by  $\Omega(\log n)$ clusters.  
One last log factor comes from the bound $\tau = O(\log n)\cdot \psi$ in $(\tau,\psi)$-balanced cut. 
Removing this log factor seems to require a completely different technique to construct balanced cuts since the core component of our construction, which is the $(a,b,c)$-contraction sequence, is the best possible. 
While we are closer to \Cref{conj:planar-embed}, its resolution remains elusive.

\section{Technical Overview}

In this section, we provide key definitions from the discussion above, as well as a high-level overview of our embedding construction.
A \EMPH{clustering} of an edge-weighted graph $G = (V,E,\len{\cdot})$  is a partition of the vertex set $V$ where each \EMPH{cluster} in the clustering is a connected subset of vertices. 
The strong (resp.\ weak) \EMPH{diameter} of a cluster $C$ is $\max_{x,y\in C} d_{G[C]}(x,y)$ (resp.\ $\max_{x,y\in C} d_G(x,y)$). 
We say that an edge $e\in E(G)$ is \EMPH{cut}
by a clustering $\cC$ if the endpoints of the edge belong to two different clusters in $\cC$. 
We now recall the notion of \emph{clustering chain}~\cite{CLPP23}.

\begin{definition}[Clustering Chain]
\label{def:clusterin-chain} 
A \EMPH{clustering chain} of a connected graph $G$ is a sequence $\mathbb{C} = (\cC_0, \cC_1, \ldots, \cC_k)$ of ever coarser clusterings $G$ such that:
\begin{itemize}
 \item $\cC_k$ contains a single cluster $V$.
 \item  $\cC_0$ consists of vertices in $V$ each as a singleton cluster.
 \item for every $i$, $\cC_i$ \emph{refines} $\cC_{i+1}$: every cluster of $\cC_i$ is contained within a unique cluster of~$\cC_{i+1}$. 
 \item  every cluster $C\in \cC_{i}$ has \textul{strong} diameter at most $2^i$.
\end{itemize}    
\end{definition}

A clustering $\mathbb{C}$ naturally forms a hierarchy where the root of the hierarchy corresponds to the (only) cluster in $\cC_k$, the leaves are singleton clusters in $\cC_0$, scale-$i$ clusters are clusters in $\cC_i$, and the parent of a cluster $C\in \cC_i$ is the cluster $C' \in \cC_{i+1}$ such that $C\subseteq C'$.  
See \Cref{fig:clustering-chain}.
We denote by \EMPH{$\cC_{i}[C']$} the set of scale-$i$ clusters in $\cC_{i}$ that make up  the scale-$(i+1)$ cluster $C'$.
We often abuse the language and say that a cluster $C$ of $G$ is \EMPH{in $\bbC$} if $C$ belongs to the clustering at some scale of $\bbC$. 
For a cluster $C$ in $\bbC$, we denoted by \EMPH{$\bbC_{\downarrow C}$} a clustering chain of $C$ obtained by taking all the clusters in $\bbC$ that are subsets of $C$ (including $C$). 
Intuitively, $\bbC_{\downarrow C}$ is a sub-hierarchy of $\bbC$ rooted at $C$. 

\begin{figure}[t]
\centering
\includegraphics[width=1.0\linewidth]{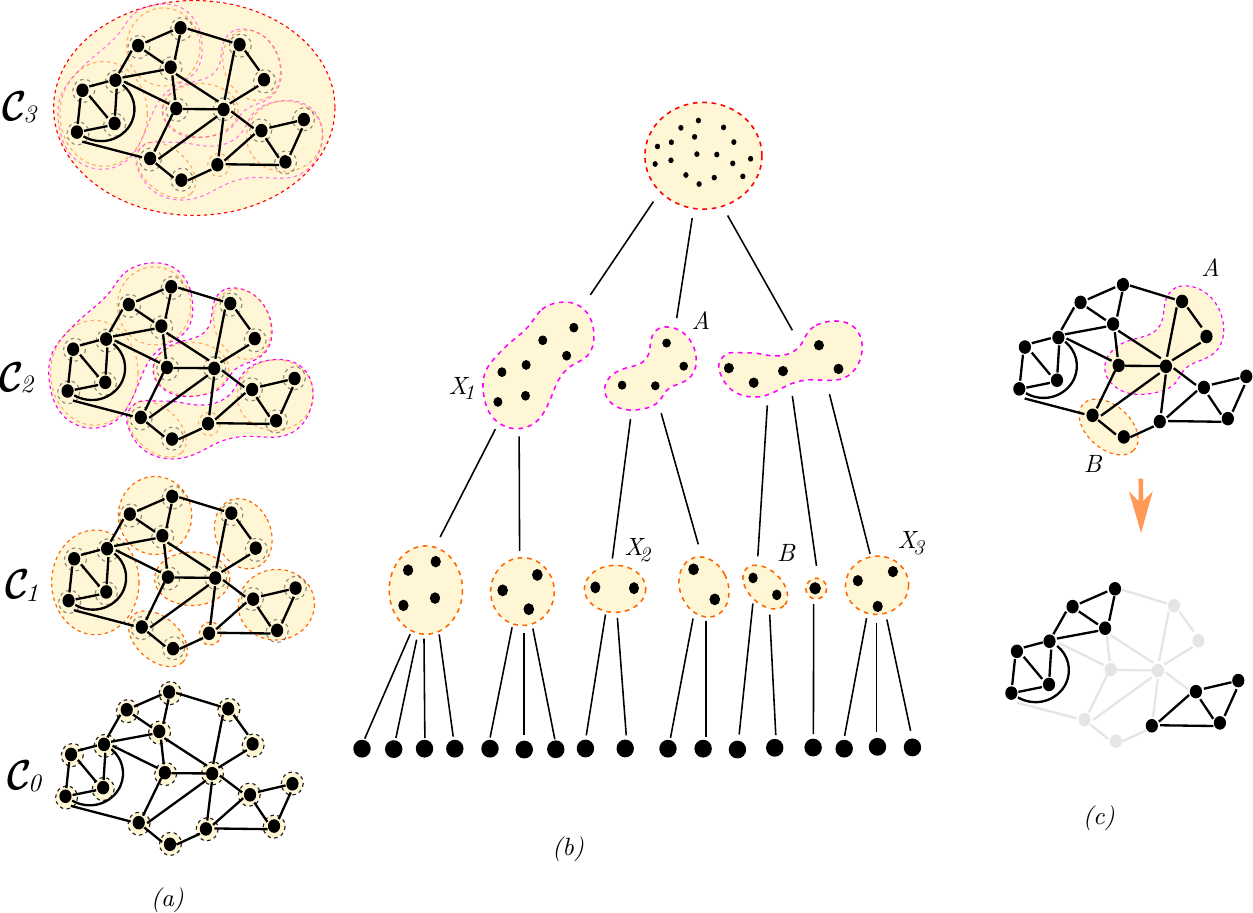}
\caption{%
(a) A clustering chain $\bbC = \{\cC_0,\cC_1,\cC_2,\cC_3\}$. 
(b) Viewing $\bbC$ as a hierarchy represented by a tree. 
(c) Balanced cut $\cF = \{A,B\}$ respects $\bbC$ since the cut contains two clusters $A$ and $B$ (at different levels) of $\bbC$; removing this cut results in connected components of size at most $2n/3$ each. 
If $\cX = \{X_1,X_2,X_3\}$ as in (b), then the cut $\cF$ in (c) conforms $\cX$.} 
\label{fig:clustering-chain}
\end{figure}

\medskip
For our randomized embedding, we need to construct a \emph{distribution} of clustering chains so that the probability of an edge being cut by a clustering $\cC_i$ at a specific scale $i$ is inverse-proportional to the diameter of the clusters in $\cC_i$. We formalize this via a notion of \emph{$\beta$-separating distribution}.

\begin{definition}[$\beta$-Separating Distribution of Clustering Chains with Hop Guarantee]
\label{def:beta-sep} 
We say that a distribution $\frakC$ of clustering chains of $G$ is \EMPH{$\beta$-separating} for some parameter $\beta \geq 1$ if for every edge $e\in E(G)$ and every $i$ in $\set{0, \ldots, k}$,
\begin{equation}
    \Pr_{\bbC \sim \frakC}[\text{$e$ is cut by $\cC_i \in \bbC$}] \leq \beta \cdot \frac{\len{e}}{2^i}.
\end{equation}
For a clustering chain $\bbC$, 
We say that $\bbC$ has \EMPH{hop bound} $h$
if for every $i$ in $\set{0, \ldots, k-1}$ and every cluster $C \in \cC_{i+1}$, the graph obtained by starting from $G[C]$ and contracting every cluster in $\cC_i$, denoted by $G[C] / \cC_i[C]$, has hop-diameter at most $h$. 
\end{definition}

\cite{CLPP23} showed that any $n$-vertex graph admits a $\beta$-separating distribution of clustering chains for  $\beta = \log\Phi \cdot \log(n)$. 
Here in Section~6, we show $\beta$ can be improved to $O(1)$ for minor-free graphs, and furthermore, the clustering chain has a small hop bound.
\begin{theorem} \label{thm:beta-separating}
    Let $G$ be a graph excluding a fixed minor. There exists a $\beta$-separating distribution of clustering chains $\frakC$ with $\beta = O(1)$. Furthermore, every chain $\bbC = \{ \cC_0, \ldots, \cC_k\}$ in the support of $\frakC$ has hop bound $h = O(1)$.   The big-O in both $\beta$ and $h$ hides a dependence on the minor size.  
\end{theorem}

Next, we introduce the notion of a balanced cut.  
A \EMPH{cut} of $G$, denoted by $\cF$, is a set of vertex-disjoint clusters that are proper subsets of $V$. (We do not allow the cut to contain a single cluster $V$.) In this paper, we construct a cut by taking some clusters from a given clustering chain $\mathbb{C} = (\cC_0, \ldots, \cC_k)$ of $G$: 
A cut $\cF$ \EMPH{respects} a clustering chain $\mathbb{C}$ if every cluster in $\cF$ is in $\mathbb{C}\setminus \{\cC_k\}$ (recall that $\cC_k$ contains a single cluster $V$).
See \Cref{fig:clustering-chain}(c). 
Let $\omega_V: V\rightarrow \bbR^+$ be a weight function on vertices of $G$.  The \EMPH{size} of $\cF$ is the number of clusters that $\cF$ contains. We say that a cut $\cF$ of $G$ is \EMPH{balanced} with respect to $\omega_V$ if the total vertex weight of  every connected component $H$ of $G\setminus \cF$ is at most $W/2$ where $W \coloneqq \sum_{u\in V}\omega_V(u)$.  

\begin{definition}[$(\tau,\psi)$-Stochastic Balanced Cuts]
\label{def:stoc-Bcut} 
Let $G = (V,E,\len{.})$ be a graph with vertex weight function $\omega_V$, $\bbC = (\cC_0, \ldots, \cC_k)$ be a clustering chain for $G$, and $\cX$ be a set of clusters in $\bbC$. 
We say $(G,\bbC)$ has a \EMPH{$(\tau,\psi)$-stochastic balanced cut with respect to $(\omega_{V},\cX)$} if there is a distribution $\frakF$ of cuts of $G$:
\begin{itemize}
    \item \textnormal{[Respecting.]} Every cut $\cF$ in the support of $\frakF$ respects $\bbC$.
    
    \item \textnormal{[Balanced.]} Every cut $\cF$ in the support of $\frakF$ is balanced with respect to $\omega_{V}$.
    
    \item  \textnormal{[Conforming.]} Every cut $\cF$ in the support of $\frakF$ has the property that: for every cluster $C\in \cF$, $C$ is not strictly contained in a cluster in $\cX$. ($C$ may be in or above clusters in $\cX$.) 
    We say that $\cF$ \EMPH{conforms $\cX$}. 
    
    \item \textnormal{[Small size.]}  Every cut $\cF$ in the support of $\frakF$ contains at most $\tau$ clusters.
    
    \item \textnormal{[Low probability.]} For every \emph{non-singleton} cluster $C$ in $\bbC$ such that $C\not\in \cX$:
    \begin{equation}
        \Pr_{\cF \sim \frakF}[\text{$C$ is contained in $\cF$}] \leq 1/\psi.
    \end{equation}
\end{itemize}
We say that \EMPH{$(G,\bbC)$ admits a $(\tau,\psi)$-stochastic balanced cut} if for any cluster  $C\in\bbC$, for any given weight function $\omega_{C}$, and any subset of clusters $\cX$ in $\bbC$, there exists a  $(\tau,\psi)$-stochastic balanced cut w.r.t  $(G[C],\omega_{C},\bbC_{\downarrow C},\cX)$. 
(Notice that $\cX$ can be in $\bbC$ outside $\bbC_{\downarrow C}$.) 
\end{definition}
The [conforming] property in the definition of stochastic balanced cut is somewhat artificial and unintuitive. 
This property is needed for a purely technical reason: at some step of the  embedding construction, we will sample a balanced cut $\cF$ from some subgraph $H$ of $G$. 
However, we are not allowed to break some ``boundary clusters'' of $H$---these are clusters in $\bbC$ that are adjacent to vertices not in $H$---in the sense that every cluster in $\cF$ either contains a whole boundary cluster or is disjoint from any boundary cluster. In other words, $\cF$ has to conform the boundary clusters as defined in the [conforming] property. 

\medskip
When constructing a stochastic balanced cut, the most representative case is when $\cX = \varnothing$, and it is simpler to focus on this case.

\begin{theorem}\label{thm:balanced-cuts} 
Let $G$ be an apex-minor-free graph, $\bbC$ be a clustering chain of $G$ that has hop bound $h$. Then for any parameter $\psi \geq 1$, $(G,\bbC)$ admits a $(\tau,\psi)$-stochastic balanced cut $\frakF$ with $\tau = O(h^2\cdot\log (\Phi))\cdot\psi$. 
\end{theorem}

Finally, we show that one can construct a low-distortion embedding into a low-treewidth graph for a graph $G$ if it admits a good clustering chain and a stochastic balanced cut. 

\begin{theorem}\label{thm:embedding} 
Let $G$ be an $n$-vertex edge-weighted graph with aspect ratio $\Phi$. Suppose that (i) $G$ admits a $\beta$-separating distribution $\frakC$ of clustering chains for some $\beta\geq 1$ and (ii) for every clustering chain $\bbC$ in the support of $\frakC$, $(G,\bbC)$ admits a $(\tau,\psi)$-stochastic balanced cut. Then $G$ can be stochastically embedded into graphs of treewidth $O(\tau + \log(n\Phi))$ and with expected distortion $1 + \frac{O(\beta\log(n\Phi)\log \Phi)}{\psi}$.
\end{theorem}

We observe that \Cref{thm:beta-separating}, \Cref{thm:balanced-cuts}, and \Cref{thm:embedding} together implies \Cref{thm:main}.

\begin{proof}[of \Cref{thm:main}] 
Let $G = (V,E,\len{\cdot})$ be an apex-minor-free graph. By \Cref{thm:beta-separating}, there exists a $\beta = O(1)$-separating distribution $\frakC$ of clustering chains of hop bound $h = O(1)$.  We choose
\begin{equation*}
    \psi \coloneqq \frac{\beta \log(n\Phi)
\log(\Phi)}{\eps}
\end{equation*}

Let $\bbC$ be a clustering chain sampled from $\frakC$. Let $\cX$ be any set of clusters in $\bbC$. Let $C$ be any cluster in $\bbC$. Since $G$ is apex-minor-free, $G[C]$ is also apex-minor-free. Thus, by \Cref{thm:balanced-cuts}, $(G[C], \bbC_{\downarrow C})$ admits a $(\tau,\psi)$-balanced cut conforming $\cX$ with $\tau = O(h^2\log\Phi)\cdot \psi = O(\log^2 \Phi\log(n\Phi)/\eps)$. By \Cref{thm:embedding}, we can embed $G$ stochastically into graphs with treewidth $O(\tau + \log(n\Phi)) = O(\log^2 \Phi \log(n\Phi)/\eps)$. The expected distortion of the embedding is:

\begin{equation*}
    1 + O\Paren{ \frac{\beta \log(n\Phi) \log\Phi}{\psi} } = 1 + O(\eps),
\end{equation*}
by our choice of $\psi$.
\end{proof}

\section{Preliminaries}

We use mostly standard graph terminology. We denote by $G$ an edge-weighted graph with nonnegative length function $\len{\cdot}$ on the edges. 
For any graph $H$, we denote by $V(H)$ and $E(H)$ the vertex set and edge set of $H$, respectively. 
For a subset of vertices $S$, we denote by $G[S]$ the subgraph of $G$ induced by $S$.

A \EMPH{tree decomposition} of a graph $G = (V,E)$ is a tree $\cT$ where each node $x\in \cT$ is associated with a subset of $V$ called \EMPH{$\bag(x)$}, such that: (i) $\cup_{x\in \cT}\, \bag(x) = V$, (ii) if $u$ and $v$ are adjacent in $G$, there exists a node $x$ in $\cT$ such that $\{u,v\}\subseteq \bag(x)$, and (iii) for every $u\in V$, all the nodes whose bags contains $u$ induce a connected subtree of $\cT$.  The \EMPH{width} of a tree decomposition $\cT$ is $\max_{x\in \cT} |\bag(x)| - 1$ and the \EMPH{treewidth} of $G$, denoted by \EMPH{$\tw(G)$}, is the minimum width over all valid tree decompositions of $G$. 

Given a graph $H$, we say that $G$ is \EMPH{$H$-minor-free} if $G$ excludes $H$ as a minor. We say that $G$ is \EMPH{minor-free} if it excludes a graph of constant size as a minor.  A graph $K$ is an \EMPH{apex} graph if there exists a vertex $u$, called the apex of $K$, such that $K-\{u\}$ is a planar graph. We say that $G$ is \EMPH{apex-minor-free} if it excludes an apex graph of constant size as a minor. We will rely on the following lemma in our construction:

\begin{lemma}[Lemma 2.1 in~\cite{CLPP23}]
\label{lem:tw:balls}
For every apex graph $K$, there is a positive integer $\alpha_K$ such that if $G = (V,E)$ is $K$-(apex-)minor-free and $Z\subseteq V$ is a subset of vertices such that every vertex of $G$ is at hop distance at most $d$ from a vertex of $Z$, then $\tw(G)\leq \alpha_K \cdot d \sqrt{|Z|}$.
\end{lemma}

\section{The Embedding Framework: Proof of Theorem~\ref{thm:embedding}}

Let $G$ be a graph satisfying the assumptions in \Cref{thm:embedding}. 
Let $\bbC$ be a clustering chain for $G$ sampled from the distribution $\frakC$. 
Roughly speaking, the high-level idea is to (i) sample a balanced cut $\cF$ to cut $G$ into connected components of size at most $n/2$, (ii) pick a representative vertex $v_C$ for each cluster  $C\in \cF$ and make a bag from $\set{v_C: C\in \cF}$, and (iii) recursively embed vertices in each connected component%
\footnote{More precise, for presentation, the graph considered in the next recursive step to embed a connected component of $G\setminus \cF$, say $H$, remains to be named $G$. 
However, we will set the weight $\omega_V(u) = 1$  for every $u \in V(H)$ and $0$ otherwise. Therefore, effectively, we only embed vertices of $H$ in the recursion.} 
of $G\setminus \cF$  and each cluster in $\cF$.  As we recurse on each cluster $C$ in $\cF$, we will need to sample a balanced cut for $C$ from $\bbC_{\downarrow C}$; this is possible by the definition of stochastic balanced cut in \Cref{def:stoc-Bcut}. 

When we recurse to embed vertices in a connected component, say $H$, of  $G\setminus \cF$, clusters in $\cF$ become ``boundary clusters'' of $H$, meaning that vertices in $H$ can only connect to other vertices in $G\setminus H$ by going through these clusters. 
In the algorithm below, we will take one representative per boundary cluster and embed these representatives along with vertices in $H$ in the recursion.  
We will keep the number of such representatives small by the following trick reminiscent to the construction of $r$-division: whenever the number of boundary clusters is above a certain threshold, we will sample a balanced cut to balance the number of boundary clusters.

\begin{figure}[ph!]
\centering
   
\begin{tcolorbox}
\paragraph{$\textsc{Embed}_{\bbC}(T, P\/, \bdry T, \bdry \cC)$:} Terminal set $T$, current graph $P$ (called piece), boundary terminals $\bdry T$, and boundary clusters $\bdry \cC$.
\begin{enumerate}
    \item \emph{Base case.}
    
    If $T$ contains at most $4 \tau$ vertices, then define $\hat{G}$ to be a clique on $T$, where the weight between $u$ and $v$ is $\dist_G(u,v)$; define $\hat{\cT}$ to be a single bag containing all vertices in $T$; and return $(\hat{G}, \hat{T})$. If $T$ contains more than $4 \tau$ vertices, continue.

    \item \emph{Compute a balanced cut $\cF$.}
    
    If $\bdry T$ contains more than $4 \tau$ vertices, set $\EMPH{$S$} \gets \bdry T$; otherwise set $S \gets T$.
    Sample a cut \EMPH{$\cF$} by calling $\textsc{Cut}_{\bbC}(S, P\/, \bdry \cC)$. Notice that $\cF$ is balanced with respect to $S$ and conforms $\bdry \cC$.
    
    \item \emph{Recurse on $\cF$.}
    
    For each cluster $C$ in $\cF$, let $\EMPH{$T_C$} \coloneqq T \cap C$, and let $\EMPH{$\bdry T_C$} \coloneqq \bdry T \cap C$. 
    Let \EMPH{$\bdry \cC_C$} be the set of clusters in $\bdry \cC$ whose representatives are in $\bdry T_C$. 
    Recursively call $\textsc{Embed}_{\bbC}(T_C, G[C], \bdry T_C, \bdry \cC_C)$ to find an embedding \EMPH{$(\hat{G}_C, \hat{\cT}_C)$}.
    
    \item \emph{Recurse on $P \setminus \bigcup \cF$.}
    
    For each cluster $C$ in $\cF$ with $C \cap T \neq \varnothing$, select an arbitrary representative terminal \EMPH{$v_C$} from $C \cap T$; let \EMPH{$T_{\cF}$} denote the set of representatives for all clusters in $\cF$.
    For each connected component $H$ of $P \setminus \bigcup \cF$, let $\EMPH{$T_H$} \coloneqq (T \cap H) \cup T_{\cF}$, let $\EMPH{$\bdry T_H$} \coloneqq (\bdry T \cap H) \cup T_{\cF}$, and let $\EMPH{$\bdry \cC_{H}$}$ be the set containing all clusters in $\cF$ and in $\bdry \cC$ whose representatives are in $\bdry T_H$.   
    Recursively call $\textsc{Embed}_{\bbC}(T_H, P\/, \bdry T_H, \bdry \cC_H)$ to find an embedding \EMPH{$(\hat{G}_H, \hat{\cT}_H)$}. 
    (Note here that we pass $P$ to the next recursive call instead of $H$.)

    \item \emph{Construct the host graph $\hat{G}$.}
    
    Take the union of all $\hat{G}_C$ and $\hat{G}_H$ recursively constructed above (identifying all vertices in $T_{\cF}$, which may appear in multiple graphs $\hat{G}_H$). For each cluster $C$ in $\cF$, add an edge from the representative vertex $v_C$ to every vertex $u$ in $C$, with weight equal to $\dist_G(v_C, u)$. The resulting graph is $\hat{G}$.

    \item \emph{Construct the tree decomposition $\hat{\cT}$.}
    
    The root bag of $\hat{\cT}$ contains every vertex in $T_{\cF} \cup \bdry T$. The children of the root bag are the roots of the tree decompositions for the recursive embeddings ($\hat{\cT}_{C}$ and $\hat{\cT}_{H}$). For every cluster $C$ in $\cF$, we add $v_C$ to every bag of $\hat{\cT}_{C}$. 
    The resulting tree is $\hat{\cT}$. Return $(\hat{G}, \hat{\cT})$.
\end{enumerate}
\end{tcolorbox}
\begin{tcolorbox}
    \paragraph{$\textsc{Embed}(G)$:} graph $G = (V,E,\len{\cdot})$ 
    \begin{enumerate}
        \item Sample $\bbC$ from distribution $\frakC$ in the assumption of \Cref{thm:embedding}.
        \item  Return  $\textsc{Embed}_{\bbC}(V, G, \varnothing, \varnothing)$.
    \end{enumerate}
\end{tcolorbox}
    \caption{The embedding algorithm.}
    \label{fig:embed-main}
\end{figure}

\medskip
Now, we describe the embedding algorithm in more detail.  Throughout this section, fix $\e$ in $(0,1)$.  The main embedding procedure \EMPH{$\textsc{Embed}_{\bbC}(T, P\/, \bdry T, \bdry \cC)$} takes four parameters as input:
\begin{itemize}
    \item the \EMPH{terminals $T$}, a subset of the vertices of $G$ that we want to embed, 
    
    \item the \EMPH{current piece $P$}, a subgraph%
    \footnote{we remark that we pass $P$ as input only to help intuition; a suitable $P$ could be inferred from $T$ and $\bbC$ alone} 
    induced by a cluster in $\bbC$ such that all terminals are within $P$,
    
    \item the \EMPH{boundary vertices $\bdry T$}, a subset of $T$ that contains representatives of boundary clusters, and
    
    \item the \EMPH{boundary clusters $\bdry \cC$}, a set of clusters in $\bbC$ 
    where every cluster $C$ in $\bdry \cC$ satisfies $|C \cap \bdry T| = 1$. \\
    (We remark that intuitively $\bdry \cC$ are clusters that were added to a cut during earlier recursive calls.)
\end{itemize}
The procedure returns a pair \EMPH{$(\hat{G}, \hat{\cT})$}, where $\hat{G}$ is a small treewidth graph with vertex set $T$, and $\hat{\cT}$ is a tree decomposition for $\hat{G}$ where every vertex in $\bdry T$ appears in the root bag of $\hat{\cT}$. 
Implementation details of $\textsc{Embed}_{\bbC}(T, P, \bdry T, \bdry \cC)$ are in \Cref{fig:embed-main}. 
The wrapper function \EMPH{$\textsc{Embed}(G)$} sample a clustering chain $\bbC$ 
from distribution $\frakC$ in the assumption of \Cref{thm:embedding}, and produce an embedding of $G$ by calling $\textsc{Embed}_{\bbC}(V, G, \varnothing, \varnothing)$. 
See \Cref{fig:emb-example} for an illustration. 
We bound the treewidth of~$\hat{\cT}$ in \Cref{subsec:treewidth} and bound the expected distortion of embedding into $\hat{G}$ in \Cref{subsec:distortion}, completing the proof of \Cref{thm:embedding}. 

\begin{figure}[!htb]
\centering
\begin{tikzpicture}
	 \node at (0,0){\includegraphics[width=1.0\linewidth]{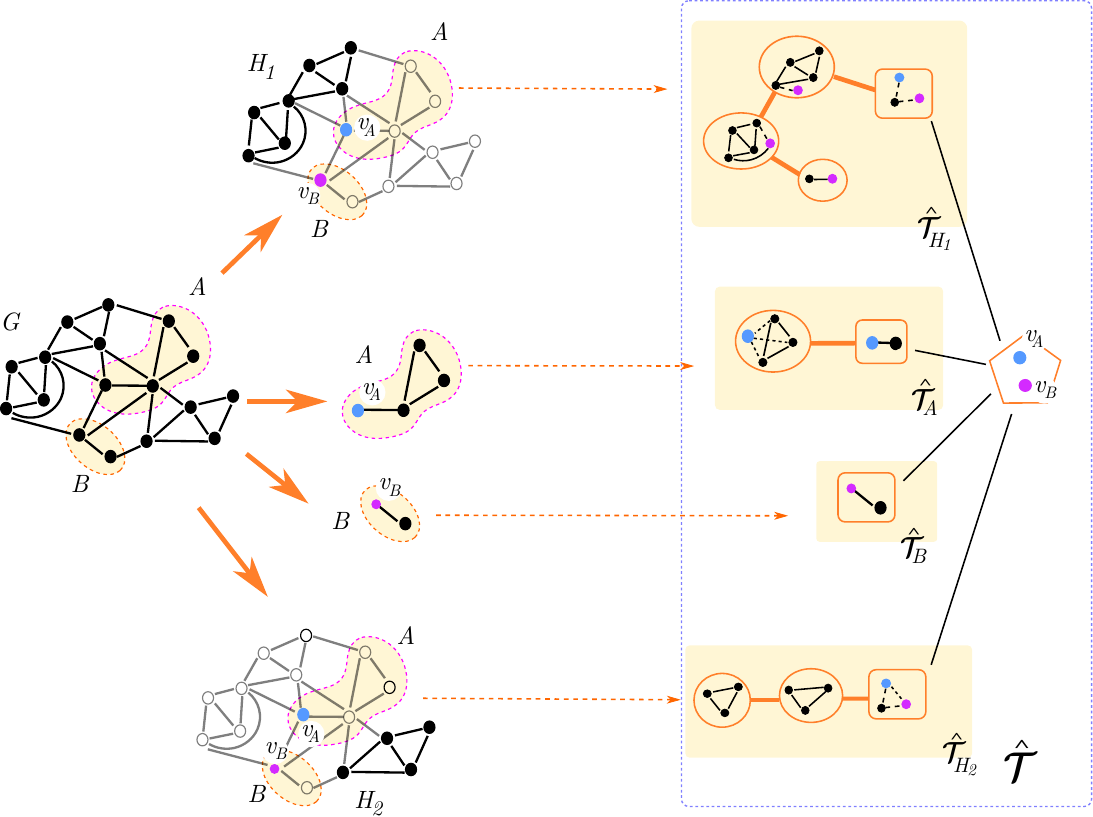}};		
			\node at (0.6,5.1){\footnotesize $\textsc{Embed}_{\bbC}(T_{1}, G, \{v_A,v_B\}, \{A,B\})$};
			\node at (0.2,0.9){\footnotesize $\textsc{Embed}_{\bbC}(T_{A}, A, \varnothing, \varnothing)$};
            \node at (0.2,-1.3){\footnotesize $\textsc{Embed}_{\bbC}(T_{B}, B, \varnothing, \varnothing)$};
            \node at (0.1,-4.1){\footnotesize $\textsc{Embed}_{\bbC}(T_{2}, G, \{v_A,v_B\}, \{A,B\})$};
		\end{tikzpicture}   
\caption{An embedding of $G$ obtained by calling $\textsc{Embed}_{\bbC}(V, G, \varnothing, \varnothing)$: the algorithm will call $\textsc{Cut}_{\bbC}(V, G, \varnothing)$ to sample a cut $\cF$; assume that $\cF = \{A,B\}$ as in the figure. Then, the algorithm will recursively embed the terminals in components of $G\setminus \cF$, which are $H_1,H_2$. Note that the graphs we recurse on to embed terminals in $H_1$ and $H_2$ are $G$. Terminals are filled vertices, while non-terminals are hollow vertices. The algorithm also recursively embeds clusters in $\cF$, which are $A$ and $B$. The final tree decomposition $\hat{\cT}$ is obtained by connecting four tree decompositions from four recursive calls via the root bag $\{v_A,v_B\}$, which are representative vertices of $A$ and $B$, respectively. Dashed edges are new edges added during the embedding process.}  
\label{fig:emb-example}
\end{figure}

We will use a (randomized) helper procedure \EMPH{$\textsc{Cut}_{\bbC}(S, P\/, \bdry \cC)$}, where $S\subseteq V(P)$ and $\bdry \cC$ is the set of boundary clusters of graph $P$ with respect to $\bdry T$. 
This procedure will sample a stochastic random cut $\cF$ in $P$ conforming $\bdry \cC$ to balance the number of vertices in $S$, achieved by setting the weight $\omega_{V(P)}(u) = 1$ for every $u \in S$ and  $\omega_{V(P)}(v) = 0$ for every $v\in V(P)\setminus S$. 
Since $\cF$ conforms $\bdry \cC$, no cluster in $\cF$ will be a proper subset of a cluster in $\bdry\cC$ (though $\cF$ could contain clusters in $\bdry \cC$).  
Note that we only call procedure $\textsc{Cut}_{\bbC}(S, P\/, \bdry \cC)$ when $P$ is the subgraph of $G$ induced by some cluster $\hat{C}\in \bbC$; we will show in \Cref{obs:P-cluster} below that this is the case during the course of the embedding algorithm.

\begin{observation}\label{obs:P-cluster} 
For any recursive call  $\textsc{Embed}_{\bbC}(T, P\/, \bdry T, \bdry \cC)$, $P = G[\hat{C}]$ for some cluster $\hat{C} \in \bbC$.
\end{observation}

\begin{proof}
    In the beginning, the observation follows from the fact that $P = G$ and the root of $\bbC$ is $V$. Inductively, we assume that in the call $\textsc{Embed}_{\bbC}(T, P\/, \bdry T, \bdry \cC)$, $P = G[\hat{C}]$ for some cluster $\hat{C} \in \bbC$. The call $\textsc{Embed}_{\bbC}(T, P\/, \bdry T, \bdry \cC)$ spawns two types of recursive calls: (i)  $\textsc{Embed}_{\bbC}(T_C, G[C], \bdry T_C, \bdry \cC)$ for a cluster $C\in \cF$ in Step~3 and (ii) $\textsc{Embed}_{\bbC}(T_H, P\/, \bdry T_H, \bdry \cC_H)$ in Step~4. The observation clearly holds for type (ii).  For type (i), the next piece is $G[C]$, so the observation also holds.
\end{proof}

\subsection{Bounding the treewidth}
\label{subsec:treewidth}

First, we show that $\hat{\cT}$ output by $\textsc{Embed}(G)$ is a valid tree decomposition of $\hat{G}$.  

\begin{lemma}
    $\hat{\cT}$ is a valid tree decomposition of $\hat{G}$.
\end{lemma}

\begin{proof} By induction, we assume that in procedure $\textsc{Embed}_{\bbC}(T, P\/, \bdry T, \bdry \cC)$,  the tree decompositions $\hat{\cT}_C$  in Step 3  and  $\hat{\cT}_H$ in Step 4 are valid tree decompositions of $\hat{G}_C$ and $\hat{G}_H$, respectively.  Recall that recursive calls are made on subsets $T_H$ and $T_C$, which are pairwise disjoint \emph{except} for the vertices of $T_{\cF}$. Thus, if a vertex in $\hat{G}$ appears in multiple subtrees, then it is in $T_{\cF}$ (and thus in the root bag). Further, $T_{\cF}$ appears (if at all) in the root bag of all tree decompositions for the recursive embeddings. Recursively, each subtree is a tree decomposition. Thus, the bags containing any given vertex form a connected subtree.

 Let $e$ be an edge in $\hat{G}$. Recall that in Step 5, we only add edges between the representative $v_C$ of $C$ to vertices in $C$ for every cluster $C\in \cF$ (aside from edges in $\hat{G}_C$ and $\hat{G}_H$)  to form $\hat{G}$. If $e$ is an edge between some $v_C$ and a vertex in $C$ (that is, if it was added in Step 5), then there is a bag containing both endpoints of $e$ (as we added $v_C$ to every bag of the tree decomposition for cluster $C$). Otherwise, recursively, there is some bag that contains both endpoints of $e$.
\end{proof}

To bound the width of $\hat{T}$, we must first bound the depth of the recursion tree.
Let $X$ be a call $\textsc{Embed}_{\bbC}(T, P\/, \bdry T, \bdry \cC)$ that occurs during the algorithm. 
Let \EMPH{$\mathrm{scale}(P)$} denote the scale of the cluster $\hat{C}$ inducing $P = G[\hat{C}]$ in $\bbC$. Note that $\mathrm{scale}(P)\leq \lceil \log \Phi\rceil$.  We define the potential function: 
\[
\EMPH{$\phi(T, P\/, \bdry T)$} \coloneqq 5 \log_2 |T| + \mathrm{scale}(P) + \frac{|\bdry T|}{\tau}.  
\]
Recall that $\tau$ is the parameter in the $(\tau,\psi)$-stochastic balanced cut $\cF$ sampled in Step~2. 
The root call has potential $\phi(T, P\/, \bdry T) = O(\log n) + O(\log \Phi) + 0 = O(\log(n\Phi))$.

\begin{lemma}
\label{lem:recursion-depth}
    If call $X \coloneqq \textsc{Embed}_{\bbC}(T, P\/, \bdry T, \bdry \cC)$ makes a recursive call $X' \coloneqq \textsc{Embed}_{\bbC}(T', P', \bdry T', \bdry \cC')$, then $\phi(T', P', \bdry T') \le \phi(T, P\/, \bdry T) - 1$. 
    As a corollary, the recursion tree has depth~$O(\log (n\Phi))$.
\end{lemma}

\begin{proof}
In Step 3, $X$ makes calls $X_C \coloneqq \textsc{Embed}_{\bbC}(T_C, G[C], \bdry T_C, \bdry \cC)$ for every cluster $C$ in $\cF$. 
We have $|T_C| \le |T|$ (as $T_C \subseteq T$) and $|\bdry T_C| \leq |\bdry T|$ (as $\bdry T_C \subseteq \bdry T$). 
We have $\mathrm{scale}(G[C]) \le \mathrm{scale}(P) - 1$ by the definition of a cut. 
We conclude that $\phi(T_C, C, \bdry T_C) \le \phi(T, P\/, \bdry T) - 1$.

\medskip \noindent 
In Step 4, $X$ makes calls $X_H \coloneqq \textsc{Embed}_{\bbC}(T_H, P\/, \bdry T_H, \bdry \cC)$ for every connected component $H$ in $P \setminus \bigcup \cF$.  
We consider two cases based on Step~2 of $X$.
\begin{itemize}
    \item Case 1: $|\bdry T| \le 4 \tau$, so that $X$ selects a balanced cut $\cF$ with respect to the terminals $T$.
    
    In this case, we have $|T_H| \le \frac{1}{2} \cdot |T| + |T_{\cF}|$.  
    Note that $|T_{\cF}|\leq \tau$ since $\cF$ is sampled from a $(\tau,\psi)$-stochastic balanced cut. Observe that $|T| \ge 4 \tau$ (because $X$ makes a recursive call which means the base case of Step 1 was not reached). 
    Thus, $|T|\geq 4 \cdot |T_{\cF}|$, giving $|T_H| \le \frac{3}{4} \cdot  |T|$, and so $5 \log_2 |T_H| \le 5 \log_2 |T| - 2$. 
    Further, $|\bdry T_H| \le |\bdry T| + \tau$ (because we add at $|T_{\cF}| \le \tau$ points into the boundary), and so $\frac{|\bdry T_H|}{\tau} \le \frac{|\bdry T|}{\tau} + 1$. We conclude that $\phi(T_H, P\/, \bdry T_H) \le \phi(T, P\/, \bdry T) - 1$.
    \item Case 2: $|\bdry T| > 4 \tau$, so that $X$ selects a balanced cut $\cF$ with respect to $\bdry T$.
    
    In this case, we have $|T_H| \le |T|$. Further, 
    $|\bdry T_H| 
    \le \frac{1}{2} \cdot |\bdry T| + \tau 
    \le |\bdry T| - \tau$ 
    and hence $\frac{|\bdry T_H|}{\tau} \le \frac{|\bdry T|}{\tau} - 1$.  We conclude that $\phi(T_H, P\/, \bdry T_H) \le \phi(T, P\/, \bdry T) - 1$. \qed
\end{itemize}
\vspace{-16pt}
\end{proof}

\begin{lemma}
    $\hat{\cT}$ has width $O(\tau + \log (n\Phi))$.
\end{lemma}
\begin{proof} By~\Cref{lem:recursion-depth}, the recursion tree has depth $O(\log (n\Phi))$.
    This means that, in total, we add at most $O(\log (n\Phi))$ representative vertices $v_C$ to a bag in Step 6 over the course of the algorithm. Excluding the vertices $\set{v_C}$, every bag contains a set of vertices that is $T_{\cF} \cup \bdry T$, for some $\bdry T$ that is passed as a parameter to $\textsc{Embed}$. Thus, every bag contains at most $\tau + |\bdry T|$ vertices.
    
    Let $X \coloneqq \textsc{Embed}_{\bbC}(T, P\/, \bdry T, \bdry \cC)$ be a call that occurs during the algorithm. We show that $|\bdry T| \le 5 \cdot \tau$.
    We proceed by induction, starting at the root of the recursion tree and working downward. 
    Indeed, let $X' \coloneqq \textsc{Embed}_{\bbC}(T', P', \bdry T', \bdry \cC')$ be the parent call of $X$ in the recursion tree (and if $X$ is the root, then $\bdry T = \varnothing$ by construction). Let $\cF'$ be the cut sampled during $X'$. If $P = G[C']$ for some cluster $C' \in \cF'$  (that is, if $X$ was called recursively during Step 3), then $\bdry T \subseteq \bdry T'$ and the induction hypothesis proves the claim. 
    Otherwise, $X$ was called recursively during Step 4, and so all terminals $T$ are contained in $H \cup \cF'$ for some connected component $H$ of $P' \setminus \bigcup \cF'$. 
    In particular, $T$ contains at most $\tau$ terminals in $\cF'$ (the representative terminals $T_{\cF'}$), and so we have $|\bdry T| \le |\bdry T' \cap H| + \tau$. 
    If $|\bdry T'| \le 4 \cdot \tau$, then we are done; otherwise, if $|\bdry T'| > 4 \cdot \tau$, the cut $\cF$ is a balanced separator for $\bdry T'$ and so $|\bdry T' \cap H| \le |\bdry T'|/2 \le 2.5 \cdot \tau$ --- note here that $|\bdry T'| \le 5 \cdot \tau$ by induction --- giving $|\bdry T| \leq 5 \cdot \tau$ as desired.
\end{proof}

\subsection{Bounding the expected distortion}
\label{subsec:distortion}
It will be helpful to bound the distortion of an \emph{edge} $e = (u,v)$, rather than an arbitrary \emph{path} between $u$ and $v$. This will simplify the analysis because either a cut $\cF$ cuts edge $e$, or it does not --- if we instead consider a path, $\cF$ may cut multiple edges along the path.
We generalize this notion. Let $\cX$ be a set of clusters in a clustering chain $\bbC$. 
Let \EMPH{$G / \cX$} denote the graph obtained by contracting every cluster in $\cX$ into a supernode, and let $\eta: V \to V(G/\cX)$ be the surjective map that takes a vertex in $G$ to the corresponding vertex in $G/\cX$ after contraction. 
We say that a pair of vertices $(u,v)$ is a \EMPH{pseudo-edge with respect to $\cX$  in $G$} if there is an edge between $\eta(u)$ and $\eta(v)$ in $G / \cX$. We simply say that  $(u,v)$ is a pseudo-edge with respect to $\cX$ if the graph $G$ is clear from the context.

\begin{observation}
\label{obs:edge-cut}
    Let $(u,v)$ be a pseudo-edge with respect to $\cX$ in $G$ where $\cX$ is a set of clusters in $\bbC$. Let $\cF$ be a cut conforming $\cX$, such that $u$ is in some cluster $C_u$ in $\cF$, but $v$ is not in $C_u$. Then, for every vertex $u'$ in $C_u$, the pair $(u', v)$ is a pseudo-edge with respect to $\cX \cup \set{C_u}$. (Note that $C_u$ might contain some clusters in $\cX$ as subclusters.)
\end{observation}

\begin{proof}
    Let $\eta:V \to V(G/\cX)$ be as above. Define $\eta':V \to V(G/(\cX \cup \set{C_u}))$ to be the function that sends vertex $x$ to the supernode $C_u$ if $x \in C_u$, and to $\eta(x)$ if $x \not \in C_u$. As $\eta(v)$ and $\eta(u)$ are adjacent in $V(G/\cX)$, they remain adjacent after the cluster $C_u$ is contracted; that is, $\eta'(u)$ and $\eta'(v)$ are adjacent. Now observe that every vertex $u'$ in $C_u$ is mapped to $\eta'(u)$ by $\eta'$, and so $(u',v)$ is a pseudo-edge.
\end{proof}

Let $(u,v)$ be a pair of vertices.
We say that a clustering chain $\bbC$ \EMPH{splits $(u,v)$ at scale $i$} if $i$ is the \emph{smallest} number such that, for every scale $s > i$, there is some cluster $C$ in $\bbC$ at scale $s$ that contains both $u$ and $v$.

\begin{observation}
\label{obs:split}
    Let $u$ and $v$ be vertices that are split at scale $i$ by $\bbC$. 
    Let $C$ be a cluster (at \emph{any} scale) in $\bbC$ that contains $u$ but not $v$. 
    Then, for every vertex $u'$ in $C$, the pair $(u', v)$ is split at scale $i$ by $\bbC$.
\end{observation}

\begin{proof}
    This follows from the nested structure of the clusters in $\bbC$. For every scale $s > i$, there is a cluster $C_s$ that contains both $u$ and $v$. The nested structure of $\bbC$ implies that $C \subseteq C_s$, and so $C_s$ contains both $u'$ and $v$. At scale $i$, there is a cluster $C_i$ in $\bbC$ that contains $u$ but not $v$. We have $C \subseteq C_i$, so $C_i$ contains $u'$ but not $v$.
\end{proof}

The key to the distortion analysis is the following lemma. 

\begin{lemma}
    \label{lem:distortion-inductive}
    Let $u$ and $v$ be vertices in $G$. 
    Let $\bbC$ be a clustering chain sampled from distribution $\frakC$ for $G$ that splits $(u,v)$ at scale $i$. 
    For any call to $\textsc{Embed}_{\bbC}(T, P\/, \bdry T, \bdry \cC)$ such that
    \begin{itemize}
        \item the set of terminals $T$ contains $u$ and $v$, and
        \item $(u,v)$ is a pseudo-edge with respect to $\bdry \cC$,
    \end{itemize}
    the output embedding $\hat{G}$ satisfies \[
    \mathbb{E}[\dist_{\hat{G}}(u,v)] \le \dist_G(u,v) + O\Paren{ 2^i \cdot \frac{\log (n\Phi)}{\psi} },
    \]
    where the expectation is over the sampled random cuts.
\end{lemma}

Before proving this lemma, we show that it suffices for a good multiplicative distortion bound. 

\begin{lemma}
\label{lem:distortion}
    Let $(\hat{G},\hat{\cT})$ be the output of $\textsc{Embed}_{\bbC}(V, G, \varnothing, \varnothing)$. 
    For any edge $e = (u,v)$ in $G$, in expectation we have 
    \[
    \mathbb{E}[\dist_{\hat{G}}(u,v)] \le \Paren{ 1 + O\Paren{\beta \cdot \frac{\log(n\Phi)\log(\Phi)}{\psi} } } \cdot \len{e}.
    \]
\end{lemma}

\begin{proof}
The first step in creating $\hat{G}$ is to sample a clustering chain $\bbC = (\cC_0, \ldots, \cC_k)$ from a $\beta$-separating distribution $\frakC$. We split the analysis into cases, based on the scale $i$ in $\set{0, \ldots, k}$ that splits $(u, v)$. There are $k = O(\log \Phi)$ different cases: 
    \begin{align*}
        \mathbb{E}[\dist_{\hat{G}}(u,v)] &=  \sum_{i=1}^{k} \mathbb{E}[\dist_{\hat{G}}(u,v) \mid \bbC \text{ splits $(u, v)$ at $i$}] \cdot \Pr[\bbC\text{ splits $(u, v)$ at $i$}]\\
        &\le \sum_{i=1}^{k} \Paren{ \len{e} + O\Paren{ 2^i \cdot \frac{\log (n\Phi)}{\psi} } } \cdot \Pr[\bbC\text{ splits $(u, v)$ at $i$}]&\text{(by \Cref{lem:distortion-inductive})}\\
        &\le \len{e} + \sum_{i=1}^{k} O\Paren{ 2^i \cdot \frac{\log (n\Phi)}{\psi} } \cdot \Pr[\bbC\text{ splits $(u, v)$ at $i$}]\\
        &\le \len{e} + \sum_{i=1}^{k} O\Paren{ 2^i \cdot \frac{\log (n\Phi)}{\psi} } \cdot \beta\cdot \frac{\len{e}}{2^i}
        &\text{(by \Cref{def:beta-sep})}\\
        &\le \len{e} + O\Paren{\beta \cdot \frac{\log (n\Phi)\log \Phi}{\psi}} \cdot \len{e}
    \end{align*}
The inequality in the second line follows because every edge $(u,v)$ is a pseudo-edge with respect to $\varnothing$; thus we may apply~\Cref{lem:distortion-inductive}. The last line follows from the fact that $k = O(\log \Phi)$.
\end{proof}

\bigskip
\noindent It remains to prove~\Cref{lem:distortion-inductive}.
\begin{proof}[of~\Cref{lem:distortion-inductive}]
Recall that we defined the potential function \[
\phi(T, P\/, \bdry T) \coloneqq (5 \log |T|) + \mathrm{scale}(P) + \frac{|\bdry T|}{\tau}. 
\]
To prove the lemma, we prove the stronger statement:
\begin{quote}
    If $\textsc{Embed}_{\bbC}(T, P\/, \bdry T, \bdry \cC)$ is called, where (1) $T$ contains $u$ and $v$, and (2) $(u,v)$ is a pseudo-edge with respect to $\bdry \cC$ in $P$, and then
    \[
    \mathbb{E}[\dist_{\hat{G}}(u,v)] \le \dist_G(u,v) + \alpha \cdot \frac{2^i \cdot \phi(T, P\/, \bdry T)}{\psi}
    \]
    where $i$ is the scale at which $\bbC$ splits $(u,v)$, and where $\alpha$ is a sufficiently large constant.
\end{quote}
Notice that $\phi(T, P\/, \bdry T) = O(\log (n\Phi))$, so proving this is sufficient to prove the lemma. We proceed by induction on $\phi(T, P\/, \bdry T)$. 

\medskip \noindent \textbf{Base case.} 
The base case is when $|T| \le 4 \cdot \tau$. (In particular, this case must occur if $\phi(T, P\/, \bdry T) \le  4 \cdot \tau$.) In this case, the graph $\hat{G}$ is a clique on $T$, and $\dist_{\hat{G}}(u,v) = \dist_G(u,v)$.

\medskip \noindent \textbf{Inductive case.} 
In the inductive case, the procedure $\textsc{Embed}$ samples a cut $\cF$ that conforms to $\bdry \cC$ in Step 3 from distribution $\frakF$. We will consider three cases below. To find the expected distortion, we will take a weighted sum of the distortion incurred for each of these cases. To simplify notation, we define the \EMPH{distortion overhead $D$} to be
\[
D \coloneqq \dist_{\hat{G}}(u,v) - \Paren{\dist_G(u,v) + \alpha \cdot \frac{2^i \cdot (\phi(T, P\/, \bdry T) - 1)}{\psi}}.
\]
(Note that the term in the parenthesis will arise naturally when we appeal to the induction hypothesis.) To prove the claim, we show that $\mathbb{E}[D] \le \alpha\cdot \frac{2^i}{\psi}$ where the expectation is over the sampled random cuts.

\begin{figure}[!htb]
\centering
\includegraphics[width=0.8\linewidth]{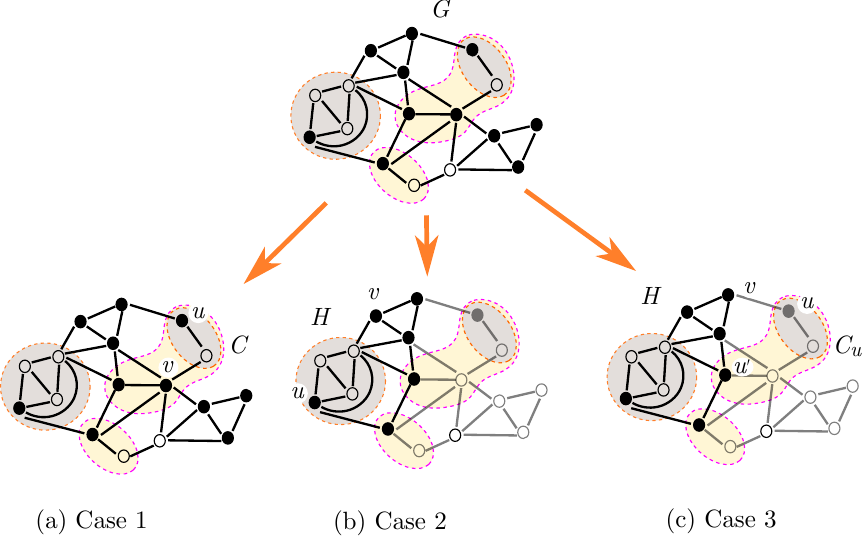}
\caption{Three cases in the analysis of the distortion overhead $D$. Filled vertices are terminals. The gray clusters form $\bdry \cC$.  The yellow clusters form $\cF$. 
}
\label{fig:distortion}
\end{figure}

\bigskip \noindent \EMPH{Case $\xi_1$}: Vertices $u$ and $v$ are assigned to the same cluster $C$ in $\cF$; see \Cref{fig:distortion}(a).

\smallskip \noindent In this case, the procedure recurses on $\textsc{Embed}(T_C, G[C], \bdry T_C, \bdry \cC_C)$.
By~\Cref{lem:recursion-depth}, $\phi(T_C, G[C], \bdry T_C) \le \phi(T, P\/, \bdry T) - 1$. By assumption, both $u$ and $v$ are in $T_C$, and $(u,v)$ is a pseudo-edge with respect to $\bdry \cC_C$. Thus, we can apply induction, and conclude that
    $\mathbb{E}[\dist_{\hat{G}}(u,v) \mid \text{Case $\xi_1$ occurs}] \le \dist_G(u,v) + \alpha \cdot \frac{2^i \cdot (\phi(T, P\/, \bdry T) - 1)}{\psi}.$ In other words, 
    \[\mathbb{E}[D \mid \text{Case $\xi_1$ occurs}] = 0.\]

\medskip \noindent \EMPH{Case $\xi_2$}: Vertices $u$ and $v$ are assigned to the same connected component $H$ of $P \setminus \bigcup \cF$; see \Cref{fig:distortion}(b).

\smallskip \noindent  This case is almost the same as the Case~$\xi_1$.
In this case, we recurse on $\textsc{Embed}_{\bbC}(T_H, P\/, \bdry T_H, \bdry \cC_H)$. By~\Cref{lem:recursion-depth}, $\phi(T_H, P\/, \bdry T_H) \le \phi(T, P\/, \bdry T) - 1$. 
Vertices $u$ and $v$ are in $T_H$, and $(u,v)$ is a pseudo-edge with respect to $\bdry \cC_H$ (because it is a pseudo-edge with respect to $\bdry \cC \subseteq \bdry \cC_H$). 
By induction,  
$\mathbb{E}[\dist_{\hat{G}}(u,v) \mid \text{Case $\xi_2$ occurs}] \le \dist_G(u,v) + \alpha \cdot \frac{2^i \cdot (\phi(T, P\/, \bdry T) - 1)}{\psi}$. 
In other words, 
\[\mathbb{E}[D\mid \text{Case $\xi_2$ occurs}] = 0.\]

\medskip \noindent \EMPH{Case $\xi_3$}: Vertices $u$ and $v$ are cut by $\cF$; 
see \Cref{fig:distortion}(c).
    
\smallskip \noindent  We claim that (without loss of generality) \emph{there is some cluster \EMPH{$C_u$} in $\cF$ that contains $u$ (but not $v$).} 
Indeed, let $\eta:V \to V(G/\bdry \cC)$ be the map from $G$ to the contracted graph $G/\bdry \cC$. 
Because $\cF$ conforms $\bdry \cC$, every cluster $C$ in $\cF$ corresponds to a cluster in $G / \bdry \cC$, denoted by $\eta(C)$, whose preimage is $C$.  (That is, every vertex $c$ with $\eta(c) \in \eta(C)$ satisfies $c \in C$.) 
Because $u$ and $v$ are cut by $\cF$, the corresponding vertices $\eta(u)$ and $\eta(v)$ are also cut by the clusters corresponding to $\cF$.
As there is an edge between $\eta(u)$ and $\eta(v)$, they can be cut only if there is some cluster $\eta(C_u)$ from $\cF$ containing either $\eta(u)$ or $\eta(v)$ (without loss of generality, $\eta(u)$). Thus, $C_u$ contains $u$.

The procedure $\textsc{Embed}$ selects an arbitrary terminal $u'$ from $C_u \cap T$ in Step~4, and adds an edge to $\hat{G}$ between $u$ and $u'$, with weight $\dist_G(u',u)$ in Step~5. 
It also recursively calls $\textsc{Embed}_{\bbC}(T_H, P\/, \bdry T_H, \bdry \cC_H)$ to get output $(\hat{G}_H, \hat{\cT}_H)$.
By~\Cref{obs:edge-cut}, the pair $(u', v)$ is a pseudo-edge with respect to $\bdry \cC_H$ which contains ${C_u}$ by definition of $\bdry \cC_H$ in Step 4 of $\textsc{Embed}_{\bbC}(T, P\/, \bdry T, \bdry \cC)$. 
Further, $T_H$ contains both $u'$ and $v$, and by~\Cref{lem:recursion-depth} $\phi(T_H, P\/, \bdry T_H) \le \phi(T, P\/, \bdry T) - 1$.
Finally, \Cref{obs:split} implies that $\bbC$ splits $(u', v)$ at scale $i$.
Thus, we can apply induction to conclude that, in the output graph $\hat{G}_H$, we have 
\[\mathbb{E}[\dist_{\hat{G}_H}(u',v)] \le \dist_G(u',v) + \alpha \cdot \frac{2^i \cdot (\phi(T, P\/, \bdry T) - 1)}{\psi}.\]
We now split into subcases based on $C_u$. 

\begin{itemize}
    \item \EMPH{Case $\xi_{3}[\varnothing]$}: The cluster $C_u$ satisfies $|C_u \cap T| = 1$.
    
    In Case $\xi_3[\varnothing]$, we have $u = u'$. 
    Thus, in this case $\mathbb{E}[\dist_{\hat{G}_H}(u',v) \mid \text{Case $\xi_3[\varnothing]$ occurs}] = \dist_G(u,v) + \alpha \cdot \frac{2^i \cdot (\phi(T, P\/, \bdry T) - 1)}{\psi}$, and so 
    \[\mathbb{E}[D \mid \text{Case $\xi_3[\varnothing]$ occurs}] = 0.\]
    \item  \EMPH{Case $\xi_3[s]$} (see \Cref{fig:distortion}(c)): The cluster $C_u$ satisfies $|C_u \cap T| > 1$ and is at scale $s$. (Note that $s\leq i$ since $(u,v)$ is split at scale $i$.)
    
    In this case, cluster $C_u$ has diameter $2^s$. 
    Thus, we have $\dist_G(u', u) \le 2^s$, and $\dist_G(u', v) \le \dist_G(u', u) + \dist_G(u, v) \le 2^s + \dist_G(u,v)$.
    We conclude:
    \begin{align*}
        \mathbb{E}[\dist_{\hat{G}}(u,v) \mid \text{Case $\xi_3[s]$ occurs}] &\le \dist_G(u, u') + \mathbb{E}[\dist_{\hat{G}_H}(u',v)]\\
        &\le \dist_G(u,v) + 2^{s+1} + \alpha \cdot \frac{2^i \cdot (\phi(T, P\/, \bdry T) - 1)}{\psi}
    \end{align*}
    In other words, 
    \[
    \mathbb{E}[D \mid \text{Case $\xi_3[s]$ occurs}] \le 2^{s+1}.
    \]
\end{itemize}
The expected value $\mathbb{E}[\dist_{\hat{G}}(u,v)]$ is a weighted sum of these cases. We have 
\begin{align*}
    \mathbb{E}[\dist_{\hat{G}}(u,v)]
    &= \sum_{j = 1}^3 \mathbb{E}[\dist_{\hat{G}}(u,v) \mid \text{Case $\xi_j$ occurs}] \cdot \Pr[\text{Case $\xi_j$ occurs}]\\
    &= \dist_G(u,v) + \alpha \cdot \frac{2^i \cdot (\phi(T, P\/, \bdry T) - 1)}{\psi} + \sum_{j = 1}^3 \mathbb{E}[D \mid \text{Case $\xi_j$ occurs}] \cdot \Pr[\text{Case $\xi_j$ occurs}]\\
    &= \dist_G(u,v) + \alpha \cdot \frac{2^i \cdot (\phi(T, P\/, \bdry T) - 1)}{\psi} + \sum_{s = 1}^{i} \mathbb{E}[D \mid \text{Case $\xi_3[s]$ occurs}] \cdot \Pr[\text{Case $\xi_3[s]$ occurs}]\\
    &\le \dist_G(u,v) + \alpha \cdot \frac{2^i \cdot (\phi(T, P\/, \bdry T) - 1)}{\psi} + \sum_{s = 1}^{i} 2^{s+1} \cdot \Pr[\text{Case $\xi_3[s]$ occurs}]
\end{align*}
where the second-to-last line follows because $\mathbb{E}[D]$ is only nonzero in case $\xi_3[s]$. 
(Notice that the sum involves $i$ terms, because by assumption $(u,v)$ is not cut by any clusters with scale larger than $i$.)
Now, observe that for each scale $s$, there are at most two clusters in $\bbC$ at scale $s$ that contain $u$ or $v$. 
(This is because each clustering in $\bbC$ is a partition of vertices.) 
If Case $\xi_3[s]$ occurs, then one of these two clusters was chosen. 
Since $\cF$ is sampled from a $(\tau,\psi)$-stochastic balanced cut, $C_u$ is chosen with probability at most $1/\psi$. 
(This is by [Low probability] property in \Cref{def:stoc-Bcut}; 
one can readily check that if some cluster $C_u$ is chosen in Case $\xi_3[s]$, then $C_u$ is not a singleton in $\bbC$, 
nor is it in $\bdry \cC$ (as by definition every cluster in $\bdry \cC$ contains at most 1 vertex in $T$.)
Thus, for any $s$, Case $\xi_3[s]$ occurs with probability at most $2/\psi$. We have:
\begin{align*}
    \mathbb{E}[\dist_{\hat{G}}(u,v)]
    &\le \dist_G(u,v) + \alpha \cdot \frac{2^i \cdot (\phi(T, P\/, \bdry T) - 1)}{\psi} + \sum_{s = 1}^{i} 2^{s+1} \cdot \frac{2}{\psi}\\
    &\le \dist_G(u,v) + \alpha \cdot \frac{2^i \cdot (\phi(T, P\/, \bdry T) - 1)}{\psi} + 8\cdot \frac{2^{i}}{\psi}\\
    &\le \dist_G(u,v) + \alpha \cdot \frac{2^i \cdot \phi(T, P\/, \bdry T)}{\psi}
\end{align*}
where the last line follows if $\alpha$ is sufficiently large. (Choosing $\alpha = 8$ suffices.)
\end{proof}

\section{Stochastic Balanced Cuts: Proof of Theorem~\ref{thm:balanced-cuts}}
\label{sec:balanced-cuts}

In this section, we prove \Cref{thm:balanced-cuts}. 
A key concept in the construction of a stochastic balanced cut is the notion of \emph{$(a,b,c)$-contraction sequence} introduced by \cite{CLPP23}. Roughly speaking, it represents a contraction of a graph $G$ into a single vertex through $b$ rounds, where in each round, we contract a certain number of vertex-disjoint subgraphs of radius at most $c$; the total number of subgraphs we contract in $b$ rounds is $a$. See \Cref{fig:contract-sequene}(a). 

\begin{definition}[$(a,b,c)$-Contraction Sequence]\label{def:contr-seq} 
For a connected unweighted graph $G$ and integers $a,b,c\geq 0$, an \EMPH{$(a,b,c)$-contraction sequence} consists of:
\begin{itemize}
\item a sequence of graphs $G_0,G_1,\ldots,G_b$, and
\item for every $i \in [1 \,..\, b]$, a collection of pairwise disjoint subgraphs $H_i^1,\ldots,H_i^{a_i}$
of $G_{i-1}$,
\end{itemize}
such that the following conditions hold:
\begin{itemize}
\item $G_0 = G$ and $G_b$ is a one-vertex graph.
\item Each $H_i^j$ is of radius at most $c$.
\item $G_i$ is obtained from $G_{i-1}$ by contracting each of the subgraphs $H_i^j$, $j \in [1 \,..\, a_i]$, to a single vertex.
\item It holds that $\sum_{i=1}^b a_i \leq a$.
\end{itemize}
\end{definition}

\begin{figure}[!t]
\centering
\includegraphics[width=1.0\linewidth]{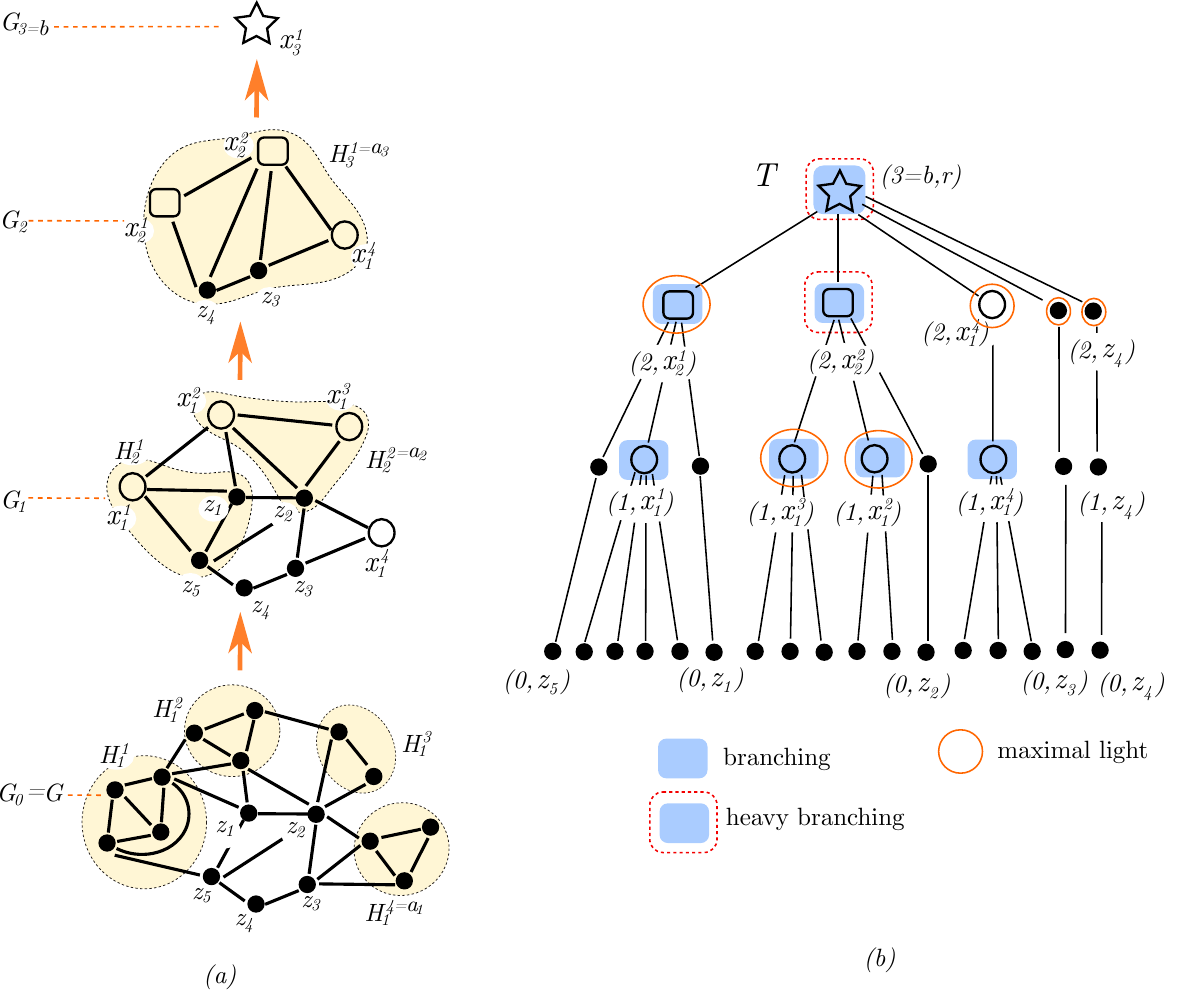}
\caption{(a) An $(a,b,c)$ contraction sequence with $a = 7, b= 3, c=1$. (b) A tree $T$ associated with the contraction sequence in the proof of \Cref{lem:tw}. 
$M$ contains exactly all the nodes at level 2; $(2,x^2_2)$ is minimally heavy and the only node in $M_1$; the rest of nodes at level 2 are in $M_2$.}
\label{fig:contract-sequene}
\end{figure}
\cite{CLPP23} showed that if an apex-minor-free graph $G$ admits an $(a,b,c)$-contraction sequence, then $\tw(G) = O(\sqrt{a}bc)$. 
Their basic idea is to argue that there exists a set $S$ of $a$ vertices such that all other vertices are within hop distance $O(bc)$ from the set, and then apply \Cref{lem:tw:balls}. 
The set $S$ is obtained by picking one vertex per contracted subgraph.   
Here, we improve the treewidth bound by a factor of $\sqrt{b}$ to $O(\sqrt{ab}c)$, which is best possible; see \Cref{rm:contraction-opt}.  
First, by a simple reduction, we can assume that $c = 1$. 
Second, which is also the bulk of the technical details, we show the following structural property: Let $Z$ be any maximal set of vertices in $G$ such that their pairwise (hop) distance in $G$ is $\Omega(b)$\footnote{known as an \emph{$r$-net} with $r = \Omega(b)$}, then $|Z| \leq O(\frac{a}{b})$. 
By \Cref{lem:tw:balls}, we can conclude that $\tw(G) = O(\sqrt{Z}b) = O(\sqrt{ab})$ (for the case $c=1$), which ultimately leads to our improved bound. 

\begin{lemma}\label{lem:tw}
For every apex graph $K$, there exists a constant $\beta_K$ such that if a $K$-minor-free graph $G$ admits an $(a,b,c)$-contraction sequence for some integers $a,b,c \geq 1$, 
   then  $\tw(G)\leq \beta_K \sqrt{ab}c$. 
\end{lemma}

\begin{proof}
\emph{Reduction to the case when $c=1$.~}
Observe first that if a graph $H$ is of radius at most $c$, then it admits a $(c,c,1)$-contraction sequence: 
if $v \in V(H)$ is such that every vertex of $H$ is within distance at most $c$ from $v$, then iteratively contracting all neighbors of $v$ onto $v$ for $c$ rounds
turns $H$ into a one-vertex graph. 
Consequently, if $G$ admits a $(a,b,c)$-contraction sequence, then it also admits a $(ac,bc,1)$-contraction sequence: 
replace every contraction of a graph $H_i^j$ of radius at most $c$ with at most $c$ contractions of radius $1$. 
Since $\beta_K \sqrt{ab} c = \beta_K \cdot 1 \cdot \sqrt{ac \cdot bc}$, it suffices to prove the lemma for the case $c=1$.
That is, in the remainder of the proof we assume that a $K$-minor-free graph $G$ admits a $(a,b,1)$-contraction sequence for some $a,b \geq 1$
and we prove that the treewidth of $G$ is bounded by $\beta_K \sqrt{ab}$ for some constant $\beta_K$ depending only on $K$.

\bigskip
Let $G_0,G_1,\ldots,G_b$, $a_1,\ldots,a_b$, and $H_i^j$ for $i \in [1 \,..\, b]$ and $j\in [1 \,..\, a_i]$ be as in the definition of an $(a,b,1)$-contraction sequence for $G$,
where $G = G_0$ and $V(G_b) = \{r\}$. 
For $i \in [1 \,..\, b]$ and $j\in [1 \,..\, a_i]$, let $\EMPH{$x_i^j$} \in V(G_i)$ be the vertex of $G_i$ that is the result of the contraction of $H_i^j$.
Without loss of generality, we can assume that every graph $H_i^j$ has at least two vertices, and that $a = \sum_{i=1}^b a_i$ and $a_i > 0$ for every $i \in [1 \,..\, b]$
(in particular, $a \geq b$). 
We view the contraction sequence as a rooted tree $T$; see \Cref{fig:contract-sequene}(b). That is, 
we set $V(T) = \bigcup_{i=0}^b \{i\} \times V(G_i)$, indicate $(b,r)$ as the root of $T$, 
and for every $i \in [1 \,..\, b]$ and $v \in V(G_{i-1})$:
\begin{itemize}
    \item If there exists  $j\in [1 \,..\, a_i]$ such that $v \in V(H_i^j)$, then the parent of $(i-1,v)$ is $(i,x_i^j)$.
    \item Otherwise, the parent of $(i-1,v)$ is $(i,v)$. 
\end{itemize}
A node of $T$ is \EMPH{branching} if it has at least two children. Since every $H_i^j$ has at least two vertices, a node $(i,v)$ is branching if and only if
there exists $j\in [1 \,..\, a_i]$ such that $v = x_i^j$. Consequently, there are exactly $a$ branching nodes of $T$.
The \EMPH{weight} of a node $(i,v)$ of $T$, denoted \EMPH{$\mathrm{weight}(i,v)$}, is the number of branching vertices in the subtree of $T$ rooted at $(i,v)$ (including possibly the node $(i,v)$ itself). 
A node is \EMPH{heavy} if it is of weight at least $b$ and \EMPH{light} otherwise. 
Note that all leaves of $T$ are light (as they are of weight $0$) while the root is heavy (as it is of weight $a$, which by assumption is at least $b$).

\medskip
We say that a node is \EMPH{maximal light} if it is light, but its parent is heavy, and \EMPH{minimal heavy} if it is heavy, but all its children are light. 
Note that all children of a minimal heavy node are maximal light, but a parent of a maximal light node is always heavy, but may not be minimal heavy. 
Let \EMPH{$M_1$} be the family of all minimal heavy nodes and \EMPH{$M_2$} be the family of those maximal light nodes, whose parents are not minimal heavy. 
Let $\EMPH{$M$} \coloneqq M_1 \cup M_2$. 
We observe the following.
\begin{equation}\label{eq:M}
\text{On every root-to-leaf path in $T$ there is exactly one node of $M$.}
\end{equation}
Indeed, since all leaves of $T$ are light while the root is heavy, on every root-to-leaf path in $T$ there is exactly light node $(i,v)$
with a heavy parent $(i+1,w)$. Then either $(i+1,w)$ is minimal heavy and is in $M_1$ (but then $(i,v)$ is not in $M_2$) or $(i+1,w)$ is not minimal heavy
(but then $(i,v)$ is in $M_2$). This proves~\eqref{eq:M}.

\medskip
We infer that the sum of weights of elements of $M$ is at most $a$. This in particular implies that
\begin{equation}\label{eq:M1bound}
|M_1| \leq a/b.
\end{equation}
Furthermore, note that every node of $M_1$ is branching, as the weight of a non-branching non-leaf node of $T$ is equal to the weight of its only child.

\medskip
For a node $(i,v)$ of $T$, let \EMPH{$G[(i,v)]$} be the subgraph of $G$ induced by those vertices $w$ where $(0,w)$ is a descendant of $(i,v)$ in $T$.
In other (less formal) words, $G[(i,v)]$ is the subgraph that is contracted onto~$v$ in $G_i$ in the contraction process. 
Note that~\eqref{eq:M} implies that $\Set{\big. V(G[(i,v)]) \mid (i,v) \in M}$ is a partition of~$V$.
We observe the following.
\begin{equation}
\label{eq:trivial-diam}
\text{If $(i,v) \in V(T)$ is of weight $p$, then $\mathrm{diam}(G[(i,v)]) \leq 2p$.}
\end{equation}
Indeed, to see~\eqref{eq:trivial-diam}, observe that a contraction of a subgraph of radius $1$ in a connected graph can decrease the diameter by at most $2$. 
By induction, if a graph $H$ can be turned into a one-vertex graph by a series of $p$ contractions of subgraphs of radius $1$, then the diameter of $H$ is bounded by $2p$.
This proves~\eqref{eq:trivial-diam}.

\medskip
We infer that for every light $(i,v)$, the diameter of $G[(i,v)]$ is bounded by $2(b-1)$. 
We now show a bound on the diameter of $G[(i,v)]$ for minimal heavy $(i,v)$.
\begin{equation}\label{eq:M1-diam}
\text{For every $(i,v) \in M_1$, $\mathrm{diam}(G[(i,v)]) \leq 6b-1$.}
\end{equation}
Let $(i,v)$ be a minimal heavy node. Recall that $(i,v)$ is branching, and thus $v = x_i^j$ for some $j\in [1 \,..\, a_i]$.
Hence, the children of $(i,v)$ in $T$ are $(i-1,w)$ for $w \in V(H_i^j)$. 
Let $w_0 \in H_i^j$ be such that every vertex of $H_i^j$ is $w_0$ or a neighbor of $w_0$. 
As $(i-1,w)$ is light for $w \in V(H_i^j)$, the diameter of $G[(i-1,w)]$ is at most $2(b-1)$. 
We infer that the diameter of $G[(i,v)]$ is bounded by $3 \cdot 2(b-1) + 2 = 6b-1$. 
This proves~\eqref{eq:M1-diam}.

\medskip
Let \EMPH{$G/M$} be the graph obtained from $G$ by contracting, for every $(i,v) \in M$, the graph $G[(i,v)]$ into a single vertex. 
We identify $V(G/M)$ with the set $M$. 
(Recall that {$G / \mathcal{X}$} denote the graph obtained by contracting every cluster in $\mathcal{X}$ into a vertex; here we contract clusters associated with all $(i,v) \in M$.)
The crucial observation is the following.
\begin{equation}\label{eq:H-diam}
\text{For every $(i,v) \in M$ there exists $(i',v') \in M_1$ within distance $b$ in $G/M$.}
\end{equation}
Indeed, by the definition of $M$ and the original $(a,b,1)$ contraction sequence, one can turn $G/M$ into a one-vertex graph by iteratively, in $b$ rounds, contracting some neighborhoods of elements of $M_1$. This proves~\eqref{eq:H-diam}.

\medskip
We fix a threshold \EMPH{$\kappa \coloneqq 9b$}. 
For a vertex $x \in V$ and a node $(i,v)$ of $T$, we say that $(i,v)$ is \EMPH{close} to $x$ if every vertex of $G[(i,v)]$ is within distance at most $\kappa$ from $x$ in $G$. 
Let \EMPH{$\mathrm{Close}(x)$} be the set of nodes close to $x$.
We now prove the following critical property.
\begin{equation}\label{eq:crux}
\text{For every vertex $x$, $\sum_{(i,v) \in M \cap \mathrm{Close}(x)} \mathrm{weight}(i,v) \geq b$.}
\end{equation}
To see~\eqref{eq:crux}, first note that if $(i,v) \in M_1$ is close to $x$, then already $\mathrm{weight}(i,v)$ contributes at least $b$
to the sum of~\eqref{eq:crux}. In the remaining case, we observe that, by~\eqref{eq:M1-diam} and the fact that $\Set{\big. V(G[(i,v)]) \mid (i,v) \in M }$ is a partition of $V$, every vertex within distance at most $\kappa-6b$ of $x$ belongs to $G[(i,v)]$ for some $(i,v) \in M_2$.

\begin{figure}[!htb]
\centering
\includegraphics[width=0.4\linewidth]{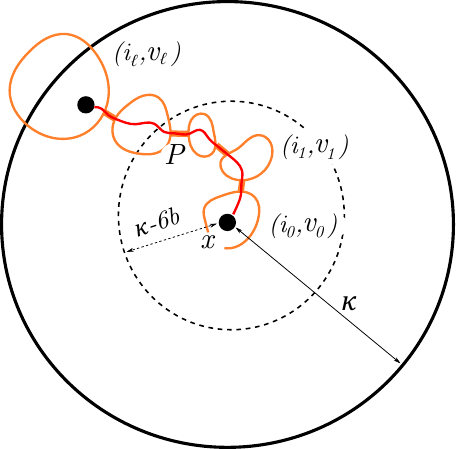}
\caption{Path $P$ from $x$ to a vertex in $G[(i_\ell, v_\ell)]$.}
\label{fig:patX}
\end{figure}

Let $(i_0,v_0) \in M$ be such that $x \in V(G[(i_0,v_0)])$ and let $Q$ be a shortest path (in hops) in $G/M$ from $(i_0,v_0)$ to a node of $M$
that is not close to $x$. 
By~\eqref{eq:H-diam}, the length of $Q$ is $\ell\leq b$, as every node of $M_1$ is not close to $x$.
Let $(i_0,v_0), (i_1,v_1), \ldots, (i_\ell,v_\ell)$ be consecutive vertices of $Q$ where $(i_\ell,v_\ell) \in M_1$ and $(i_j,v_j) \in M_2$ for $j \in [0 \,..\, \ell-1]$. 
By the definition of $G/M$, there is a path $P$ from $x$ to a vertex of $G[(i_\ell,v_\ell)]$ that goes via $G[(i_0,v_0)]$, $G[(i_1,v_1)]$, \ldots, $G[(i_{\ell-1}, v_{\ell-1})]$. See \Cref{fig:patX}.  
By~\eqref{eq:trivial-diam}, the length of $P$ is bounded~by
\[ 
\sum_{j=0}^{\ell-1} 1 + 2 \mathrm{weight}(i_j,v_j) \leq b + 2\sum_{j=0}^{\ell-1} \mathrm{weight}(i_j,v_j). 
\]
On the other hand, the definition of close, together with~\eqref{eq:trivial-diam} and~\eqref{eq:M1-diam} ensure that the length of $P$ is at least $\kappa-6b$. 
We infer that
\[ \sum_{j=0}^{\ell-1} \mathrm{weight}(i_j,v_j) \geq b. \]
As every node $(i_j,v_j)$ for $j \in [0 \,..\, \ell-1]$ is close to $x$, this finishes the proof of~\eqref{eq:crux}.

\medskip
Let \EMPH{$Z$} be a maximal family of vertices of $G$ within pairwise distance more than $2\kappa = 18b$. 
Since the total weight of all nodes of $M$ is at most $a$, we infer from~\eqref{eq:H-diam} that $|Z| \leq a/b$. 
By the definition of $Z$, every node of $G$ is within distance at most $2\kappa$ from an element of $Z$.
Lemma~\ref{lem:tw:balls} implies that the treewidth of $G$ is bounded by
\[ \alpha_K \cdot 2\kappa \cdot \sqrt{|Z|} \leq \alpha_K \cdot 18b \cdot \sqrt{\frac{a}{b}} \leq 18\alpha_K \sqrt{ab}. \]
This finishes the proof of Lemma~\ref{lem:tw} with $\beta_K := 18\alpha_K$, where $\alpha_K$ comes from Lemma~\ref{lem:tw:balls}.
\end{proof}

\begin{remark}\label{rm:contraction-opt}  The bound of Lemma~\ref{lem:tw} is  asymptotically optimal. Pick integers $p,q$ with $2^q \gg p \gg q$ and consider a $pq \times pq$ grid. Let $Z$ be a set of $p^2$ equidistributed vertices of the grid, so that every vertex of the grid is within distance $2q$ from a vertex of $Z$. Iteratively, for $2q$ rounds, contract the neighborhood of $Z$ onto $Z$. This contracts the grid into
a grid of sidelength $p$. Then, within $\Oh(\log p)$ rounds, contract this grid onto a single vertex, by just taking every second vertex of the grid and contracting its neighborhood onto it. 
This gives a $(\Oh(p^2(q + \log p)), \Oh(q + \log p), 1)$-contraction sequence for the original grid of sidelength $pq$. If $q > \log p$, this is a matching lower bound to Lemma~\ref{lem:tw}.
\end{remark}

We are now ready to prove \Cref{thm:balanced-cuts}. 

\begin{proof}[of \Cref{thm:balanced-cuts}] It suffices to show that $(G,\bbC)$ has a $(\tau,\psi)$-stochastic balanced cut with respect to $(\omega_V,\cX)$ since the only property of $G$ that we will use in our construction is apex-minor-free. 
That means the same construction applies to construct a $(\tau,\psi)$-stochastic balanced cut with respect to $(\omega_{C}, \cX)$ for $(G[C], \bbC_{\downarrow C})$ since $G[C]$ is also apex-minor-free.
We will construct a set of balanced cut $\bbF$ such that:
\begin{itemize}
    \item[(i)] $|\bbF| = \psi$.
    \item[(ii)] every cut $\cF \in \bbF$ conforms $\cX$ and contains at most $\tau$ clusters in $\bbC$. 
    \item[(iii)] every \emph{non-singleton} cluster $C \in \bbC \setminus \cX$  appears in (some cut in) $\bbF$ at most once. 
\end{itemize}
Then, the distribution $\frakF$ is simply a uniform distribution over cuts in $\bbF$. That is, one samples a cut $\cF$ with probability $1/|\bbF| = 1/\psi$. As each non-singleton cluster $C \in \bbC \setminus \cX$  appears in $\bbF$ at most once, the probability that $C$ is contained in a sampled cut $\cF$ is at most $1/\psi$. Thus, the existence of $\bbF$ implies \Cref{thm:balanced-cuts}. 

Henceforth, we focus on constructing $\bbF$. 
The construction is the same as in~\cite{CLPP23}, but uses \Cref{lem:tw} instead. We include the details tailored to our notation for completeness. 
Set the threshold for number of clusters in any $\cF$ to be $\EMPH{$\tau$} \coloneqq {\beta'} \cdot h^2 \log\Phi \cdot \psi$ for a constant ${\beta'}$ chosen later. The algorithm is greedy: starting from $\bbF = \varnothing$, if $|\bbF| < \psi$, we will show below that we can add one more cut $\cF$ of size at most $\tau$ to $\bbF$. 
(The size of a cut is the number of clusters within.) 
Thus, the algorithm will terminate when $\bbF$ has exactly $\psi$ cuts, and every cut has a size at most~$\tau$.

Now, we construct a balanced cut $\cF$ of size at most $\tau$. 
Let \EMPH{$|\bbF|$} be the number of cuts currently in~$\bbF$. 
Note that $|\bbF| < \psi $, implying that:
\begin{equation}\label{eq:psi-kappa}
    |\bbF|  < \frac{\tau}{{\beta'} \cdot h^2 \log\Phi}.
\end{equation}

View the clustering chain $\bbC = (\cC_0,\cC_1,\ldots, \cC_k)$ where $k = \lceil \log\Phi \rceil$ as a tree with the root corresponding to a single cluster $V$ (in $\cC_k$) with leaves being singletons in $\cC_0$. We mark the root of $\bbC$, and all clusters in $\bbF\setminus \cX$ as \EMPH{unavailable}. All other clusters are marked \EMPH{available}; these include (not only) clusters in $\cX$ and singleton clusters in $\cC_0$.  We say that a cluster $C$ is \EMPH{maximally available} if $C$ is available and its parent in $\bbC$ is unavailable. 

Let \EMPH{$\cS$} be the set of maximally available clusters. 
Observe that $\cS$ is a partition $V$.  
Let \EMPH{$\check{G}$} be the graph obtained from $G$ by contracting every cluster in $\cS$ into a single vertex. For each vertex $\check{v}\in \check{G}$ corresponding to a cluster $C_v$, we assign a weight $\check{\omega}(\check{v}) \coloneqq \sum_{u\in C_v} \omega(u)$. The key observation is that:
\begin{equation}\label{obs:contraction-sequence}
    \text{$\check{G}$ admits a $(\tau\cdot|\bbF|+1,k,h)$-contraction sequence.}
\end{equation}

\noindent To see~\eqref{obs:contraction-sequence}: Starting from level $1$ (corresponding to $\cC_1$), for each level $i\in [1 \,..\, k]$, contract all unavailable clusters at level $i$ into a single vertex. (Think of each contracted unavailable cluster $C$ at level $i$ as the graph $G[C]/\cC_{i-1}[C]$ to align with the definition of the contraction sequence.) Clearly, the number of contraction rounds is $k$, and the total number of contracted clusters is at most $\tau\cdot|\bbF|+1$, which is an upper bound on the number of unavailable clusters. (The $+1$ is for including the root cluster, which is always unavailable.)   Since the hop bound if $\bbC$ is $h$, the radius of each contracted graph is at most $h$.

\medskip
By \eqref{obs:contraction-sequence} and \Cref{lem:tw}, there is a constant $\beta$ that only depends on the size of the minor such that:
\begin{align*}
    \tw(\check{G} )
    &\leq \beta \cdot \sqrt{(\tau\cdot|\bbF|+1)k} \cdot h
    = \beta \cdot \sqrt{(\tau\cdot|\bbF|+1)\log\Phi } \cdot h\\
    &\leq  \beta \cdot \sqrt{2\tau\cdot|\bbF|\log\Phi} \cdot h\\
    &\leq \beta \cdot \sqrt{2\tau \cdot \log\Phi \cdot \frac{\tau}{{\beta'}\cdot h^2 \log\Phi} } \cdot h \qquad \text{(by \Cref{eq:psi-kappa})}\\
    &= \frac{\sqrt{2}\cdot \beta}{\sqrt{{\beta'}}} \cdot \tau \leq \tau
\end{align*}
by choosing a sufficiently large constant ${\beta'} \ge 2\beta^2$ (that only depends on the size of the minor). 
Thus, $\check{G}$ admits a balanced separator $\check{S}$ with respect to $\check{\omega}$ of size at most $\tau$. 
Then $\cF = \{C_v: \check{v}\in \check{S}\}$ is a balanced cut of size at most $\tau$.  

Clearly, $\cF$ conforms $\cX$ since clusters in $\cX$ are always marked available, and clusters in $\cF$ are maximally available clusters. Thus, property (ii) holds.  For every cluster $C\in \cF$, since $C$ is available, then $C\in \cX$ or $C$ is a singleton, or $C$ has not been added to any other cut in $\bbF$. Thus, property (iii) holds.  
\end{proof}

\section{Separating Distribution of Clustering Chains: Proof of Theorem~\ref{thm:beta-separating}}

Recall the definition of a \emph{buffered cop decomposition}, introduced in \cite{CCLMST24}. 
For any graph $G$, a \EMPH{supernode $\eta$} with \EMPH{skeleton $T_\eta$} and \EMPH{radius $\Delta$} is an induced subgraph of $G$ containing a tree $T_\eta$ where every vertex in $\eta$ is within distance $\Delta$ of $T_\eta$, where distance is measured with respect to the induced shortest-path metric of $\eta$. A \EMPH{buffered cop decomposition} for $G$ is a partition of $G$ into vertex-disjoint supernodes, together with a tree \EMPH{$\cT$} called the \EMPH{partition tree} whose nodes are the supernodes of $G$. For any supernode $\eta$, the \EMPH{domain $\dom(\eta)$} is the subgraph induced by the union of all vertices in supernodes in the subtree of $\cT$ rooted at $\eta$.
\begin{definition}
    A \EMPH{$(\Delta, \gamma, w)$-buffered cop decomposition} for $G$ is a buffered cop decomposition $\cT$ for $G$ that satisfies the following properties:
    \begin{itemize}
        \item \textnormal{[Supernode radius.]} Each supernode $\eta$ has radius at most $\Delta$.
        
        \item \textnormal{[Shortest-path skeleton.]} For every supernode $\eta$, the skeleton $T_\eta$ is an SSSP tree in $\dom(\eta)$ with at most $w-1$ leaves (not counting the root). 
        Further, let \EMPH{$\cA_\eta$} denote the set of ancestor supernodes $\eta'$ of $\eta$ such that there is an edge between $\eta'$ and $\dom(\eta)$ in $G$. 
        Then $\lvert \cA_\eta \rvert \le w-1$, 
        and there is an edge between $T_\eta$ and $\eta'$ in $G$ for each $\eta' \in \cA_\eta$.%
        \footnote{The second half of this property was not explicitly stated by \cite{CCLMST24}; rather, it was stated by \cite{Fil24}, who observed that the construction of \cite{CCLMST24} also satisfies this property.}
        
        \item \textnormal{[Supernode buffer.]} Let $\eta$ be a supernode, and let $X$ be another supernode that is an ancestor of $\eta$. Then either $\eta$ and $X$ are adjacent in $G$, or for every vertex $v$ in $\dom(\eta)$, we have $\dist_{\dom(X)}(v, X) \ge \gamma$.%
        \footnote{If $A$ and $B$ are two sets, we use $\dist(A,B)$ to denote $\min_{a \in A, b \in B} \dist(a,b)$. 
        When $A$ is a singleton set $\set{a}$ we write $\dist(a,B)$ instead.}
        
        \item \textnormal{[Tree decomposition.]} Define the \EMPH{expansion $\hat \cT$} to be a tree isomorphic to $\cT$, such that every supernode $\eta$ in $\cT$ corresponds to a node $B_\eta$ (called the \EMPH{bag at $\eta$}) that contains (all vertices in) the union of $\eta$ and all supernodes in $\cA_\eta$. The expansion $\hat \cT$ is a tree decomposition for $G$, and every bag of $\hat \cT$ is the union of at most $w$ supernodes.
    \end{itemize}
\end{definition}

We say that an edge $e$ is \EMPH{cut} by a buffered cop decomposition $\cT$ if the endpoints of $e$ belong to different supernodes in $\cT$. The majority of this section is devoted to proving the following theorem.

\begin{theorem}
\label{thm:random-cop}
    Let $G$ be a $K_r$-minor-free graph, and let $\Delta$ be a positive number. 
    There exist a constant~$\beta$ such that there is a distribution $\mathbb{T}$ of $(4 \Delta, \Delta/r, r-1)$-buffered cop decompositions of $G$, such that for any edge $e$, $\Pr_{\cT \sim \mathbb{T}}[\text{$e$ is cut by $\cT$}] \leq \beta \cdot \len{e}/\Delta$.
\end{theorem}

\cite{CCLMST24} also construct a \emph{shortcut partition} (a type of partition first introduced by \cite{CCLMST23a} as a weaker version of the scattering partition of \cite{filtser2024scattering}). We recall the definition here.
\begin{definition}
    An \EMPH{$(\e, h)$-shortcut partition} of $G$ is a clustering $\cC = \set{C_1, \ldots, C_m}$ of $G$ such that:
    \begin{itemize}
        \item \textnormal{[Diameter.]} the strong diameter of each cluster $C_i$ is at most $\e \cdot \diam (G)$;
        \item \textnormal{[Low-hop.]} for any two vertices $u$ and $v$ in $G$, there is a path $\check \pi$ in the graph $G / \cC$ (obtained by contracting every cluster $C_i$ into a single vertex) between the two clusters containing $u$ and $v$, such that 
        (1) $\check \pi$ has hop length at most $\e h \cdot \lceil \frac{\dist_G(u,v)}{\e \cdot \diam(G) }\rceil$, and 
        (2) there is a shortest path $\pi$ in $G$ between $u$ and $v$, such that every cluster on $\check \pi$ has nontrivial intersection with $\pi$.
    \end{itemize}
\end{definition}

We remark that $\cC$ is an $(\e, h)$-shortcut partition, then the graph $G / \cC$ has hop-diameter at most $h + 1$ (see the proof of Theorem~\ref{thm:beta-separating} below). \cite{CCLMST24} used buffered cop decomposition to construct shortcut partition for minor-free graphs. In particular, they showed that for any $\e < 1$, $K_r$-minor-free graphs admit $(\e, O_r(1/\e))$-shortcut partition. We give a random version of their algorithm, using \Cref{thm:random-cop} as a subroutine, to prove the following lemma.

\begin{lemma}
\label{lem:random-shortcut}
Let $G$ be a $K_r$-minor-free graph, and let $\e$ be some fixed number in $(0, 1)$. There exist constants $h, \beta = O_r(1)$ such that there is a distribution $\mathbb{C}$ of $(\e, h)$-shortcut partitions of $G$ such that for each edge $e$, $\Pr_{\cC \sim \mathbb{C}}[\text{$e$ is cut by $\cC$}] \le \beta \cdot \frac{\norm{e}}{\e \cdot \diam(G)}$.
\end{lemma}

Below, we prove Theorem~\ref{thm:beta-separating} from Lemma~\ref{lem:random-shortcut}. In Section~\ref{SS:cop-construction}, we give a randomized construction for buffered cop decomposition (and prove that the output is in fact a buffered cop decomposition). In Section~\ref{SS:cop-probability}, we upper-bound the probability an edge $e$ is cut by the buffered cop decomposition; the proof of a crucial lemma in this analysis is deferred
to Section~\ref{SS:cop-threateners}. This completes the proof of \Cref{thm:random-cop}. In Section 6.5, we give a randomized construction of shortcut partition, proving Lemma~\ref{lem:random-shortcut}.

\begin{proof}[of Theorem~\ref{thm:beta-separating}]
We will build the $\beta$-separating distribution $\frakC$ of clustering chains level-by-level, from top to bottom.
Start with the original graph $G$
and level $i \gets \lceil \log (\diam(G)) \rceil$.
At each level $i$ we deal with a graph $H_i$ with diameter upper bounded by $\Delta_i \coloneqq 2^i$.
Set $\e = \frac{1}{2}$, and
sample a shortcut partition $\cC_i$ from the distribution of $(\e, h)$-shortcut partitions for $G$ promised by Lemma~\ref{lem:random-shortcut}, where $h$ is some constant depending on the size of the minor excluded by $G$.
The decomposition $\cC_i$ partitions $H_i$ into vertex-disjoint clusters each of diameter $\diam(H_i)/2 \le 2^{i-1}$.
Further, the graph $H_i/\cC_i$ has hop-diameter at most $h$; indeed, any two vertices $u$ and $v$ in $H_i$ satisfy $\dist_{H_i}(u,v) \le \diam(H_i)$, so the [low-hop] property of shortcut partition guarantees that the cluster containing $u$ and $v$ are connected in $H_i / \cC_i$ by a path of hop-length at most $\e h \cdot \lceil \frac{\diam(G)}{\e \cdot \diam(G)} \rceil \le \e h \cdot \lceil 1/\e \rceil \le h$.
We recursively build a $\beta$-separating distribution of clustering chains for each of the clusters, now at level $i-1$, until we reach level $0$ where every cluster is a single vertex.

The resulting distribution $\frakC$ of clustering chains clearly satisfies all properties of clustering chains, and each clustering chain in the support has hop bound $h = O(1)$. 
To demonstrate that the distribution $\frakC$ is $\beta$-separating for some $\beta = O(1)$, we consider an arbitrary edge $e$ and prove $\Pr_{\mathbb{C} \sim \frakC}[\text{$e$ is cut by $\cC_i \in \mathbb{C}$]} \le \beta \cdot \norm{e}/2^i$ by induction on $i$, starting with the highest level ($i = \lceil \log \diam(G) \rceil$).
Indeed, if $e$ is cut by the the $i$-th level of the clustering chain $\cC_i$, then either (1) $e$ was cut by $\cC_{i+1}$, or (2) $e$ is contained in some cluster $C \in \cC_{i+1}$ but $e$ is cut by the shortcut partition of $C$ at level $i$. 
Case (1) occurs with probability at most $\beta \cdot \norm{e} / 2^{i+1}$, by induction. 
Case (2) occurs with probability at most $\beta' \cdot \norm{e}/(\e \cdot 2^{i+1}) = \beta' \cdot \norm{e}/2^i$ for some $\beta' = O(1)$ promised by Lemma~\ref{lem:random-shortcut}. 
By choosing $\beta \coloneqq 2 \beta'$ and applying a union bound, we conclude $e$ is cut with probability at most $\beta \cdot \norm{e}/2^i$.
\end{proof}

\subsection{Construction of stochastic buffered cop decomposition}
\label{SS:cop-construction}

We slightly modify the construction of buffered cop decomposition from \cite{CCLMST24}. 
Before giving the full algorithm, we recall some terminology.  Throughout the algorithm, we maintain a global variable \EMPH{$\cS$}, a set of supernodes (which changes over the course of the algorithm). 
We say a subgraph $H$ \EMPH{sees} a supernode $X$ in $\cS$ if (1) $X$ is disjoint from $H$ and (2) there exists some \EMPH{witness vertex $v_X$} in $H$ that is adjacent (in $G$) to some vertex in $X$. 
For any subgraph $H$, let \EMPH{$\cS_{|H}$} denote the set of supernodes in $\cS$ that $H$ sees. 

Every supernode $\eta$ in $\cS$ 
is associated with a subgraph of $G$ called the \EMPH{initial domain} and denoted \EMPH{$\initdom(\eta)$}, as well as a set of supernodes in $\cS$ that $\initdom(\eta)$ sees, denoted \EMPH{$\cS_{|\eta}$}.
We would like to allow the supernodes in $\cS_{|\eta}$ \emph{as induced subgraphs} to grow over the course of the algorithm as $\cS$ changes, but the \emph{collection of} supernodes will be fixed even when $\initdom(\eta)$ later intersects some (grown) supernode $X$ in $\cS_{|\eta}$ and thus technically no longer ``sees'' $X$ because of the disjointness requirement is violated.
At any point in the algorithm, with supernode assignment $\cS$, the subgraph \EMPH{$\dom_{\cS}(\eta)$} is defined to be $\initdom(\eta) \setminus \bigcup \cS_{|\eta}$.%
\footnote{We remark that our definition of $\dom_{\cS}(X)$ makes specific of the definition given in \cite{CCLMST24}; their definition is somewhat ambiguous. 
All the proofs of \cite{CCLMST24} also work with our clarified definition of $\dom_{\cS}(\eta)$.}
It will hold that at the end of the algorithm, $\dom_{\cS}(\eta) = \dom(\eta)$.

Our new algorithm is described in Figure~\ref{fig:cop-decomposition} below; the lines that differ from \cite{CCLMST24} are in red. The parameters $\Delta$ and $r$ in the algorithm are fixed throughout execution of the algorithm.
For intuition behind the algorithm, refer to \cite{CCLMST24}. 
To construct a buffered cop decomposition for a graph $G$, we run $\textsc{BuildTree}(\varnothing, G)$.
If a vertex $v$ is assigned to a supernode $X'$ in Step 2 of a call $C \coloneqq \textsc{GrowBuffer}(\cS, \cX, H)$, then we say that $v$ is \EMPH{assigned during $C$}, and $X'$ is \EMPH{expanded during $C$}.
If the $C$ selects a supernode $X \in \cX$ in Step~1, we say that $C$ \EMPH{processes} $X$.

\begin{figure}[ph!]
\centering
\begin{tcolorbox}
\paragraph{$\textsc{BuildTree}(\cS, H)$:} Input a set of supernodes $\cS$, and a subgraph $H$ of $G$ consisting of vertices unassigned to a supernode in $\cS$. Returns a tree of supernodes.
\begin{enumerate}
    \item \emph{Initialize a new supernode $\eta$.}
    
    Let $v$ be an arbitrary vertex in $H$.
    Let $T_\eta$ be an SSSP tree in $H$, connecting $v$ to a witness vertex for every supernode in $\cS_{|H}$.
    Initialize $\eta \coloneqq T_\eta$ to be a new supernode with skeleton $T_\eta$, and add $\eta$ to $\cS$.
    {\textcolor{BrickRed}{Let $\alpha \sim \mathrm{Unif}[0,1]$. 
    For every vertex $v'$ satisfying $\dist_{H}(v', T_\eta) \le \alpha \cdot \Delta/r$, assign $v'$ to supernode~$\eta$.}}
    Initialize tree $\cT$ with root $\eta$.

    \item \emph{Assign vertices to existing supernodes, to guarantee the supernode buffer property.}

    For each connected component $H'$ of $H \setminus \eta$: let $\cX$ be a list of every supernode seen by $H$ but not $H'$, and call $\textsc{GrowBuffer}(\cS, \cX, H')$. (The \textsc{GrowBuffer} procedure modifies the global variable $\cS$, assigning some vertices in $H'$ to supernodes in $\cS$).

    \item \emph{Recurse.}

    For each connected component $H'$ of $H \setminus \bigcup \cS$: let $\cT'$ be the output of $\textsc{BuildTree}(\cS,H')$, and attach $\cT'$ as a child to the root of $\cT$.
\end{enumerate}
\end{tcolorbox}

\begin{tcolorbox}
\paragraph{$\textsc{GrowBuffer}(\cS, \cX, H)$:} Input a set of supernodes $\cS$, 
a subgraph $H$ of $G$ consisting of vertices unassigned to any supernode in $\cS$, and a list of supernodes $\cX$ in $\cS$ that are not seen by $H$. 
This procedure modifies the global variable $\cS$ by assigning some vertices in $H$ to $\cS$.
\begin{enumerate}
    \item \emph{Grow a buffer around some supernode in $\cX$.}

    Let $X$ be an arbitrary supernode in $\cX$. (If $\cX$ is empty, do nothing and return.)
    Let $\bdry H_{\downarrow X}$ be the set of vertices in $G \setminus H $ that are (1) adjacent to $H$, and (2) in $\dom_{\cS}(X)$.
    {\textcolor{BrickRed}{Let $\alpha \sim \mathrm{Unif}[1,2]$}}, and let $\cN H_X$ be the set of vertices in $H$ such that $\dist_{\dom_{\cS}(X)}(v, \bdry H_{\downarrow X}) \le \text{\textcolor{BrickRed}{$\alpha$}} \cdot \Delta/r$

    \item \emph{Assign the vertices in the buffer to existing supernodes.}

    {\textcolor{BrickRed}{For each supernode $X'$ seen by $H$, choose a random value $\alpha_{X'} \sim \mathrm{Unif}[0,\Delta/r]$.}} For each vertex $v$ in $\cN H_X$, assign $v$ to the supernode $X'$ that minimizes $\dist_{\dom_{\cS}(X)}(v, X' \cap \bdry H_{\downarrow X}) + \textcolor{BrickRed}{\alpha_{X'}}$.
    Update $\cS$ with these assignments. (Note that this changes the supernode $X'$, and may also $\dom_{\cS}(\eta)$ for any supernodes $\eta$ initialized after $X'$.)

    \item \emph{Update $\cX$ to include any newly ``cut-off'' supernodes, and recurse.}

    For each connected component $H'$ of $H \setminus \bigcup \cS$: initialize $\cX'$ to be $\cX \setminus \set{X}$, then add to $\cX'$ all the supernodes in $\cS$ that are seen by $H$ but not $H'$, and then call $\textsc{GrowBuffer}(\cS, \cX', H')$.
\end{enumerate}
\end{tcolorbox}
\caption{The \textsc{BuildTree} and \textsc{GrowBuffer} algorithms.}
\label{fig:cop-decomposition}
\end{figure}

\medskip\noindent
To summarize, we have three modifications to the algorithm of~\cite{CCLMST24}:
\begin{enumerate}
    \item When initializing a supernode $\eta$ in Step 1 of $\textsc{BuildTree}$, we immediately grow it by some random radius in $[0, \Delta/r]$. (\cite{CCLMST24} omits this initial growth.)
    \item When growing a buffer ($\cN H_X$) around old supernodes in Step 1 of $\textsc{GrowBuffer}$, we choose the buffer to have a random radius in $[\Delta/r, 2 \Delta/r]$. (In \cite{CCLMST24}, the buffer had fixed radius $\Delta/r$.)
    \item When assigning each vertex $v$ in the buffer to an existing supernode, in Step 2 of \textsc{GrowBuffer}, we assign $v$ to the closest existing supernode after doing some small random perturbation. (In \cite{CCLMST24}, $v$ was assigned deterministically to the closest existing supernode.)
\end{enumerate}

\cite{CCLMST24} show that, if graph $G$ excludes a $K_r$-minor, then their procedure outputs a $(\Delta, \Delta/r, r-1)$-buffered cop decomposition. Their proof carries over almost verbatim for our stochastic algorithm, with slightly worse constants. We sketch the differences below.
\begin{lemma}
\label{lem:stochastic-cop-decomposition}
    If $G$ is a graph excluding a $K_r$-minor, then for any $\Delta > 0$, the stochastic procedure $\textsc{BuildTree}(\varnothing, G)$ outputs a tree $\cT$ that is a $(4 \Delta, \Delta/r, r-1)$-buffered cop decomposition.
\end{lemma}
\begin{proof}
The claims and proofs in Section 3.2 of \cite{CCLMST24} (``Basic Properties'') carry over nearly verbatim; $\cT$ satisfies the [shortest-path skeleton] property and the [tree decomposition] property. 
The only change is in the proof that every supernode is a connected subgraph (Claim 3.4(1)), where we must replace references to ``closest supernode in $\bdry H_{\downarrow X}$ to $v$'' with ``the supernode $X'$ minimizing $\dist_{\dom_{\cS}(X)}(v, X' \cap \bdry H_{\downarrow X}) + \alpha_{X'}$''. 
The other proofs of Section 3.2 only rely on the definition of the skeleton $T_\eta$ in $\textsc{BuildTree}$, and on the fact that calls to $\textsc{BuildTree}(\cdot, H)$ or $\textsc{GrowBuffer}(\cdot, \cdot, H)$ are made only on subgraphs $H$ that are maximal connected components of unassigned vertices.

The \cite{CCLMST24} proof of the [supernode buffer] property relies on the fact that, during any call to $\textsc{GrowBuffer}$, the set $\cN H_{X}$ includes every vertex within distance $\Delta/r$ of $\bdry H_{\downarrow X}$. As Step 2 of our modified \textsc{GrowBuffer} chooses $\alpha \ge 1$, we also have this property. The rest of their proof carries over verbatim.

To adapt the \cite{CCLMST24} proof of the [supernode radius] property, we need to make two changes. 
First, in the \cite{CCLMST24} algorithm, each supernode $\eta$ was initialized with radius $0$, and expands only when $\textsc{GrowBuffer}$ is called. 
In our modified algorithm, supernode $\eta$ is initialized with radius up to $\Delta/r$; thus, our final radius bound is weaker by $+\Delta/r$ than the bound of \cite{CCLMST24}. Secondly, we need to modify Claim 3.11 of \cite{CCLMST24}, which states that each time a supernode $\eta$ is expanded by some call to $\textsc{GrowBuffer}$, the radius of $\eta$ increases by at most $\Delta/r$. Because of our modifications in Step 1 and 2 of $\textsc{GrowBuffer}$, we instead prove that the radius increases by at most $3 \Delta/r$; thus, our final radius bound is weaker by a factor 3 than the bound of \cite{CCLMST24}.

\begin{quote}
    Suppose that $v$ is assigned to a supernode $\eta$ during a call $C \coloneqq \textsc{GrowBuffer}(\cS, \cX, H)$. Let $X$ denote the supernode processed during $C$, and let $\bdry H_{\downarrow X}$ denote the boundary vertices. Let $\tilde v$ be the closest vertex in $\bdry H_{\downarrow X} \cap \eta$ to $v$ (with respect to $\dom_{\cS}(X)$). Then $\dist_{\eta}(v, \tilde v) \le 3\Delta/r$ (with respect to the final~$\eta$).
\end{quote}

    \noindent To prove the claim,
    let $\cN H_X$ denote the set of points assigned during $C$. Let $P$ be a shortest path between $v$ and $\tilde v$ in $\dom_{\cS}(X)$. 
    Every vertex in $P$ (other than $\tilde v$) is in $\cN H_X$. Because we assign every vertex $v'$ in $\cN H_X$ to the supernode $X'$ minimizing $\dist_{\dom_{\cS}(X)}(v', X' \cap \bdry H_{\downarrow X}) + \alpha_{X'}$, every vertex in $P$ is assigned to $\eta$. 
    Furthermore, we claim $P$ has length at most $3 \Delta/r$. 
    Indeed, $v$ is within distance $2 \Delta/r$ of $\bdry H_{\downarrow X}$ (as $\alpha \le 2$ in Step 1 of \textsc{GrowBuffer}), so there is some supernode $X'$ with $\dist_{\dom_{\cS}(X)}(v, X' \cap \bdry H_{\downarrow X}) \le 2 \Delta/r$. 
    As $\alpha_{X'} \le \Delta/r$, we have $\dist_{\dom_{\cS}(X)}(v, X') + \alpha_{X'} \le 3 \Delta/r$. 
    By choice of $\eta$, $\dist_{\dom_{\cS}(X)}(v, \eta) + \alpha_{\eta} \le 3 \Delta/r$, and thus $\dist_{\dom_{\cS}(X)}(v, \eta) \le 3 \Delta/r$. 
    We conclude that $\norm{P} \le 3 \Delta/r$, and so $\dist_{\eta}(v, \tilde v) \le 3\Delta/r$.

\medskip
The rest of the [supernode radius] proof from \cite{CCLMST24} applies for our modified algorithm. Because of the two modifications described above, we end up with a radius bound of $3 \Delta + \Delta/r \le 4 \Delta$. This concludes the proof that $\cT$ is a $(4 \Delta, \Delta/r, r-1)$-buffered cop decomposition.
\end{proof}

\subsection{Bounding the cut probability}
\label{SS:cop-probability}

Let $e = (u,v)$ be an edge in $G$. Let $C_i$ be a random variable that denotes the $i$-th call, whether it is to $\textsc{BuildTree}$ or $\textsc{GrowBuffer}$, that is made during the execution of the algorithm.
For any $i$, we define three events below. Observe that if $u$ and $v$ are in different supernodes, then one of these events occurs.
\begin{itemize}
    \item \EMPH{$\xi^i_{\rm build}$}: call $C_i$ is a call $\textsc{BuildTree}(\cS, H)$, such that both $u, v \in H$, and exactly one of $u$ and $v$ is assigned to a supernode during $C$.
    \item \EMPH{$\xi^i_{\rm buffer}$}: call $C_i$ is a call $\textsc{GrowBuffer}(\cS, \cX, H)$, such that both $u, v \in H$, and exactly one of $u$ and $v$ is in $\cN H_X$.
    \item \EMPH{$\xi^i_{\rm split}$}: call $C_i$ is a call $\textsc{GrowBuffer}(\cS, \cX, H)$, such that $u, v \in \cN H_X$, but $u$ and $v$ are assigned to different supernodes during step 3 of $\textsc{GrowBuffer}$.
\end{itemize}
Let \EMPH{$\xi_i$} be the event that either $\xi^i_{\rm build}$, $\xi^i_{\rm buffer}$, or $\xi^i_{\rm split}$ occurs.  
If $\xi_i$ occurs, we say that \EMPH{call $C_i$ cuts $e$}.
For each $i$, we define an indicator random variable (1 or 0) that indicates whether call $C_i$ ``threatens'' $e$, that is, whether edge $e$ could possibly be cut during call $C_i$.
\begin{itemize}
    \item \EMPH{$X_{\rm build}^i$}: indicator random variable that is 1 if $C_i = \textsc{BuildTree}(\cS, H)$ initializes some supernode $X$ with skeleton $T_X$, such that $\dist_{H}(\set{u,v}, T_X) \le 2\Delta$ and $u, v \in H$.
    \item \EMPH{$X_{\rm buffer}^i$}: indicator random variable that is 1 if $C_i = \textsc{GrowBuffer}(\cS, \cX, H)$ such that $u,v \in H$ and $\dist_{H}(\set{u,v}, \cN H_X) \le 2\Delta$.
    \item \EMPH{$X_{\rm split}^i$}: indicator random variable that is 1 if $C_i = \textsc{GrowBuffer}(\cS, \cX, H)$ such that $u,v \in \cN H_X$.
\end{itemize}
We define \EMPH{$X_i$} to be the indicator random variable that is 1 if $X^i_{\rm build}$, $X^i_{\rm buffer}$, or $X^i_{\rm split}$ is 1. 
If $X_i = 1$, we say that \EMPH{call $C_i$ threatens $e$}.
In Section~\ref{SS:cop-threateners}, we prove two lemmas that immediately let us bound the number of threatening calls, under \emph{any} execution of the algorithm.

\begin{lemma}
\label{lem:buildtree-threaten}
    There is a constant $\mu_{\rm build} = O_r(1)$ such that for any vertex $\basev$, 
    there are at most $\mu_{\rm build}$ calls $\textsc{BuildTree}(\cS, H)$ that initializes some supernode $X$ with skeleton $T_X$, such that $\dist_{H}(\basev, T_X) \le 2 \Delta$.
\end{lemma}

\begin{lemma}
\label{lem:growbuffer-threaten}
    There is a constant $\mu_{\rm buffer} = O_r(1)$ such that for any vertex $\basev$, there are at most $\mu_{\rm buffer}$ calls $\textsc{GrowBuffer}(\cS, \cX, H)$ such that $\basev \in H$ and $\dist_H(\basev, \cN H_X) \le 2 \Delta$.
\end{lemma}

\begin{claim}
\label{clm:total-threateners}
    For any execution of the algorithm, $\sum_{i \in \mathbb{N}} X_i \le O_r(1)$.
\end{claim}
\begin{proof}
    For any vertex set $A$, the condition $\dist(\set{u,v}, A) \le 2 \Delta$ implies that either $\dist(u, A) \le 2 \Delta$ or $\dist(v, A) \le 2 \Delta$. Thus, by applying Lemma~\ref{lem:buildtree-threaten} twice (first with $\hat v = u$, and then with $\hat v = v$), we conclude that $\sum_{i \in \mathbb{N}} X^i_{\rm build} \le 2 \cdot \mu_{\rm build} = O_r(1)$.
    Similarly Lemma~\ref{lem:growbuffer-threaten} implies that $\sum_{i} X^i_{\rm buffer} \le 2 \cdot \mu_{\rm buffer} = O_r(1)$.
    Finally, we claim that $\sum_{i} X^i_{\rm split} = 1$. 
    Indeed, observe that if some vertex is in $\cN H_X$ for some call $C_i = \textsc{GrowBuffer}(\cS, \cX, H)$, then that vertex is assigned by $C_i$; the claim follows from the fact that every vertex is assigned only once through the course of the algorithm. 
\end{proof}

We now bound the probability that any specific threatening call cuts $e$.
\begin{claim}
\label{clm:cop-single-probability}
    For any $i$, we have $\Pr[\xi_i \mid X_i = 1] < O_r(\norm{e}/\Delta)$ and $\Pr[\xi_i \mid X_i = 0] = 0$.
\end{claim}

\begin{proof}
First observe that $\Pr[\xi^i_{\rm build} \mid X^i_{\rm build} = 0] = \Pr[\xi^i_{\rm build} \mid X^i_{\rm buffer} = 0] = \Pr[\xi^i_{\rm split} \mid X^i_{\rm split} = 0] = 0$. For the remaining cases, we prove that the probability is upper-bounded by $O_r(\norm{e}/\Delta)$.

\medskip \noindent \textbf{Case 1, \boldmath$\xi^i_{\rm build}$:} Suppose that $X^i_{\rm build} = 1$, and call $C_i = \textsc{BuildTree}(\cS, H)$. The skeleton $T_\eta$ constructed by $C_i$ is a random variable depending on the randomness of the previous calls $C_j$ with $j < i$. 
Crucially, however, $T_\eta$ is independent of the random value $\alpha$ chosen in Step 1 of $\textsc{BuildTree}$. 
Recall that even $\xi^i_{\rm build}$ occurs if exactly one of $u$ and $v$ is assigned to supernode $\eta$ by $C_i$; 
that is, $\xi^i_{\rm build}$ occurs if some $\alpha$ is chosen such that  $\dist_H(v, T_\eta) \le \alpha \cdot \Delta/r$ and  $\dist_H(u, T_\eta) > \alpha \cdot \Delta/r$, or if  $\dist_H(u, T_\eta) \le \alpha \cdot \Delta/r$ and  $\dist_H(v, T_\eta) > \alpha \cdot \Delta/r$. 
By triangle inequality, if $\alpha \cdot \Delta/r < \dist_H(v, T_\eta) - \norm{e}$, then both $\dist_H(v, T_\eta)$ and $\dist_H(u, T_\eta)$ are smaller than $\alpha \cdot \Delta/r$; on the other hand, if $\alpha \cdot \Delta/r > \dist_H(v, T_\eta) + \norm{e}$, then both $\dist_H(v, T_\eta)$ and $\dist_H(u, T_\eta)$ are larger than $\alpha \cdot \Delta/r$. Thus, $\xi^i_{\rm build}$ occurs only if
\begin{equation*}
    \alpha \in \left( \dist_H(v, T_\eta)\cdot \frac{r}{\Delta} - \norm{e}\cdot \frac{r}{\Delta},  \dist_H(v, T_\eta)\cdot \frac{r}{\Delta} + \norm{e} \cdot \frac{r}{\Delta} \right)
\end{equation*}
This is an interval of width $2\norm{e} \cdot r/ \Delta = O_r(\norm{e}/\Delta)$. As $\alpha$ is uniformly distributed and independent of $\dist_H(v, T_\eta)$, event $\xi^i_{\rm build}$ occurs with probability at most $O_r(\norm{e}/\Delta)$.

\medskip \noindent \textbf{Case 2, \boldmath$\xi^i_{\rm buffer}$:} Suppose that $X^i_{\rm buffer} = 1$, and call $C_i = \textsc{GrowBuffer}(\cS, \cX, H)$. As above, the set $\bdry H_{\downarrow X}$ is a random variable that depends only on the randomness of previous calls $C_j$ with $j < i$, and so $\alpha$ is independent from $\bdry H_{\downarrow X}$. By the same argument from Case 1, the probability that $\alpha$ is chosen such that exactly one of $u$ and $v$ is in $\cN H_X$ is at most $O_r(\norm{e}/\Delta)$.
    
\medskip \noindent \textbf{Case 3, \boldmath$\xi^i_{\rm split}$:} Suppose that $X^i_{\rm split} = 1$, and call $C_i = \textsc{GrowBuffer}(\cS, \cX, H)$. The sets $\cS$ and $\cN H_X$ are random variables that depends on the randomness of previous calls $C_j$ with $j < i$ and on the value $\alpha$ chosen is Step 1 of the $C_i$ call to $\textsc{GrowBuffer}$; in particular, for any vertex $v' \in \cN H_X$ and any supernode $X'$ seen by $H$, the distance $\dist_{\dom_{\cS}(X)}(v', X' \cap \bdry H_{\downarrow X})$ is independent of the $\alpha_{X'}$ values chosen during Step 2 of the $C_i$ call. 
(For ease of presentation, define $X''_1 \coloneqq X'_1 \cap \bdry H_{\downarrow X}$ and $X''_2 \coloneqq X'_2 \cap \bdry H_{\downarrow X}$.)
For any two supernodes $X'_1$ and $X'_2$ seen by $H$, we claim that $u$ can be assigned to $X'_1$ and $v$ assigned to $X'_2$ only if
\begin{equation}
\label{eq:voronoi-interval}
    \dist_{\dom_{\cS}(X)}(u, X''_1) + \alpha_{X'_1} \in \left( \dist_{\dom_{\cS}(X)}(v, X''_2) + \alpha_{X'_2} - \norm{e}, \dist_{\dom_{\cS}(X)}(v, X''_2) + \alpha_{X'_2} + \norm{e}\right).
\end{equation}
Indeed, suppose that
\(\dist_{\dom_{\cS}(X)}(u, X''_1) + \alpha_{X'_1} < \dist_{\dom_{\cS}(X)}(v, X''_2) + \alpha_{X'_2} - \norm{e}.\)
Then triangle inequality implies that
\(\dist_{\dom_{\cS}(X)}(v, X''_1) + \alpha_{X'_1} \le \norm{e} + \dist_{\dom_{\cS}(X)}(u, X''_1) + \alpha_{X'_1}
< \dist_{\dom_{\cS}(X)}(v, X''_2) + \alpha_{X'_2},\)
and thus $v$ is not assigned to $X'_2$.
Similarly, if 
\(\dist_{\dom_{\cS}(X)}(u, X''_1) + \alpha_{X'_1} > \dist_{\dom_{\cS}(X)}(v, X''_2) + \alpha_{X'_2} + \norm{e}\),
then one can show
\(\dist_{\dom_{\cS}(X)}(u, X''_1) + \alpha_{X'_1} >\dist_{\dom_{\cS}(X)}(u, X'_2) + \alpha_{X''_2} \),
and so $u$ is not assigned to $X'_1$.

\smallskip
The interval in (\ref{eq:voronoi-interval}) has length $2 \norm{e}$. Recall that $\alpha_{X'_1}$ is uniformly chosen from an interval of length $\Delta/r$, independent of $\alpha_{X'_2}$ and $\dist_{\dom_{\cS}(X)}(v, X''_2)$ random variable. 
Thus for any fixed $X'_1$ and $X'_2$ seen by~$H$, $u$ is assigned to $X'_1$ and $v$ is assigned to $X'_2$ with probability at most $\norm{e}/(\Delta/r)$. Applying union bound over all $O(r^2)$ pairs of supernodes seen by $H$, we conclude that $u$ and $v$ are assigned to different supernodes by Step 2 of $C_i$ with probability at most $O_r(\norm{e}/\Delta)$.
\end{proof}

We now bound the total probability that edge $e$ is cut. 
\begin{lemma}
    Let $G$ be a graph excluding a fixed minor, and let $\cT$ be a buffered cop decomposition constructed by the stochastic procedure $\textsc{BuildTree}(\varnothing, G)$. Then there is a constant $\beta$ such that, for any edge $e$ in $G$, $\Pr[\text{$e$ is cut by $\cT$}] \le \beta \cdot \norm{e}/\Delta$.
\end{lemma}

\begin{proof}
By Lemma~\ref{clm:total-threateners}, $\sum_{i \in \mathbb{N}} X_i \le \mu = O_r(1)$. Then clearly  $\mathbb{E}[\sum_{i \in \mathbb{N}} X_i] \le \mu$, and linearity of expectation implies that $\sum_{i \in \mathbb{N}} \Pr[X_i = 1] \le \mu$.
We have
\begin{align*}
    \Pr[\text{$e$ is cut by $\cT$}]
    &= \Pr[\bigcup_{i \in \mathbb{N}} \xi_i]
    \le \sum_{i \in \mathbb{N}}\Pr[\xi_i]\\
    &= \sum_{i \in \mathbb{N}} \Pr[\xi_i \mid X_i = 1] \cdot \Pr[X_i = 1] + \Pr[\xi_i \mid X_i = 0] \cdot \Pr[X_i = 0] \\
    &\le \sum_{i \in \mathbb{N}} O_r(1) \cdot \frac{\norm{e}}{\Delta} \cdot \Pr[X_i = 1] &\text{(by Claim~\ref{clm:cop-single-probability})}\\
    &\le \mu \cdot O_r(1) \cdot \frac{\norm{e}}{\Delta}
    = O_r(1) \cdot \frac{\norm{e}}{\Delta}.
\end{align*}
\vspace{-11pt}
\aftermath
\end{proof}

\subsection{Bounding the number of threateners}
\label{SS:cop-threateners}

Our goal in this section is to prove 
\Cref{lem:buildtree-threaten} and 
\ref{lem:growbuffer-threaten} from the previous subsection.
Both lemmas turn out to depend on
bounding the number of threateners $X$ to any vertex $\basev$ with respect to $\initdom(X)$.

\begin{lemma}
\label{lem:bdd-threatener-dom0}
For any execution of $\textsc{BuildTree}(\varnothing, G)$, and for any vertex $\basev$, there are $O_r(1)$ supernodes $X$ such that $\dist_{\initdom(X)}(\basev, X) \le 2 \Delta$. (The dependency on $r$ is around $r^{O(r^2)}$.)
\end{lemma}
We will prove \Cref{lem:buildtree-threaten} and \ref{lem:growbuffer-threaten} using \Cref{lem:bdd-threatener-dom0} after some properties of the buffered cop decomposition are introduced in Section~\ref{SSS:properties}.
We remark that if one follows the original construction of buffered cop decomposition, Filtser~\cite[\S4.2]{Fil24} showed that there are $O_r(1)$ supernodes $X$ such that $\dist_{\dom(X)}(\basev, X) \le 2 \Delta$; notice that the distance is measured in the final domain $\dom(X)$ after the buffered cop decomposition is built, as oppose to $\initdom(X)$ which is possibly bigger.
As a result, we cannot use the threatener bound from Filtser~\cite[\S4.2]{Fil24} directly; in fact, we need to perform an involved charging argument based on the recursive call structure of $\textsc{GrowBuffer}$.

\medskip
To bound the number of threateners $X$ with respect to $\initdom(X)$, intuitively we will prove the existence of a \emph{threatening sequence} $(X_0, \dots, X_\ell)$ from $X$ to the supernode $\eta$ containing $\basev$, such that (1) $X_0 = X$, (2) $X_{i-1}$ threatens $X_i$ for each $i$, and (3) $X_\ell$ threatens $\eta$, where each ``threatening'' relation is with respect to the final domain of the threatener.  
This way we can apply the bound from Filtser~\cite[\S4.2]{Fil24} iteratively. 
Unfortunately, we cannot actually prove the second property of the threatening sequence;
instead, we can only guarantee the existence of another supernode ``related'' to $X_i$ that is threatened by $X_{i-1}$.
(See \Cref{clm:threat-sequence}.)
This leads to many complications, and along the way we need to re-establish and strengthen several properties from the buffered cop decomposition, including the supernode buffer property (Lemma~\ref{clm:out-buffer}) and the threatener bound (\Cref{clm:out-neighbor-count} and \Cref{cor:out-filtser}), both with respect to a modification of the final domain.

\paragraph{Charging scheme.}
For every supernode $X_0$ with $\dist_{\initdom(X_0)}(\hat{v}, X_0) \le 2\Delta$, let $(X_0, \ldots, X_\ell)$ denote the threatening sequence of supernodes from $X_0$ to the supernode $\eta$ containing $\basev$ guaranteed by Claim~\ref{clm:threat-sequence}, such that $X_i$ is a \emph{victim} of $X_{i-1}$ for every $i$, and $\eta$ is a victim of $X_\ell$.
(The exact statement of the threatening sequence or the meaning of victim is not important at the moment; one only has to know that $\ell$, the length of any sequence, is less than $r$.)
For each supernode $X$, create $r$ pairs \EMPH{$(X , i)$}, one for each $i \in [0 \,..\, r-1]$.
For each threatening sequence, add a \EMPH{charge} to the pair $(X_0, 0)$.
We then move the charge along every threatening sequence from $(X_{i-1}, i-1)$ to $(X_{i}, i)$ until the charge reaches $(X_\ell, \ell)$.
Finally, for every sequence we move the charge on $(X_\ell, \ell)$ to $\eta$ (length $\ell$ could be different from one sequence to another).
Intuitively speaking, the amount of charge on the pair $(X, i)$ across all threatening sequences counts the number of times a supernode $X$ serves as the $i$-th step threatener in some threatening sequence.
And the total amount of charge on $\eta$ is equal to the number of supernode $X_0$ such that $\dist_{\initdom(X_0)}(\hat{v}, X_0) \le 2\Delta$.

\begin{claim}
\label{clm:charge}
    For any $i \in \mathbb{N}$ and any supernode $X$, the pair $(X, i)$ receives at most $\mu^i$ charges, where $\mu = (r-1) \cdot (r-2) \cdot \binom{7r-1}{r-1}$.
    The supernode $\eta$ containing $\basev$ receives at most $\mu^{r}$ charges.
\end{claim}

We postpone the proof of \Cref{clm:charge}, and proceed to prove Lemma~\ref{lem:bdd-threatener-dom0} using the charging bound.

\begin{proof}[of Lemma~\ref{lem:bdd-threatener-dom0}]

    Fix a vertex $\hat{v}$.
    We need to count the number of supernodes $X$ satisfying $\dist_{\initdom(X)}(\basev, X) \le 2 \Delta$.
    This is equal to the amount of charge on all pairs $(X, 0)$, 
    which is also equal to the final charge received by $\eta$.
    By Claim~\ref{clm:charge}, $\eta$ receives at most $\mu^{r}$ charges, where $\mu = (r-1) \cdot (r-2) \cdot \binom{7r-1}{r-1}$.
    In total there are
    \[
    \Paren{(r-1) \cdot (r-2) \cdot \binom{7r-1}{r-1}}^{r} \le O_r(1)
    \]
    many supernodes $X$ satisfying $\dist_{\initdom(X)}(\basev, X) \le 2 \Delta$.
\end{proof}

The rest of the section is devoted to the existence of threatening sequence (\Cref{clm:threat-sequence}) and 
a proof of the charging bound (\Cref{clm:charge}).

\subsubsection{Properties of buffered cop decomposition}
\label{SSS:properties}

Let $\basev$ be a fixed vertex.
Imagine fixing an arbitrary sequence of random bits, so that there is a well-defined and deterministic notion of ``the set of calls to \textsc{BuildTree} and \textsc{GrowBuffer} made during the execution of $\textsc{BuildTree}(\varnothing, G)$''.
Define \EMPH{$\out \basev$} to be the set of all vertices assigned by some call $\textsc{GrowBuffer}(\cdot, \cdot, H)$
where the subgraph $H$ does not contain $\basev$. 
We need a few observations from \cite{CCLMST24} that extends to the presence of $\out \basev$.
The main reason we introduce $\out \basev$, which may not be clear at the moment, is for a technical but crucial claim for the threatening sequence (see \Cref{clm:find-next-call}).

We elaborate a bit on the motivation for $\out \basev$.  
In the threatening sequence sketched at the start of Section~\ref{SS:cop-threateners}, we would like to guarantee that $X_{i-1}$ threatens $X_i$ (or, more precisely, threatens some supernode related to $X_i$) with respect to $\dom(X_{i-1})$.
However, 
the amount of charge on $X_i$ would depend on 
the number of possible threateners $X_{i-1}$, which is tied to the number of calls to $\textsc{GrowBuffer}$ that expand $X_i$.
The number of such calls could be unbounded.%
\footnote{While the out-of-box [supernode radius] property from \cite[\S3.4]{CCLMST24} implies that $X_i$ can expand at most $O(1)$ times ``in one direction'', there could be an unbounded number of $\textsc{GrowBuffer}$ calls that each expand $X_i$ in different directions. 
Indeed, imagine a long-star graph with a large number of paths attached to a single center vertex $x$, which we initialized to be a supernode $X$.  For each path $P$ hanging off $x$, there are $O(1)$ $\textsc{GrowBuffer}$ calls that assign vertices in $P$ to be part of $X$; however there could be unbounded many calls to $\textsc{GrowBuffer}$ that expands $X$ along different paths $P$.}
We only have a bound on the number of calls $\textsc{GrowBuffer}(\cdot, \cdot, H)$ that expand $X_i$ \emph{and satisfy $\basev \in H$}, for any fixed vertex $\basev$; see Claim~\ref{clm:expand-count}.
To exploit this bound, we construct a threatening sequence with the weaker guarantee that $X_{i-1}$ threatens (a supernode related to) $X_i$ with respect to $\dom(X_{i-1}) \cup \out \basev$, rather than with respect to $\dom(X_{i-1})$.

\medskip
The following observation is essentially a rephrasing of Invariant~3.3 in \cite{CCLMST24}, and follows from the fact that a call to $\textsc{GrowBuffer}$ or $\textsc{BuildTree}$ is made on a subgraph $H$ only if $H$ is a maximal connected component of unassigned vertices.

\begin{observation}[Rephrasing of \cite{CCLMST24} Invariant~3.3]
\label{obs:no-out-adj}
Suppose that some call $C$, whether it is $\textsc{BuildTree}(\cS, H)$ or $\textsc{GrowBuffer}(\cS, \cX, H)$, occurs during execution of the algorithm. 
Then every vertex $x$ in $G \setminus H$ that is adjacent to $H$ was assigned in $\cS$ at the time $C$ was called.  Further, $x$ was assigned during some call to $\textsc{BuildTree}(\cS', H')$ or $\textsc{GrowBuffer}(\cS', \cX', H')$ for which $H$ is a subgraph of $H'$. In particular, this implies that if $\basev \in H$, then every such vertex $x$ cannot be in $\out \basev$. 
\end{observation}

One immediate consequence of Observation~\ref{obs:no-out-adj} is:

\begin{observation}
\label{obs:see-ancestor}
    Let $X$ and $\eta$ be supernodes, where $\initdom(\eta)$ sees $X$ at the time $\eta$ is initialized. Then $X$ is a proper ancestor of $\eta$ in the partition tree.
\end{observation}

\begin{claim}[Rephrasing of \cite{CCLMST24} Claim 3.4(2)]
\label{clm:cclmst-dom}
    Suppose that call $C$, whether it is $\textsc{BuildTree}(\cS, H)$ or $\textsc{GrowBuffer}(\cS, \cX, H)$, occurs at some point in the execution of the algorithm. 
    Over the course of execution, every vertex in $H$ is either assigned to a supernode initialized by or below $C$, or is assigned to some supernode that $H$ sees (at the time $C$ is called).
    In particular, if $C = \textsc{BuildTree}(\cS, H)$ which initializes some supernode $\eta$,
    then every vertex $x$ in $\initdom(\eta)$ but not in $\dom(\eta)$ is assigned to some supernode in $\cS_{|\eta}$.
\end{claim}

\begin{claim}[Implicit from \cite{CCLMST24} Claim~3.5]
\label{clm:cclmst-seen}
    Suppose $\textsc{BuildTree}(\cS, H)$ or~$\textsc{GrowBuffer}(\cS, \cX, H)$ is called during the algorithm. 
    Let $\cS_{|H}$ be the set of supernodes in $\cS$ seen by $H$. 
    Then $\cS_{|H}$ contains at most $r - 2$ supernodes;
    furthermore, the supernodes in $\cS_{|H}$ are pairwise adjacent.
\end{claim}

\begin{claim}[Implicit from {\cite[\S3.4]{CCLMST24}} on supernode radius property]
\label{clm:expand-count}
    For any supernode $\eta$, over the course of execution of the algorithm, there are at most $r-1$ calls $C = \textsc{GrowBuffer}(\cdot, \cdot, H)$ such that $C$ expands $\eta$ and $\hat{v} \in H$.
\end{claim}

\begin{proof}
    First observe that if two calls  $C = \textsc{GrowBuffer}(\cS, \cX, H)$  and $\tilde C = \textsc{GrowBuffer}(\tilde \cS, \tilde \cX, \tilde H)$ are made during the execution of the algorithm with $\hat{v}$ being in both $H$ and $\tilde H$, then either $H$ is a subgraph of $\tilde H$, or $\tilde H$ is a subgraph of $H$. 
    This follows form the fact that calls to $\textsc{GrowBuffer}$ are only made on maximal connected components of unassigned vertices.
    As a result, all calls to $\textsc{GrowBuffer}(\cS, \cX, H)$ that expands $\eta$ where $\hat{v} \in H$ can be put into a linear order based on the containment relationship of $H$.
    Among such calls,
    the contrapositive of Claim~3.12 in \cite{CCLMST24} implies that, for any call $C = \textsc{GrowBuffer}(\cS, \cX, H)$ that processes a supernode $X$, there is no earlier call $\textsc{GrowBuffer}(\tilde \cS, \tilde \cX, \tilde H)$ with $H \subseteq \tilde H$ that processes $X$. 
    Thus, for any supernode $X$, there is at most one call to $\textsc{GrowBuffer}(\cS, \cX, H)$ with $\hat{v} \in H$.

    Finally, it was shown in Claim~3.13 of \cite{CCLMST24} that if some supernode $\eta$ is expanded during a $\textsc{GrowBuffer}$ call that processes $X$, then $\initdom(\eta)$ sees $X$ at the time $\eta$ was initialized.
    It follows from the [shortest-path skeleton] property that $\initdom(\eta)$ sees only $r - 1$ supernodes at the time $\eta$ was initialized. This proves the claim.
\end{proof}

Before we continue, first we prove 
\Cref{lem:buildtree-threaten} and 
\Cref{lem:growbuffer-threaten} from the previous subsection.

\begin{proof}[of \Cref{lem:buildtree-threaten}]
    By \Cref{lem:bdd-threatener-dom0},
    for any vertex $\basev$, there are $O_r(1)$ supernodes $X$ such that $\dist_{\initdom(X)}(\basev, X) \le 2 \Delta$; 
    these includes those supernode $X$ satisfying $\dist_{\initdom(X)}(\basev, T_X) \le 2 \Delta$. 
\end{proof}

\begin{proof}[of \Cref{lem:growbuffer-threaten}]
    For some $i$, suppose $C_i$ is a call to $\textsc{GrowBuffer}(\cS, \cX, H)$ such that $\basev \in H$ and $\dist_H(\basev, \cN H_X) \le 2\Delta$. 
    By definition of $\cN H_X$, there is some vertex $x \in \cN H_X$ such that $\dist_{\dom_{\cS}(X)}(\basev, x) \le 2 \Delta$. 
    Let $x$ be the vertex that minimizes this distance. Let $\eta$ be the supernode containing $x$. Observe that (1) call $C_i$ expands $\eta$, and (2) $\basev \in H$. 
    We further claim that (3) $\dist_{\initdom(\eta)}(\basev, \eta) \le 2 \Delta$; indeed, by choice of $x$, there is a path from $v$ to $x$ in $H \cup \set{x}$ with length at most $2 \Delta$, and Observation~\ref{obs:no-out-adj} implies that $H$ is a subgraph of $\initdom(\eta)$.
    
    By Lemma~\ref{lem:bdd-threatener-dom0}, there are $O_r(1)$ supernodes $\eta$ such that $\dist_{\initdom(\eta)}(v, \eta) < 2 \Delta$. 
    For each such supernode~$\eta$, Claim~\ref{clm:expand-count} implies that there are at most $r-1$ calls to $\textsc{GrowBuffer}(\cdot, \cdot, H)$ with $\basev \in H$ that expand the supernode $\eta$. 
    We conclude that there are at most $O_r(1)$ calls that satisfy properties (1--3).
\end{proof}

\subsubsection{Supernode buffer property}

We now strengthen the [supernode buffer] property proved by \cite{CCLMST24} to work with $\out\basev$.

\begin{claim}
\label{clm:out-buffer}
    Let $\eta$ be a supernode $\basev \in \initdom(\eta)$. 
    Let $X$ be a supernode above\footnote{that is, $X$ is a proper ancestor of $\eta$ in the partition tree} $\eta$ in the partition tree.
    If $\eta$ is not adjacent to $X$ in $G$, then $\dist_{\dom(X) \cup \out \basev} (v, X) > \Delta/r$ for every vertex $v$ in $\dom(\eta)$.
\end{claim}

\begin{proof}
Lemma 3.10 in \cite{CCLMST24} proves a similar claim,
with a guarantee of $\dist_{\dom(X)} (v,X) > \Delta/r$ instead of $\dist_{\dom(X) \cup \out \basev} (v,X) > \Delta/r$. 
The proof from \cite{CCLMST24} carries over almost verbatim for our new statement, with two modifications. First, \cite{CCLMST24} introduce a claim that they proved inductively:
\begin{quote}
    Let $C' \coloneqq \textsc{BuildTree}(\cS', H')$ be a call that is below (in the recursion tree) the call that initialized $X$. Either $H'$ sees $X$ (at the time $C'$ is called), or $\dist_{\dom(X)}(v, X) > \Delta/r$ for every vertex $v$ in $H'$.
\end{quote}
For our proof, we slightly modify the claim:
\begin{quote}
    Let $C' \coloneqq \textsc{BuildTree}(\cS', H')$ be a call that is below (in the recursion tree) the call that initialized $X$, \textcolor{BrickRed}{with $\basev \in H'$}. 
    Either $H'$ sees $X$ (at the time $C'$ is called), or $\dist_{\dom(X) \textcolor{BrickRed}{\cup \out \basev}}(v, X) > \Delta/r$ for every vertex $v$ in $H'$.
\end{quote}
The additional assumption that $\basev \in H'$ can be added without breaking the proof of \cite{CCLMST24}; indeed, they only apply the inductive claim to the call $C_\eta \coloneqq \textsc{BuildTree}(\cS_\eta, H_\eta)$ that initialized $\eta$ (for which we have $\basev \in \initdom(\eta) = H_\eta$ by assumption) or to a call $\textsc{BuildTree}(\cS', H')$ that is an ancestor of $C_\eta$ in the recursion tree (and thus $H_\eta$ is a subgraph of $H'$, and $\basev \in H'$).

\medskip
The second modification occurs in the proof of the claim introduced above. We first recall some details from the proof of \cite{CCLMST24}. During their proof, they consider a shortest path $P$ in $\dom(X)$ between $X$ and $v$, and show that $\norm{P} > \Delta/r$ if $\eta$ is not adjacent to $X$. 
They consider inductively a certain call $C= \textsc{GrowBuffer}(\cS, \cX, H)$ 
where $H'$, the graph from the call $C' \coloneqq \textsc{BuildTree}(\cS', H')$ described above, is a subgraph of $H$
and $H$ is a subgraph of $\initdom(X)$, such that $C$ processes $X$. 
They define $x$ to be some vertex on $P$ that is in $G \setminus H$ and is adjacent to $H$ (they show that such an $x$ exists, and without loss of generality we may assume that it is the first such vertex along $P$ when travelling from $v$ to $X$).
Recall that $\bdry H_{\downarrow X}$ is defined as the set of vertices in $G \setminus H$ that are adjacent to $H$ and are in $\dom_{\cS}(X)$.
They need to show that $\dist_{\dom_{\cS}(X)}(v, \bdry H_{\downarrow X}) \le \norm{P}$. (From this, they show that $\norm{P} \le \Delta/r$ implies that $v$ must have been assigned by the call $C$ and thus $v \not \in H'$, a contradiction; thus $\norm{P} > \Delta/r$. We emphasize that this argument holds with our stochastic version of \textsc{GrowBuffer} --- as we always choose $\alpha \ge 1$ in Step 1 of \textsc{GrowBuffer}, every vertex within distance $\Delta/r$ of $\bdry H_{\downarrow X}$ is assigned during $C$.) 
Their proof that $\dist_{\dom_{\cS}(X)}(v, \bdry H_{\downarrow X}) \le \norm{P}$ is the only place that uses the assumption that $P$ is in $\dom(X)$, and it is as follows: (1) $P$ is a path in $\dom(X)$ and thus in $\dom_{\cS}(X)$ (as domains only shrink over the course of the algorithm); and (2) as $P$ is in $\dom_{\cS}(X)$, the vertex $x$ is in $\dom_{\cS}(X)$ and thus in $\bdry H_{\downarrow X}$;  thus $\dist_{\dom_{\cS}(X)}(v, \bdry H_{\downarrow X}) \le \dist_{\dom_{\cS}(X)}(v, x) \le \norm{P}$. 

To prove our new claim, we instead 
let $P$ be a shortest path in $\dom(X) \cup \out \basev$ between $X$ and~$v$; we show that $\dist_{\dom_{\cS}(X)}(v, \bdry H_{\downarrow X}) \le \norm{P}$. 
(Notice that the distance is still measured in $\dom_{\cS}(X)$ as \textsc{GrowBuffer} algorithm expands supernode based on it.)
We observe that $\basev$ is in $H$; this is because $\basev \in H'$ (by assumption) and $H'$ is a subgraph of $H$.
It follows from Observation~\ref{obs:no-out-adj} that vertex $x$ is not in $\out \basev$. 
Thus, $x \in \dom_{\cS}(X)$, and so $x \in \bdry H_{\downarrow X}$. 
By choice of $x$, the prefix of $P$ from $v$ to~$x$, denoted $P[v:x]$, is contained in $H$. 
But $H$ is a subgraph of $\dom_{\cS}(X)$, as $H$ is a subgraph of $\initdom(X)$ (by assumption on the call $C$) and every vertex of $H$ is unassigned in $\cS$.
Thus, $\dist_{\dom_{\cS}(X)}(v, \bdry H_{\downarrow X}) \le \norm{P}$. By following the proof of \cite{CCLMST24}, we conclude that $\norm{P} > \Delta/r$ unless $\eta$ is adjacent to $X$.
\end{proof}

Following Filtser~\cite{Fil24}, for any buffered cop decomposition we define a \EMPH{dag $\vec G$} whose vertices are the supernodes, and there is an edge from supernode $\eta$ to supernode $\eta'$ if (1) $\eta$ and $\eta'$ are adjacent in $G$, and (2) $\eta'$ is an ancestor of $\eta$ in the partition tree. As observed by \cite{Fil24}, the [shortest-path skeleton] property of buffered cop decomposition implies that every supernode in $\vec G$ has out-degree at most $r - 1$.

\begin{claim}[{\cite[Lemma 1]{Fil24}}, with $w=r-1$]
\label{clm:out-neighbor-count}
    For any supernode $\eta$ in $\vec G$ and any $q \in \mathbb{N}$, there are at most $\mu = \binom{q + r - 1}{r - 1}$ supernodes $X$ such that $\vec G$ contains an path from $\eta$ to $X$ of length at most $q$. 
\end{claim}

The next claim is similar in spirit to \cite[Lemma~2]{Fil24}, which is an analogous claim with $\dist_{\dom(X)}(\eta, X) \le q \cdot \Delta/r$ instead of $\dist_{\dom(X) \cup \out \basev}(\eta, X) \le q \cdot \Delta/r$. The details of the proof are rather different, and we provide a complete proof here.

\begin{claim}
\label{clm:out-filtser}
    Let $\eta$ be the supernode with $\basev \in \initdom(\eta)$, let $X$ be a supernode with skeleton $T_X$ such that $X$ is above $\eta$ in the partition tree, and let $q \in \mathbb{N}$ with $q \ge 1$. If $\dist_{\dom(X) \cup \out \basev}(\eta, T_X) \le q \cdot \Delta/r$,
    then $\vec G$ contains a path from $\eta$ to $X$ of length at most $2q$.
\end{claim}

\begin{proof}
    We prove the claim by induction on $q$. If $q = 1$, then the claim immediately follows from our stronger supernode buffer property (Claim~\ref{clm:out-buffer}).
    
    In the inductive case when $q > 1$, 
    if there is a path in $\vec G$ from $\eta$ to $X$ of length at most 2, then we are done. Otherwise, let $P$ be a shortest-path in $\dom(X) \cup \out \basev$ between some vertex $a \in \eta$ and a vertex $b \in T_X$, with $\norm{P} \le q \cdot \Delta/r$. 
    We claim:
    \begin{equation}
    \label{eq:hop-step}
    \parbox{14cm}{%
    There is some supernode $\hat \eta$ such that (1) some vertex on $P$ is in $\hat \eta$, (2) there is \emph{not} an edge from $\eta$ to $\hat \eta$ in $\vec G$, but (3) there \emph{is} a 2-hop path from $\eta$ to $\hat \eta$ in $\vec G$,  (4) $\basev \in \initdom(\hat \eta)$, and (5) $X$ is above $\hat \eta$ in the partition tree.
    }
    \end{equation}
    
    \noindent We will prove (\ref{eq:hop-step}) later. 
    Now we complete the inductive step assuming (\ref{eq:hop-step}). 
    Indeed, let $x$ be some vertex on $P \cap \hat \eta$ from (\ref{eq:hop-step})(1).  
    By (\ref{eq:hop-step})(2) and the buffer property of Claim~\ref{clm:out-buffer}, we have $\dist_{\dom(X) \cup \out \basev}(\eta, \hat \eta) > \Delta/r$. 
    As $P$ passes through $\hat \eta$, we have $\dist_{\dom(X) \cup \out \basev}(\hat \eta, X) \le \norm{P} - \Delta/r \le (q - 1) \cdot \Delta/r$. 
    By the inductive hypothesis (which we may apply due to conditions (\ref{eq:hop-step})(4--5), $\vec G$ contains a path from $\hat \eta$ to $X$ of length at most $2q-2$. Thus by (\ref{eq:hop-step})(3), $\vec G$ contains a path from $\eta$ to $X$ of length at most $2q$.

    \medskip
    We now prove (\ref{eq:hop-step}). Define $\EMPH{$\eta_0$} \coloneqq \eta$ and $\EMPH{$a_0$} \coloneqq a$. 
    For every $i > 0$, inductively define \EMPH{$a_i$} to be the first vertex along the subpath $P[a_{i-1}: b]$ (when moving from $a_{i-1}$ to $b$) that is not in $\dom(\eta_{i-1}) \cup \out \basev$, and let \EMPH{$\eta_i$} be the supernode containing $a_i$.
    We make three observations about $\eta_i$.
    \begin{itemize}
        \item \emph{There is an edge from $\eta_{i-1}$ to $\eta_i$ in $\vec G$, and $\eta_i$ is a proper ancestor of $\eta_0$ in the partition tree.} 
        It suffices to show that $\initdom(\eta_{i-1})$ sees $\eta_i$ at the time $\eta_{i-1}$ is initialized: in this case, by construction there is an edge between $\eta_{i-1}$ and $\eta_i$ in $G$, and Observation~\ref{obs:see-ancestor} implies that $\eta_i$ is a proper ancestor of $\eta_{i-1}$ (and thus of $\eta_0$ as well by induction) in the partition tree.
    To this end, consider the subpath $P[a_{i-1}:a_i]$ of path $P$ that starts at $a_{i-1}$ and ends at $a_i$. By choice of $a_i$, path $P[a_{i-1}:a_i]$ is contained in $\dom(\eta_{i-1}) \cup \out \basev \cup \eta_i$.
    Consider the path $P[a_{i-1}:a_i]$ at the time $\eta_{i-1}$ was initialized by some call $\textsc{BuildTree}(\cS, H)$. As $a_{i-1} \in \eta_{i-1}$, we have $a_{i-1} \in \initdom(\eta_{i-1}) = H$. 
    Now, there are two cases. 
    If $a_i \in H$, then $a_i$ is assigned to some supernode seen by $\eta_{i-1}$, by Claim~\ref{clm:cclmst-dom}, and we are done. 
    Otherwise, $a_i \not \in H$. 
    As one endpoint of $P[a_{i-1}:a_i]$ is in $H$ and the other is not in $H$, there exists some vertex $y$ on $P[a_{i-1}:a_i]$ that lies in $G \setminus H$ but is adjacent to $H$. 
    As $\basev \in \initdom(\eta_0)$ by assumption, and $\eta_{i-1}$ is a proper ancestor of $\eta_0$ in the partition tree by induction, we have $\basev \in \initdom(\eta_{i-1}) = H$, so Observation~\ref{obs:no-out-adj} implies that $y \not \in \out \basev$. 
    But $y \not \in \dom(\eta_{i-1})$ as it is not in $H$, and so $y \in \eta_i$. We conclude that $\initdom(\eta_{i-1})$ sees $\eta_i$.

    \item \emph{For every $\eta_i$, we have $\basev \in \initdom(\eta_i)$, and either $X = \eta_i$ or $X$ is above $\eta_i$ in the partition tree.}
    We have $\basev \in \initdom(\eta_i)$ immediately from the fact that every $\eta_i$ is above $\eta_{i-1}$ in the partition tree, and $\basev \in \initdom(\eta_0)$. 
    For the latter claim, observe that $a_i \not \in \out \basev$ (as, by definition, it is the first vertex not in $\dom(\eta_{i-1}) \cup \out \basev$)
    and so it is in $\dom(X)$ (as $P$ is a path in $\dom(X) \cup \out \basev$); the fact that $a_i \in \eta_i$ and $a_i \in \dom(X)$ implies that either $\eta_i = X$ or $X$ is above $\eta_i$ in the partition tree.

    \item Let $i_{\max}$ be the largest number such that $a_{i_{\max}}$ exists;
    that is, every vertex on the subpath $P[a_{i_{\max}}:b]$
    is in $\dom(\eta_{i_{\max}}) \cup \out \basev$. 
    \emph{We argue that $\eta_{i_{\max}} = X$. }
    Indeed, vertex $b$ is in the skeleton $T_X$, so it is assigned during a call to $\textsc{BuildTree}$ and not during a call to $\textsc{GrowBuffer}$. Thus, $b \not \in \out \basev$, and so (by assumption on $P[a_{i_{\max}}:b]$) we have $b \in \dom(\eta_{i_{\max}})$.
    As $b \in T_X$ and $b \in \dom(\eta_{i_{\max}})$, either $\eta_{i_{\max}} = X$ or $\eta_{i_{\max}}$ is above $X$ in the partition tree. On the other hand, we showed above that either $\eta_{i} = X$ or $X$ is above $\eta_{i}$ in the partition tree for every $\eta_i$. 
    We conclude that $\eta_{i_{\max}} = X$.
    \end{itemize}

    We can now prove (\ref{eq:hop-step}). Let $k$ be the smallest number such that there is \emph{not} an edge from $\eta$ to $\eta_k$ in~$\vec G$. 
    (If no such supernode exists, then  there is an edge from $\eta$ to $\eta_{i_{\rm max}} = X$ in $\vec G$. This contradicts our assumption that there is no path from $\eta$ to $X$ with length at most 2.)
    (1) Clearly, there is some vertex on $P$ in $\eta_k$ (namely, $a_k$), and (2) there is no edge from $\eta$ to $\eta_k$ in $\vec G$. Finally, there is an edge in $\vec G$ from $\eta_{k-1}$ to $\eta_k$ by construction, and (by assumption on $k$) there is an edge in $\vec G$ from $\eta$ to $\eta_{k-1}$; thus, (3) there is a 2-hop path from $\eta$ to $\eta_k$ in $\vec G$. As show above, (4) $\basev \in \initdom(\eta_k)$, and (5) either $X = \eta_k$ or $X$ is above $\eta_k$ in the partition tree (and if $X = \eta_k$, $\vec G$ contains a 2-hop path from $\eta$ to $X$, contradicting our assumption). Thus, choosing $\hat \eta \coloneqq \eta_k$ satisfies (\ref{eq:hop-step}).        
\end{proof}

It is easy to generalize the claim to apply to the case where $\dist_{\dom(X) \cup \out \basev}(\eta, X)$ is bounded.
We say that supernode $\eta$ is a \EMPH{$q$-step victim of a supernode $X$} (\emph{victim} for short) if 
(i) $\basev \in \initdom(\eta)$, 
(ii) $X$ is an ancestor of $\eta$ in the partition tree, where possibly $X = \eta$, and
(iii) $\dist_{\dom(X) \cup \out \basev}(\eta, X) \le q \cdot \Delta/r$.

\begin{corollary}
\label{cor:out-filtser}
    Let integer $q \ge 1$. 
    Let $\eta$ be a supernode that is a $q$-step victim of another supernode $X$.
    Then $\vec G$ contains a path from $\eta$ to $X$ of length at most $2(q+r)$.
\end{corollary}
\begin{proof}
    The statement is trivial if $X=\eta$.
    Let $T_X$ denote the skeleton of $X$. Every vertex in $X$ is within distance $\Delta$ of $T_X$,
    by the [supernode radius] property. Thus, $\dist_{\dom(X) \cup \out \basev}(\eta, X) \le q \cdot \Delta/r + \Delta = (q+r) \cdot \Delta/r$, and the claim follows from Claim~\ref{clm:out-filtser}.
\end{proof}

\begin{corollary}
\label{cor:victim-bound}
    Let $\eta_i$ be a supernode.
    There are at most $\binom{2(q+r)+r-1}{r-1}$ supernodes $X_{i-1}$ that has $\eta_i$ as a $q$-step victim.
\end{corollary}

\begin{proof}
    Corollary~\ref{cor:out-filtser} (where $\eta = \eta_{i}$ and $X = X_{i-1}$) shows that if $X_{i-1}$ has $\eta_i$ as a $q$-step victim, then there is a path from $\eta_{i}$ to $X_{i-1}$ of length at most $2(q+r)$ in dag $\vec G$.
    By \Cref{clm:out-neighbor-count}, there are at most $\binom{2(q+r)+r-1}{r-1}$ supernodes $X_{i-1}$ such that $\vec G$ contains an path from $\eta_{i}$ to $X_{i-1}$ of length at most $2(q+r)$.
\end{proof}

\subsubsection{Threatening sequence}

So far, we have bounded the number of supernodes $X$ for which $\dist_{\dom(X) \cup \out \basev}(\basev, X) \le 2\Delta$. 
However, we really want to bound the number of supernodes $X$ for which $\dist_{\initdom(X)}(\basev, X) \le 2\Delta$.
We now show the following lemma, \Cref{clm:find-next-call}, which is the heart of the proof:
if $\dist_{\initdom(X)}(\basev, X) \le 2\Delta$ but $\dist_{\dom(X) \cup \out \basev}(\basev, X) > 2\Delta$, we can ``charge'' this to some $C = \textsc{GrowBuffer}(\cdot, \cdot, H)$ call and replace $X$ with some supernode $X'$ expanded by $C$.
In this way, we bound the number of supernodes $X$ for which $\dist_{\initdom(X)}(\basev, X) \le 2\Delta$.
The complication we mentioned at the start of the section is that
$X'$ is not a victim (and thus a descendant) of $X$, but is in fact an ancestor of $X$; 
we have to choose a victim from the supernodes seen by $H$ during the $\textsc{GrowBuffer}(\cdot, \cdot, H)$ call instead.
(In fact, it is crucial that $X'$ is an ancestor of $X$, as this lets us prove that the threatening sequence has length at most $r-1$; see Claim~\ref{clm:threat-sequence}.)
See Figure~\ref{fig:threat-seq-one-step} for an example of Claim~\ref{clm:find-next-call}.

\begin{claim}
\label{clm:find-next-call}
    Let $X$ be a supernode and $q \in \mathbb{N}$. 
    Let $P$ be a path in $\initdom(X)$ from $X$ to $\basev$ with length at most $q \cdot \Delta/r$. 
    Then either $P$ is already in $\dom(X) \cup \out \basev$;
    or there is a call $C = \textsc{GrowBuffer}(\cS, \cX, H)$ 
    such that 
    \begin{enumerate}
        \item[(1)] $\basev \in H$,
        
        \item[(2)] there is some $q$-step victim $\eta'$ of $X$
        such that $H$ sees $\eta'$ at the time $C$ is called,

        \item[(3)] there is some supernode $X'$ expanded by $C$ and some suffix path $P'$ of $P$, such that $P'$ is a path from $X'$ to $\basev$ in $\initdom(X')$, 
        and $X'$ is a proper ancestor of $X$ in the partition tree.
    \end{enumerate} 
\end{claim}

\begin{proof}
    Suppose that $P$ is not in $\dom(X) \cup \out \basev$.
    We will show that some appropriate call $C$ exists.
    Let \EMPH{vertex $x$} be the first one along $P$ (walking from $X$ to $\basev$) that is not in $\dom(X) \cup \out \basev$.
    As $x$ is on $P$ which is in $\initdom(X)$, Claim~\ref{clm:cclmst-dom} implies that $x$ was assigned to some supernode $X'$ that was seen by $\initdom(X)$. 
    Let $C = \textsc{GrowBuffer}(\cS, \cX, H)$ be the call during which $x$ was assigned. (In particular, $x$ is in $H$.)
    We show that $C$ satisfies the three conditions required by the claim.
    \begin{enumerate}
        \item[(1)] 
        By choice of $x$, we have $x \not \in \out \basev$, and so $\basev \in H$ by definition of $\out \basev$.
        
        \item[(2)] Consider the prefix $P[X:x]$ of $P$ that runs from $X$ to $x$. One endpoint $x$ of $P[X:x]$ is in $H$.
        Suppose the whole prefix $P[X:x]$ is in $H$. 
        The other endpoint $x'$ of $P[X:x]$ must \emph{eventually} be assigned to $X$. 
        Notice that $X$ cannot be initialized by or after $C$, because
        the whole path $P$, including $x$, is in $\initdom(X)$.
        Every vertex in $\initdom(X)$ is unassigned at the time $X$ is initialized; as call $C$ assigns $x$, call $C$ cannot occur before $X$ is initialized.
        Claim~\ref{clm:cclmst-dom}, applied on $C$ and $x' \in H$, then implies that 
        $x'$ is assigned to a supernode (which must be $X$) that
        $H$ sees at the time $C$ is called.  
        We choose $\eta' \coloneqq X$, and $\eta'$ is trivially a $q$-step victim of itself.
        
        If some vertex on $P[X:x]$ is not in $H$, then there is some \EMPH{vertex $x'$} on $P[X:x] \setminus \set{x}$ that is in $G \setminus H$ but is adjacent to $H$. 
        We choose $\eta'$ to be the supernode containing $x'$; because $x' \in G\setminus H$, $\eta'$ must have included $x'$ at the time $C$ is called.
        We show that $\eta'$ is a $q$-step victim of $X$ and $H$ sees $\eta'$.
        By definition of $x'$ and $\eta'$, the graph $H$ sees $\eta'$ at the time $C$ is called.
        (i) As $H$ sees $\eta'$, it follows (by repeatedly applying Observation~\ref{obs:no-out-adj})
        that $\eta'$ was initialized by some call to $\textsc{BuildTree}(\cS', H')$ where $H \subseteq H'$; in particular, by (1), $\basev \in H \subseteq H' = \initdom(\eta')$.
        (ii) By Observation~\ref{obs:no-out-adj}, $x' \not \in \out \basev$; by the choice of $x$, $P[X:x']$ is contained in $\dom(X) \cup \out \basev$, thus we have $x' \in \dom(X)$. 
        This implies that either $X = \eta'$ or $X$ is above $\eta'$ in the partition tree.
        (iii) Again, because the path $P[X:x']$ is contained in $\dom(X) \cup \out \basev$, we have $\dist_{\dom(X) \cup \out \basev}(\eta', X) \le q \cdot \Delta/r$. 
        
        \item[(3)] By definition, $X'$ was expanded by $C$. 
        Further, consider the subpath $P' \coloneqq P[x:\basev]$ of $P$. 
        This is a path from $X'$ to $\basev$. 
        Further, as $X'$ is a proper ancestor of $X$ in the partition tree (by Observation~\ref{obs:see-ancestor}), $\initdom(X)$ is a subgraph of $\initdom(X')$. 
        As $P$ is a path in $\initdom(X)$, it follows that $P[x:\basev]$ is a path in $\initdom(X')$. 
        Finally, observe that $x \not \in \dom(X)$ and $x$ was assigned to $X'$, so $X \neq X'$. \qed
    \end{enumerate}
\end{proof}
\vspace{-22pt}

\begin{figure}
    \centering
    \includegraphics[width=\textwidth]{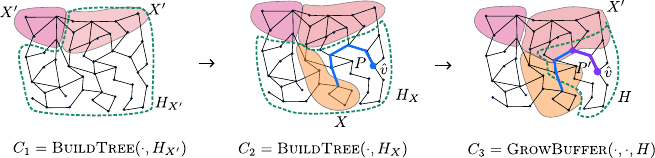}
    \caption{A series of \textsc{BuildTree} and \textsc{GrowBuffer} calls, demonstrating one case of \Cref{clm:find-next-call}. The calls $C_1$ and $C_2$ initialize supernodes $X'$ and $X$, respectively. Supernode $X'$ is an ancestor of $X$, so vertices in $X'$ are not in $\dom(X)$. The blue path $P$ is a shortest-path from $X$ to $\basev$, in the graph $\initdom(X) \cup \out \basev$. The \textsc{GrowBuffer} call $C_3$ expands supernode $X'$, causing $P$ to not be contained in $\dom(X) \cup \out \basev$. The call $C_3$ satisfies the conditions of Claim~\ref{clm:find-next-call}: $\hat v$ is in $H$, supernode $\eta' \coloneqq X$ is seen by $H$, and supernode $X'$ is connected to $\basev$ by a subpath $P'$ of $P$.}
    \label{fig:threat-seq-one-step}
\end{figure}

By applying this claim repeatedly, we can prove the existence of a threatening sequence.

\begin{claim}
\label{clm:threat-sequence}
    Let $q \in \mathbb{N}$, and let $X$ be a supernode with $\dist_{\initdom(X)}(\basev, X) \le q \cdot \Delta/r$. Define $X_0 \coloneqq X$.
    For some $\ell < r - 1$
    (possibly $\ell = 0$), there is a sequence of $\ell$ calls $C_i = \textsc{GrowBuffer}(\cS_i, \cX_i, H_i)$ and a sequence of supernodes $X_i$ expanded by $C_i$ such that:
    \begin{enumerate}
        \item[(1)] for every $i \in [1 \,..\, \ell]$, $\basev \in H_i$;

        \item[(2)] for every $i\in [1 \,..\, \ell]$, there is a $q$-step victim  $\eta_i$ of $X_{i-1}$ such that $H_i$ sees $\eta_i$ at the time $C_i$ is called;
                
        \item[(3)] denote $\eta_{\basev}$ the supernode containing $\basev$, then $\eta_{\basev}$ is a $q$-step victim of $X_\ell$.
        
    \end{enumerate}
\end{claim}

\begin{proof}
Let $P_0$ be a path in $\initdom(X_0)$ from $X_0$ to $\basev$, with $\norm{P_0} \le q \cdot \Delta/r$. 
For every $i > 0$, we either decide to end the algorithm by setting $\ell = i - 1$ or we inductively define a call $C_i$, a supernode $X_i$, and suffix path $P_i$, using supernode $X_{i-1}$ and path $P_{i-1}$.  
We maintain the following invariant: 
Let $\cS_{|X_0} \coloneqq \set{S_1, \ldots, S_k}$ denote the set of supernodes seen by $\initdom(X_0)$ at the time $X_0$ is created, where $k \le r - 2$ by \Cref{clm:cclmst-seen}.
\begin{quote}
$P_i$ is a path in $\initdom(X_i)$ from $X_i$ to $\basev$ that is a suffix path of $P_0$, and lies in the union of $\dom(X_i) \cup \set{\text{at most $k - i$ supernodes from $\cS_{|X_0}$}}$, and the call $C_i$ satisfies conditions (1--2).
\end{quote}
For the base case when $i=0$, $P$ lies in $\dom(X) \cup \bigcup \cS_{|X_0}$ by Claim~\ref{clm:cclmst-dom}.

\medskip
Apply Claim~\ref{clm:find-next-call} (choosing $X = X_{i-1}$ and $P = P_{i-1}$). 
Either $P_{i-1}$ is already in $\dom(X_{i-1}) \cup \out \basev$ and thus $\basev \in \dom(X_{i-1}$) (because $\basev$ is never in $\out \basev$), or there is a call $C_i = \textsc{GrowBuffer}(\cS_i, \cX_i, H_i)$ satisfying (1--3) of Claim~\ref{clm:find-next-call}.
In the first case, we show that $\eta_{\basev}$ is a $q$-step victim of $X_\ell \coloneqq X_{i-1}$.
We have $\basev \in \initdom(\eta_{\basev})$ and
$X_{i-1}$ is above $\eta_{\basev}$ in the partition tree (as $\basev \in \dom(X_{i-1})$), and 
\[
\dist_{\dom(X_{i-1}) \cup \out \basev}(\eta_v, X_{i-1}) \le \norm{P_{i-1}} \le q \cdot \Delta/r.
\]
Thus, we may choose $\ell \coloneqq i - 1$, and $\eta_{\basev}$ and $X_\ell$ satisfy condition (3).
Notice that $\ell$ must be less than $k$, 
because by the invariant $P_{i}$ lies in the union of $\dom(X_{i})$ and at most $k-i$ supernodes from $\cS_{|X_0}$; after up to $k$ iterations $P_{i}$ lies in $\dom(X_{i})$ completely.

In the other case, we choose $C_i$ to be the call $\textsc{GrowBuffer}(\cS_i, \cX_i, H_i)$ guaranteed by Claim~\ref{clm:find-next-call}, choose $X_i$ to be the supernode expanded by $C_i$ (called $X'$ in the statement of Claim~\ref{clm:find-next-call}), and choose $P_i$ to be the suffix path of $P_{i-1}$ that connects $X_i$ and $\basev$ in $\initdom(X_i)$ (called $P'$ in the statement of Claim~\ref{clm:find-next-call}). 
Conditions (1--2) follow immediately from Claim~\ref{clm:find-next-call}(1--2). 
It remains to show that $P_i$ is in the union of $\dom(X_i)$ and at most $k-i$ supernodes from $\cS_{|X_0}$. 
By the invariant, $P_{i-1}$ is in the union of $\dom(X_{i-1})$ and at most $k-i+1$ supernodes from $\cS_{|X_0}$, without loss of generality to be $\set{S_1, \ldots, S_{k-i+1}}$.  
By Claim~\ref{clm:find-next-call}(3), supernode $X_i$ is a proper ancestor of $X_{i-1}$ in the partition tree, and thus $X_i$ is not in $\dom(X_{i-1})$; as some vertex of $P_{i-1}$ lies on $X_i$, we thus have $X_i \in \set{S_1, \ldots, S_{k-i+1}}$. 
Without loss of generality suppose that $X_i = S_{k-i+1}$. 
Now, because $X_i$ is an ancestor of $X_{i-1}$ in the partition tree, $\dom(X_{i-1})$ is a subgraph of $\dom(X_i)$ (and clearly $X_i = S_{k-i+1}$ is a subgraph of $\dom(X_i)$). 
We conclude that $P_i$ is a path in the union of $\dom(X_i)$ and
$\set{S_1, \ldots, S_{k-i}}$.
\end{proof}

\subsubsection{Proof of charging bound}

Now we are ready to prove stronger version of \Cref{clm:charge}; \Cref{clm:charge} follows by setting $q=2r$.
To handle the technical issue that $X_{i-1}$ does not threaten $X_i$ directly but threatens instead some $\eta_i$ seen by $H_i$ where $C_i = \textsc{GrowBuffer}(\cS_i, \cX_i, H_i)$ was called, 
we argue that at most $r-1$ supernodes $X_i$ are expanded by $C_i$ and thus at most $r-1$ supernodes $\eta_i$ can be seen by $H_i$ (at the time $C_i$ is called) using \Cref{clm:expand-count}.

\begin{claim}
\label{clm:charge-strong}
    For any $i \in \mathbb{N}$ and any supernode $X$, the pair $(X, i)$ receives at most $\mu^i$ charges, where $\mu = (r-1) \cdot (r-2) \cdot \binom{2(q+r)+r-1}{r-1}$.
    The supernode $\eta_{\basev}$ containing $\basev$ receives at most $\mu^{r}$ charges.
\end{claim}

\begin{proof}
    We will prove the statement by induction on $i$.
    For the base case when $i=0$, every super\-node $X$ induces one threatening sequence with $X_0=X$ which charges $(X_0,0)$ once.
    For the inductive case,
    consider any supernode $X_0$ with $\dist_{\initdom(X_0)}(\hat{v}, X_0) \le q\cdot \Delta / r$, such that the sequence of supernodes $(X_0, \ldots, X_\ell)$ guaranteed by Claim~\ref{clm:threat-sequence} contains the supernode $X$ in question as $X_i$.
    Let $C_i = \textsc{GrowBuffer}(\cS_i, \cX_i, H_i)$ be the call that expands $X_i$ for each $i$.
    
    For any supernode $X = X_i$ expanded by $C_i = \textsc{GrowBuffer}(\cdot, \cdot, H_i)$,
    Claim~\ref{clm:threat-sequence}(1) ensures $\hat{v} \in H_i$.
    Consequently, at most $r-1$ calls of the form $\textsc{GrowBuffer}(\cdot, \cdot, H_i)$ can expand $X_{i}$ where $\basev \in H_i$, by Claim~\ref{clm:expand-count}.
    For each of these calls on $H_i$, there are at most $r-2$ many supernodes $\eta_i$ seen by $H_i$ by \Cref{clm:cclmst-seen}. 
    For each such $\eta_i$, there are at most $\binom{2(q+r)+r-1}{r-1}$ supernodes $X_{i-1}$ that has $\eta_i$ as a $q$-step victim, 
    by \Cref{cor:victim-bound}.
    Claim~\ref{clm:threat-sequence}(2) ensures there is a $q$-step victim $\eta_i$ of $X_{i-1}$ such that $H_i$ sees $\eta_i$. 
    By induction, each $(X_{i-1},i-1)$ receives is at most $\mu^{i-1}$ charges.
    In total, $(X,i)$ receives at most 
    \[
    (r-1) \cdot (r-2) \cdot \binom{2(q+r)+r-1}{r-1} \cdot \mu^{i-1} = \mu^i
    \]
    charges, if we set $\mu = (r-1) \cdot (r-2) \cdot \binom{2(q+r)+r-1}{r-1}$.

    \medskip
    As for the total amount of charges on $\eta_{\basev}$,
    Claim~\ref{clm:threat-sequence}(3) says that $\eta_{\basev}$ is a victim of $X_\ell$.
    There are at most $\binom{2(q+r)+r-1}{r-1}$ supernodes $X_{\ell}$ that has $\eta_{\basev}$ as a victim, 
    by \Cref{cor:victim-bound}.
    Each $(X_{\ell},\ell)$ receives is at most $\mu^{\ell}$.
    The index $\ell$ for the threatening sequence to $\basev$ may range from $0$ to $r-1$.
    As a result, $\eta_{\basev}$ receives at most 
    \(
    \sum_{0 \le \ell \le r-1} \mu^{\ell} \le \mu^{r}
    \)
    charges.
\end{proof}

\subsection{Stochastic shortcut partition}

\cite{CCLMST24} construct a shortcut partition from buffered cop decomposition. We slightly modify their algorithm; see Figure~\ref{fig:shortcut}. 
As before, we highlight our modification in red.

\begin{figure}[h!]
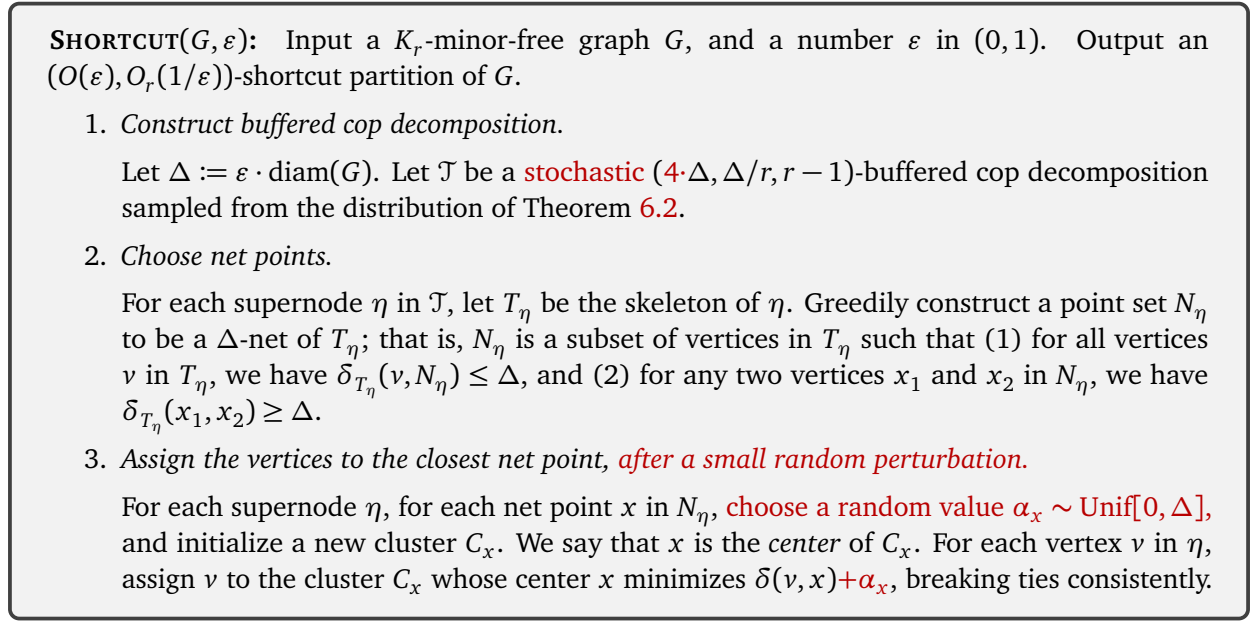

\centering
\begin{tcolorbox}
\paragraph{$\textsc{Shortcut}(G, \e)$:} Input a $K_r$-minor-free graph $G$, and a number $\e$ in $(0,1)$. Output an $(O(\e), O_r(1/\e))$-shortcut partition of $G$.
\begin{enumerate}
    \item \emph{Construct buffered cop decomposition.}
    
    Let $\Delta \coloneqq \e \cdot \diam(G)$. Let $\cT$ be a \textcolor{BrickRed}{stochastic} $(\textcolor{BrickRed}{4 \cdot} \Delta, \Delta/r, r-1)$-buffered cop decomposition sampled from the distribution of \Cref{thm:random-cop}.
    
    \item \emph{Choose net points.}
    
    For each supernode $\eta$ in $\cT$, let $T_\eta$ be the skeleton of $\eta$. Greedily construct a point set $N_\eta$ to be a $\Delta$-net of $T_\eta$; that is, $N_\eta$ is a subset of vertices in $T_\eta$ such that (1) for all vertices $v$ in $T_\eta$, we have $\dist_{T_\eta}(v, N_\eta) \le \Delta$, and (2) for any two vertices $x_1$ and $x_2$ in $N_\eta$, we have $\dist_{T_\eta}(x_1, x_2) \ge \Delta$.

    \item \emph{Assign the vertices to the closest net point, \textcolor{BrickRed}{after a small random perturbation.}}

    For each supernode $\eta$, for each net point $x$ in $N_\eta$, \textcolor{BrickRed}{choose a random value $\alpha_x \sim \mathrm{Unif}[0, \Delta]$,} and initialize a new cluster $C_x$. We say that $x$ is the \emph{center}
    of $C_x$. 
    For each vertex $v$ in $\eta$, assign $v$ to the cluster $C_x$ whose center $x$ minimizes $\dist(v, x) \textcolor{BrickRed}{+ \alpha_x}$, breaking ties consistently.
\end{enumerate}
\end{tcolorbox}
\caption{The \textsc{Shortcut} algorithm.}
\label{fig:shortcut}
\end{figure}

\begin{claim}
\label{clm:valid-random-shortcut}
    For any $K_r$-minor-free graph $G$ and parameter $\e$ in $(0,1)$, the procedure $\textsc{Shortcut}(G, \e)$ outputs an $(O(\e), O_r(1)/\e)$-shortcut partition for $G$.
\end{claim}
\begin{proof}
    The proof of correctness is almost exactly the same as in Section 4 of \cite{CCLMST24}. 
    There is one change. 
    Lemma~4.2 of \cite{CCLMST24} states that each cluster created by the \textsc{Shortcut} procedure has (strong) diameter at most $4 \Delta$. With our stochastic procedure, this is no longer true; instead, we now show a diameter bound of $12 \Delta$.
    Indeed, let $\eta$ be a supernode and let $v$ be a supernode that is assigned to some cluster $C_x$. Let $P$ be a shortest-path between $v$ and $x$. Because we break ties consistently, every vertex along $P$ is also assigned to the cluster $C_x$. We will show that $\norm{P} \le 6 \Delta$. By the [supernode radius] property, $v$ is within distance $4 \Delta$ of some point $v'$ on the skeleton $T_\eta$. By definition of the net $N_\eta$, $v'$ is within distance $\Delta$ of some point in $N_\eta$; thus, triangle inequality implies that there is some net point $x'$ such that $\dist_\eta(v, x') \le 5 \Delta$. As $\alpha_{x'} \le \Delta$, we have $\dist_\eta(v, x') + \alpha_{x'} \le 6 \Delta$. By choice of $x$, we have $\dist_\eta(v, x) \le 6 \Delta$. We conclude that cluster $C_x$ has strong diameter at most $12 \Delta$.

    The rest of the proof of correctness from Section~4 in \cite{CCLMST24} carries over almost exactly, except that some constant factors increase because we only have a diameter bound of $12 \Delta$ for each cluster. In particular, original value of $9r$ in Claim~4.4 of \cite{CCLMST24} is now replaced with a bound of $25r$; the original bound of $(54r)^k$ in Lemma~4.5 of \cite{CCLMST24} is now replaced with a bound of $(150r)^k$; and finally, this larger constant is absorbed into the Big-O notation of Theorem~1.2 of \cite{CCLMST24}. 
\end{proof}

\begin{claim}
\label{clm:shortcut-probability}
    For any $K_r$-minor-free graph $G$, the stochastic procedure $\textsc{Shortcut}(G, \e)$ produces a clustering $\cC$ such that for any edge $e$ in $G$, $\Pr[\text{$e$ is cut by $\cC$}] \le \beta \cdot \norm{e}/(\e \cdot \diam(G))$ for constant $\beta$.
\end{claim}
\begin{proof}
    Let $\cT$ be the stochastic buffered cop decomposition sampled in Step 1 of the \textsc{Shortcut} procedure. 
    Let $e = (u,v)$ be an edge in $G$. If $e$ is cut by $\cC$, then either $e$ is cut by $\cT$, or both $u$ and $v$ belong to the same supernode $\eta$ but are assigned to different clusters in Step 3 of \textsc{Shortcut}. 
    By \Cref{thm:random-cop}, edge $e$ is cut by $\cT$ with probability at most $\beta' \cdot \norm{e}/(\e \cdot \diam(G))$, for some constant $\beta'$. Thus, it suffices to bound the probability that $e$ is cut by $\cC$, given that $u$ and $v$ belong to the same supernode in $\cT$.

    Let $\eta$ denote the supernode in $\cT$ that contains both $u$ and $v$. 
    By the diameter bound proven above in the proof of Claim~\ref{clm:valid-random-shortcut}, $u$ (resp.\ $v$) is assigned to some cluster $C_x$ with cluster center satisfying $\dist_\eta(u, x) \le 6 \Delta$ (resp.\ $\dist_\eta(v, x) \le 6 \Delta$). 
    We first claim that there are at most $13(r-2) = O(r)$ cluster centers within distance $6 \Delta$ of $u$ (resp.\ $v$), where distance is measured with respect to $\eta$ --- these are the ``threatening'' cluster centers. 
    Suppose otherwise, for contradiction. 
    As the skeleton $T_\eta$ consists of at most $r-2$ shortest paths, pigeonhole principle implies that at least $14$ cluster centers $\set{x_1, \ldots, x_{14}}$ within distance $6 \Delta$ of~$u$ lie on a single shortest path in $T_\eta$. 
    By definition of the net $N_\eta$, the distance between any two cluster center is at least $\Delta$; thus, the distance between furthest two cluster centers in $\set{x_1, \ldots, x_{14}}$ is at least $13 \Delta$. 
    However, triangle inequality implies that all these cluster centers are within distance $12 \Delta$ of each other (as each one is within distance $6 \Delta$ of $v$), a contradiction.

    Now let $C_{x_1}$ and $C_{x_2}$ be two clusters with $C_{x_1} \neq C_{x_2}$, such that the cluster center $x_1$ (resp.\ $x_2$) is within distance $6\Delta$ of $u$ (resp.\ $v$). 
    We claim that $u$ is assigned to $C_1$ and $v$ is assigned to $C_2$ with probability at most $\norm{e}/(\e \cdot \diam(G))$. Indeed, if $u$ is assigned to $C_1$ and $v$ is assigned to $C_2$, we must have
    \[\dist_\eta(u, x_1) + \alpha_{x_1} \in \left( \dist_\eta(v, x_2) + \alpha_{x_2} - \norm{e}, \dist_\eta(v, x_2) + \alpha_{x_2} + \norm{e}\right).\]
    Otherwise, following the argument from Case 3 of Claim~\ref{clm:cop-single-probability}, triangle inequality implies that either $\dist_\eta(v, x_1) + \alpha_{x_1} < \dist_\eta(v, x_2) + \alpha_{x_2}$ and so $v$ is not assigned to $C_2$), or $\dist_\eta(u, x_2) + \alpha_{x_2} < \dist_\eta(u, x_1) + \alpha_{x_1}$ and so $u$ is not assigned to $C_1$. This interval has length $2 \norm{e}$; as $\alpha_{x_1}$ is chosen (independently of $\alpha_{x_2}$) uniformly from an interval of length $\Delta = \e \cdot \diam(G)$, $u$ is assigned to $C_1$ and $v$ is assigned to $C_2$ with probability at most $\norm{e}/(\e \cdot \diam(G))$.
    
    We argued above that there are $O(r)$ possible choices for $C_{x_1}$ and $O(r)$ possible choices for $C_{x_2}$. By applying a union bound over all $O(r^2)$ pairs, we conclude that $e$ is cut by $\cC$ with probability at most $O(r^2) \cdot \norm{e} / (\e \cdot \diam(G))$. This proves the claim.
\end{proof}

Claims~\ref{clm:valid-random-shortcut} and \ref{clm:shortcut-probability} together prove Lemma~\ref{lem:random-shortcut}.


\small
\bibliographystyle{alphaurl}
\bibliography{main}
\end{document}